\documentclass[draftcls,onecolumn]{IEEEtran}

\usepackage{graphicx, color, psfrag, amssymb, amsmath, pstricks, cite, algorithmic}
\usepackage[english]{babel}
\usepackage{cuted}
\usepackage{framed}

\newtheorem{thrm}{Theorem}[section]

\newtheorem{lem}[thrm]{Lemma}

\newtheorem{defn}[thrm]{Definition}
\newtheorem{rem}[thrm]{Remark}

\newtheorem{conj}[thrm]{Conjecture}

\title{Average-Consensus Algorithms in a Deterministic Framework}

\author{\IEEEauthorblockN{Kevin Topley, Vikram Krishnamurthy}\\
\IEEEauthorblockA{Department of Electrical and Computer Engineering \\
The University of British Columbia, Vancouver, Canada \\
Email: \{kevint, vikramk\}@ece.ubc.ca}}
\date{}
\begin{document}
\maketitle


\begin{abstract} We consider the average-consensus problem in a multi-node network of finite size. Communication between nodes is modeled by a sequence of directed signals with arbitrary communication delays. Four distributed algorithms that achieve average-consensus are proposed. Necessary and sufficient communication conditions are given for each algorithm to achieve average-consensus. Resource costs for each algorithm are derived based on the number of scalar values that are required for communication and storage at each node. Numerical examples are provided to illustrate the empirical convergence rate of the four algorithms in comparison with a well-known ``gossip'' algorithm as well as a randomized information spreading algorithm when assuming a fully connected random graph with instantaneous communication. \end{abstract}

\begin{center} \Large Glossary \end{center}

\begin{framed} 
\begin{itemize}
\item BM, Bench-Mark Algorithm, proposed here in $(\ref{bm1})-(\ref{bminit1})$.
\item DA, Distributed-Averaging Algorithm, proposed here in $(\ref{kda1})-(\ref{dainit1})$.
\item OH, One-Hop Algorithm, proposed here in $(\ref{aaasdasig5})-(\ref{ohinit1})$.
\item DDA, Discretized Distributed-Averaging, proposed here in $(\ref{dda1})-(\ref{ddainit1})$.
\item Gossip, Gossip Algorithm, proposed in \cite{SH06}, also defined in $(\ref{gossipsig})-(\ref{gossipinit1})$.
\item RIS, Randomized Information Spreading, proposed in \cite{SH08}.
\item ARIS, Adapted Randomized Information Spreading, defined in $(\ref{arissig})-(\ref{rislast})$.
\item S$\mathcal{V}$SC, a ``singly $\mathcal{V}$-strongly connected'' communication sequence, defined in $(\ref{scdef})$.
\item S$\mathcal{V}$CC, a ``singly $\mathcal{V}$-completely connected'' communication sequence, defined in $(\ref{symcondef2})$.
\item I$\mathcal{V}$SC, an ``infinitely $\mathcal{V}$-strongly connected'' communication sequence, defined in $(\ref{star1ga})$.
\item I$\mathcal{V}$CC, an ``infinitely $\mathcal{V}$-completely connected'' communication sequence, defined in $(\ref{star1g})$.
\end{itemize}
\end{framed}

\section{Introduction}\label{sec:intro} Average-consensus formation implies that a group of distinct nodes come to agree on the average of their initial values, see \cite{SH06}, \cite{co06}, \cite{HA05}, \cite{SH08}, \cite{MU04}, \cite{zh08}, \cite{co04} for related work. Past results indicate that obtaining average-consensus on a set of of $n$ initial vectors each with dimension $d$ requires at least one of the following assumptions:
 \begin{itemize} \item storage and communication of a set with cardinality upper bounded by $O(nd)$ at each node, e.g. the ``flooding'' method described for example in \cite{SB05}, \cite{mu07} (where $O$ denotes the Landau big Oh) 
\item construction of directed acyclical graphs in the network topology \cite{co04}, \cite{ba03}, \cite{mu07} \item knowledge at the transmitting node of the receiving node's identity \cite{mo06}, \cite{co04}, \cite{ba03}, \cite{mu07}
\item instantaneous communication as well as knowledge at the transmitting node of its out-degree \cite{ke03}, \cite{fr09}
\item strictly bi-directional and instantaneous communication \cite{SB05a}, \cite{cy89}, \cite{co06}, \cite{HA05}, \cite{ka09}, \cite{zh08} \item instantaneous communication and symmetric probabilities of node communication \cite{HA05}, \cite{ke03} \item an approximate average-consensus based on randomized number generation \cite{SH08}, \cite{co04}, \cite{ba03} \item pre-determined bounds on the communication delays as well as the use of averaging-weights that are globally balanced and pre-determined 
off-line \cite{zh10}, \cite{li08}, \cite{ta09}, \cite{VK09}, \cite{ja08}, \cite{MU04}. \end{itemize}

This paper proposes four algorithms to solve the average-consensus problem under weaker communication conditions than those listed above. We denote the four algorithms { \em Bench-Mark (BM), Distributed-Averaging (DA), One-Hop (OH)},  and { \em Discretized Distributed-Averaging (DDA) }. \\ (i) The BM algorithm is based on the ``flooding'' method described in \cite{SB05}, \cite{mu07}. We show in Theorems $\ref{bmthm},\ref{bmcor}$ 
that the BM algorithm achieves average-consensus given the weakest communication condition necessary for average-consensus under any distributed algorithm (this is the $\mbox{S$\mathcal{V}$SC}$ condition, defined in Sec.IV). \\ (ii) As the main result, we show in Theorem $\ref{thm6}$ that the DA algorithm (a reduction of the BM algorithm) can achieve average-consensus under a \emph{recurring} $\mbox{S$\mathcal{V}$SC}$ condition (this is the $\mbox{I$\mathcal{V}$SC}$ condition, defined in Sec.IV). By ``recurring'' we mean that for an infinite set of disjoint time intervals, the $\mbox{S$\mathcal{V}$SC}$ condition occurs on each interval. Previous results based on iterative averaging (e.g. \cite{SH06},\cite{HA05},\cite{ne09},\cite{MU04},\cite{ja08}) require special cases of the $\mbox{I$\mathcal{V}$SC}$ condition. \\ (iii) The OH and DDA algorithms can be viewed respectively as simplified versions of the BM and DA algorithms. We will show that analogous results hold under these algorithms with respect to the $\mbox{S$\mathcal{V}$CC}$ and $\mbox{I$\mathcal{V}$CC}$ conditions defined in Sec.IV.

In contrast to earlier work, the main results of this paper show that under general uni-directional connectivity conditions on the communication sequence, each proposed algorithm achieves average-consensus in the presence of \begin{itemize} \item arbitrary communication delays, \item arbitrary link failures. \end{itemize} Another distinct contrast each of the proposed algorithms has with previously considered consensus algorithms is, \begin{itemize} \item each node will know exactly when the true average-consensus estimate has been locally obtained, regardless of the communication pattern between nodes. \end{itemize} Of course our general results come at a price. The main drawback of three of our proposed algorithms (DA, DDA, OH) is that they require local storage and transmission upper bounded by $O(n+d)$, where we recall that $n$ is the network size and $d$ is the dimension of the initial consensus vectors. Most average-consensus algorithms in the literature require only $O(d)$ costs, however they also require assumptions such as instantaneous communication, pre-determined averaging weights, or control at the transmitting node of the receiver identity. Our algorithms (DA, DDA, OH) require none of these assumptions and are particularly advantageous for average-consensus involving $d \geq n$ distinct scalars. In this case the transmission and storage cost of each algorithm is $O(d)$, hence the assumption of relatively weak communication conditions can be leveraged against past algorithms. There are several examples in the literature where $d \geq n$. For instance, if each node observes a $\sqrt{n}$ dimensional process or parameter in noise, then a distributed Kalman filter or maximum-likelihood estimate requires an average-consensus on $d = n + \sqrt{n}$ scalars (see \cite{ol05}, \cite{SB05}).

The first algorithm we consider (BM) is an obvious solution and is presented here only because (i) it serves as a bench-mark for the other algorithms (ii) the communication conditions necessary for its convergence will be used in our main results, and (iii) its formal description has many of the same properties as the other three proposed algorithms. The BM algorithm requires local storage and transmission of $O(nd)$ and has optimal convergence rate (see Theorem $\ref{bmthm}$-$\ref{bmcor}$, in Sec.IV).

\subsection{Relation to Past Work}\label{sec:review}

Since the literature in consensus formation is substantial, we now give an overview of existing results and compare them with the results of this paper. For sufficiently small $\epsilon > 0$, if $x_k = [ x^1_k ,x^2_k , \ldots, x^n_k] \in \mathbb{R}^n$ denotes the network state at time $k$, then \cite{MU04} proves that each element in the sequence $\{ x_k = (I - \epsilon \mathcal{L}_{k}) x_{k-1} \ , \ k = 1,2, \ldots \}$ converges asymptotically to $\sum_{i = 1}^n x^i_0/n$ if each graph Laplacian $\mathcal{L}_k$ in the sequence $\{ \mathcal{L}_k \in \mathbb{R}^{n \times n} \ : \ k = 1, 2, \ldots \}$ is balanced and induces a strongly connected graph. The work \cite{ne09} generalizes this result by allowing $\epsilon$ to decrease at a sufficiently slow rate, and assuming only that there exists some integer $\beta$ such that the union of all graph Laplacians over the interval $[k \beta, (k+1)\beta - 1]$ induces a strongly connected graph for each $k \geq 1$. However, each graph Laplacian in \cite{ne09} is still assumed to be balanced, and as is typical of many interesting papers on average-consensus, neither \cite{MU04} nor \cite{ne09} explain any method by which the nodes can \emph{distributively} ensure each $\mathcal{L}_k$ is balanced \emph{while} the sequence $\{ \mathcal{L}_k \ : \ k = 1, 2, \ldots \}$ is occurring. Hence the results of these works assume that all averaging weights are globally pre-determined off-line, or in other words, every node in the network is assumed to know what averaging weights they should locally be using at each iteration to guarantee the resulting $\mathcal{L}_k$ is globally balanced.

In contrast, \cite{SB05a} proposes a ``Metropolis'' algorithm that requires a ``two-round'' signaling process wherein nodes can \emph{distributively} compute local averaging weights. It is shown in \cite{SB05a} that each element in the sequence $\{ x_k = (I - \epsilon \mathcal{L}_{k}) x_{k-1} \ , \ k = 1,2, \ldots \}$ converges asymptotically to $\sum_{i = 1}^n x^i_0/n$ under mild connectivity conditions on the sequence $\{ \mathcal{L}_k \ : \ k = 1, 2, \ldots \}$. The work \cite{fr09} also proposes a distributed algorithm that does not require pre-determined averaging weights, however \cite{fr09} assumes each transmitting node knows the number of nodes that will receive its message, even before the message is transmitted. Similarly, the algorithm in \cite{SB05a} assumes bi-directional communication, and furthermore each of the stated results in \cite{MU04}, \cite{ne09}, \cite{SB05a}, \cite{fr09} assume that the communication is instantaneous. In contrast, only \cite{MU04}, \cite{fr09} and \cite{SB05a} assume the communication is noiseless.

The results in this paper do not assume instantaneous or bi-directional communication, nor do they assume that the transmitting node knows what node will receive their message or when. However, our results do require noiseless communication. The issue of noisy communication can be treated as a ``meta-problem'' that may be super-imposed upon the framework considered here. Similar to our current approach, there has been much research that assumes noiseless communication (for instance \cite{ke03}, \cite{mo06}, \cite{SH08}). Conversely, there is a growing body of work that considers average-consensus formation in specifically noisy communication settings \cite{ka09}, \cite{ne09}, \cite{ta09}.

Works such as \cite{SH06}, \cite{MO05}, \cite{ts84}, \cite{li08}, \cite{re05}, \cite{bl05}, \cite{st07}, \cite{MU04}, \cite{ja08}, \cite{li08} require node communication properties that are special cases of the $\mbox{I$\mathcal{V}$SC}$ condition defined in Sec.$\ref{sec:nscon}$. However, besides the flooding BM method and the DA algorithm proven to converge here, the only other known algorithm that can even be conjectured to almost surely obtain average-consensus under all $\mbox{I$\mathcal{V}$SC}$ sequences is a specific adaptation of the randomized information spreading (RIS) algorithm proposed in \cite{SH08}. Our adaptation of this algorithm is referred to as ARIS and is detailed in Sec.$\mbox{VII-B}$ of Appendix VII. We note, however, that the lower bound on convergence rate derived in \cite{SH08} assumes instantaneous bi-directional communication, and furthermore, any version of the RIS algorithm assumes that the initial consensus variables are all positive valued. The Gossip algorithm proposed in \cite{SH06} as well as the ARIS algorithm are used as points of reference for the four algorithms proposed in this paper. In Sec.$\ref{sec:rescost}$ we compare the resource costs of all six algorithms, and in Sec.$\ref{sec:figures}$ the performance of each algorithm is illustrated by simulation under various randomized communication sequences when assuming a full network graph.

We note that \cite{MO05}, \cite{ts84}, \cite{re05}, \cite{bl05}, \cite{st07} do not guarantee the final consensus value will equal the initial average; \cite{SH06} proposes an update rule that assumes instantaneous bi-directional communication; and \cite{ja08}, \cite{li08}, \cite{MU04} assume pre-determined bounds on the communication delays as well as averaging-weights that are globally balanced at each iteration (see Theorem 5, \cite{ja08}). In contrast, Theorem $\ref{thm6}$ in Sec.IV states that the DA algorithm will guarantee average-consensus under \emph{all} $\mbox{I$\mathcal{V}$SC}$ sequences, regardless of the communication delays and without requiring any pre-determined balancing of the averaging weights. We note that a distinct feature that all four of the proposed algorithms have is that they each require the local initial consensus vector to be stored in the respective database of each node until average-consensus is locally obtained at that node. Without this property, each algorithm could still ensure a ``consensus formation'' under the exact same communication conditions assumed in our main results, however the final consensus value would not necessarily equal the initial average, which is desirable in most applications \cite{SH06}, \cite{SB05}, \cite{co04}, \cite{ba03}, \cite{mu07}, \cite{ol05}. In further contrast to past results, the proofs of convergence used in this work do not rely on matrix theory \cite{SB05a}, \cite{MU04}, Lyapunov techniques \cite{MO05}, \cite{zh10}, or stochastic stability \cite{HA05}, \cite{ja08}, \cite{VK09}; instead, for our main results we obtain a variety of lower bounds on the ``error'' reduction and show that under (deterministic) recurring connectivity conditions an average-consensus will asymptotically obtain in the $L^2$ norm. As a final note, we clarify that the communication is assumed to be causal, hence a signal cannot be received at node $i$ before it has left node $j$. Given this assumption, our framework considers \emph{every} possible sequence of signals, hence any realization of a (causal) stochastic communication model is a special case of our deterministic framework.

\subsection{Outline}\label{sec:out} Sec.II formulates the problem statement as well as our assumptions regarding the node communication and algorithm framework. Sec.III defines the four proposed algorithms. Sec.IV states the communication conditions that are necessary and sufficient for each algorithm to obtain average-consensus. Sec.V considers the numerical implementation of the four algorithms together with the two comparison algorithms Gossip and ARIS. Resource costs are given in Sec.$\ref{sec:rescost}$, and numerical simulations are presented in Sec.$\ref{sec:figures}$. A summary of the results and suggested future work is provided Sec.VI. The Appendix VII presents the proofs of all theorems, the two comparison algorithms, four conjectures, the resource cost derivations, and an important example.

\section{The Average-Consensus Problem and Algorithm Assumptions}\label{sec:prob} 

This section formulates the average-consensus problem and lists the assumptions regarding communication between nodes. Sec.$\ref{sec:probstat}$ below defines the graph theoretic model for consensus formation that will be subsequently analyzed, and Sec.$\ref{sec:commod1}$ defines the class of the distributed algorithms we consider. Sec.$\ref{sec:commassump}$ details the remaining assumptions that will be made on the node communication and computational abilities, and also explains the technique by which the four proposed algorithms will obtain average-consensus.

\subsection{Problem Formulation}\label{sec:probstat} Let $t \geq 0$ denote time (the results of this paper assume a continuous time framework; any discrete time index is a special case of this framework). At initial time $t = 0$, consider a finite set of arbitrarily numbered nodes $\mathcal{V} = \{ 1 , \ldots, n \}$ and a set of $d$-dimensional vectors $\{s_i(0) \in \mathbb{R}^d \ : \ i \in \mathcal{V} \}$. The set $\{s_i(0) \in \mathbb{R}^d \ : \ i \in \mathcal{V} \}$ is referred to as the set of ``initial consensus vectors''. Suppose each node can locally store a ``knowledge set'' $\mathcal{K}_i(t)$ that consists of a group of scalars each with a distinct meaning. \begin{itemize} \item (A1): (knowledge set assumption) At any time $t \geq 0$, each node $i \in \mathcal{V}$ is equipped with a device that can update and store a ``knowledge set'' $\mathcal{K}_i(t)$. For each $i \in \mathcal{V}$, the knowledge set $\mathcal{K}_i(t)$ may have a time-varying cardinality. \end{itemize} Next we assume that a set $S^{i j} ( t^{i j}_0, t^{i j}_1 )$ can be transmitted from node $j$ at a time denoted $t^{i j}_0 \geq 0$, and received at node $i$ at a time denoted as $t^{i j}_1$, where due to causality $t^{i j}_1 \geq t^{i j}_0$. We refer to $S^{i j} ( t^{i j}_0, t^{i j}_1 )$ as a ``signal'', or ``signal set''. \begin{itemize} \item (A2): (signal set assumption) At any time $t^{i j}_0 \geq 0$, each node $j \in \mathcal{V}$ has the ability to transmit a ``signal set'' $S^{i j} ( t^{i j}_0, t^{i j}_1 ) \subseteq \mathcal{K}_j(t^{i j}_0)$ that will be received at some node $i \in \mathcal{V}$ at time $t^{i j}_1 \geq t^{i j}_0$. \end{itemize} As our final condition, we assume that at $t = 0$ each node $i \in \mathcal{V}$ will ``know'' its unique node identifier value $i$, the network size $n$, and the respective initial consensus vector $s_i(0)$. This is formalized as, \begin{itemize}  \item (A3): At time $t = 0$, the knowledge set $\mathcal{K}_i(0)$ of each node $i \in \mathcal{V}$ satisfies $\mathcal{K}_i(0) \supseteq \{ i, n , s_i(0) \}$. \end{itemize} \begin{defn}\label{def1} Under (A1)-(A3), the average-consensus problem is solved at some time instant $t$ if and only if (iff) the \emph{average} of the initial consensus vectors, \begin{equation}\label{avg0} \bar{s}(0) = \frac{1}{n} \sum_{i = 1}^n s_i(0) \ , \end{equation} is contained in the knowledge set $\mathcal{K}_i(t)$ of all nodes $i \in \mathcal{V}$. We say that a specific node $i$ has obtained average-consensus at time $t$ iff $\bar{s}(0) \in \mathcal{K}_i(t)$. $\quad \quad \quad \quad \quad \quad \quad \quad \quad \quad \ \ \ \ \ \ \ \ \ \ \ \ \ \quad \quad \ \ \ \ \ \ \ \ \ \quad \quad \quad \quad  \ \ \ \ \ \ \ \ \quad \quad  \ \ \ \ \ \ \quad \quad \ \Box$\end{defn}

The four algorithms analyzed in this paper can be adapted so that the size of the network $n$ and unique node identifier $i \in \mathcal{V}$ can be computed during the averaging process, and thus they need not be initially known at any node. For simplicity, however, it will be assumed that the network size $n$ and unique node identifier $i$ are initially known at each node $i \in \mathcal{V}$.

\subsection{The Distributed Algorithm and Local Consensus Estimates}\label{sec:commod1} To provide a unifying theme for the algorithms discussed in this paper, we first define the class of distributed algorithms that will be considered for consensus formation. Given the assumptions $(A1)-(A2)$, we define a ``distributed algorithm'' in terms of its ``knowledge set updating rule'' $f_{\mathcal{K}} \{ \cdot \}$ together with a ``signal specification rule'' $f_S \{ \cdot \}$. The knowledge set updating rule $f_{\mathcal{K}}\{\cdot\}$ defines the effect that a set of signals has on the knowledge set $\mathcal{K}_i( t^{ij}_1 )$ at the receiving node $i$, whereas the signal specification rule $f_S \{\cdot\}$ defines the elements contained in the signal $S^{ij}(t^{ij}_0, t^{ij}_1 )$ given as a function of the knowledge set $\mathcal{K}_j(t^{ij}_0)$ at the transmitting node $j$. If we assume that upon reception of a set of signals $\Delta t > 0$ time elapses before the knowledge set is fully updated, then the class of distributed algorithms we consider may be defined as follows.

\begin{framed} $$ \text{Class of Distributed Algorithms under (A1)-(A2):} $$
\begin{equation}\label{update10} \mbox{Knowledge Set Updating Rule:} \ \ \ \ \ \ \ \ \ \ f_\mathcal{K} \ : \  \mathcal{K}_i(t^{ij}_1) \bigcup_{t^{ij}_0 \leq t^{ij}_1 \ : \ j \in \mathcal{V}} S^{ij}(t^{ij}_0,t^{ij}_1) \ \longmapsto  \ \mathcal{K}_i(t^{ij}_1 + \Delta t) \ \ \ \ \ \ \ \ \ \ \ \ \end{equation}
\begin{equation}\label{update0} \mbox{Signal Specification Rule:} \ \ \ \ \ \ \ \ \ \ \ \ \ \ \ \ \ \ \ \ f_S \ : \ \mathcal{K}_j(t^{ i j}_0 ) \  \longmapsto S^{ij}(t^{ij}_0, t^{ij}_1 )  \ \ \ \ \ \ \ \ \ \ \ \ \ \ \ \ \ \ \ \ \ \ \ \ \ \ \ \ \ \ \ \ \ \ \ \ \ \ \ \ \ \ \ \ \ \ \ \ \ \ \ \ \ \  \end{equation} \end{framed}

The algorithms in this paper will assume that every knowledge set $\mathcal{K}_i(t)$ contains a local ``consensus estimate'' $\hat{s}_i(t) \in \mathbb{R}^d$ that represents the ``belief'' of node $i$ in regard to the average-consensus $\bar{s}(0)$ defined in $(\ref{avg0})$, \begin{equation}\label{a0a} \mathcal{K}_i(t) \supseteq \{ \hat{s}_i(t) \} \ , \ \forall \ i \in \mathcal{V} \ , \ \forall \ t \geq 0 \ . \end{equation} Notice that using a local consensus estimate $\hat{s}_i(t)$ is necessary for any algorithm seeking to solve the average-consensus problem; if there is no local consensus estimate then there is no means by which the knowledge set of any node can contain $\bar{s}(0)$ and hence the problem stated in Definition $\ref{def1}$ is ill-posed.

In contrast from any known past consensus algorithm, the four proposed algorithms here will also assume that at any time $t \geq 0$ the knowledge set $\mathcal{K}_i(t)$ of each node $i \in \mathcal{V}$ will contain a local ``normal consensus estimate'' $\hat{v}_i(t) \in \mathbb{R}^n$, \begin{equation}\label{a0b} \mathcal{K}_i(t) \supseteq \{ \hat{v}_i(t) \} \ , \ \forall \ i \in \mathcal{V} \ , \ \forall \ t \geq 0 \ . \end{equation} Combining $(A3)$, $(\ref{a0a})$, and $(\ref{a0b})$ we then have the initial knowledge set for each of the four proposed algorithms, \begin{itemize} \item $(A4)$: $\ \ \ \mathcal{K}_i(0) = \{ i, n, s_i(0) , \hat{s}_i(0) , \hat{v}_i(0) \} \ , \ \ \forall \ i \in \mathcal{V} \ .$ \end{itemize}

Define the ``error'' of the consensus estimate $\hat{s}_i(t)$ as follows, $$ E_{\bar{s}(0)} \big( \hat{s}_i(t) \big) = | \sqrt{ \big( \hat{s}_i(t) - \bar{s}(0) \big)^2 } | \ , $$ where $\bar{s}(0)$ is defined in $(\ref{avg0})$. Denote the ``network consensus error'' $\sum_{i=1}^n E_{\bar{s}(0)} \big( \hat{s}_i(t) \big)$. We conclude this section with the following definition of what constitutes the solution to the average-consensus problem.

\begin{defn}\label{def10} The average-consensus problem is solved at time $t$ iff the network consensus error is zero, that is, \begin{equation}\label{acdef} \sum_{i=1}^n E_{\bar{s}(0)} \big( \hat{s}_i(t) \big) = 0 \ . \end{equation} $ \ \ \ \ \ \ \ \quad \quad  \quad \quad  \quad \quad  \ \ \ \ \ \ \ \ \ \quad \quad \quad \quad  \ \ \ \ \ \ \ \ \quad \quad  \ \ \ \ \ \ \quad \quad \ \ \ \ \ \ \ \ \quad \quad  \quad \quad  \quad \quad  \ \ \ \ \ \ \ \ \ \quad \quad \quad \quad  \ \ \ \ \ \ \ \ \quad \quad  \ \ \ \ \ \ \quad \quad \ \Box$  \end{defn}

\subsection{Node Communication and Update Assumptions}\label{sec:commassump} Below we detail the node communication and update assumptions that will be used by each of the four proposed algorithms. The final update condition we propose will explain the technique by which all four algorithms achieve average-consensus.

For any $t \geq 0$, define $t(+)$ as the right-hand limit of $t$, that is $t(+) = \lim_{ T \rightarrow t^+} T$. To construct a suitable average-consensus algorithm we assume the following conditions on the node communication and knowledge set updates: \begin{itemize} \item (A5): at no time does any node $j \in \mathcal{V}$ have the ability to know \emph{when} the signal $S^{ij}(t^{ij}_0,t^{ij}_1)$ is transmitted, \emph{when} the signal $S^{ij}(t^{ij}_0,t^{ij}_1)$ is received, or \emph{what} node $i \in \mathcal{V}$ will receive the signal $S^{ij}(t^{ij}_0,t^{ij}_1)$. \item (A6): at no time does any node $i \in \mathcal{V}$ have the ability to control \emph{when} the signal $S^{ij}(t^{ij}_0,t^{ij}_1)$ is received, \emph{what} node $j \in \mathcal{V}$ transmitted the signal $S^{ij}(t^{ij}_0,t^{ij}_1)$, or \emph{when} the signal $S^{ij}(t^{ij}_0,t^{ij}_1)$ was transmitted. \item (A7): each knowledge set satisfies $\mathcal{K}_i(t(+)) = \mathcal{K}_i(t)$ at any time $t \geq 0$ that node $i$ does not receive a signal (recall that $t(+)$ denotes the right-hand limit of $t$). \item (A8): when a signal $S^{ij}(t^{ij}_0,t^{ij}_1)$ is received, the knowledge set $\mathcal{K}_i(t^{ij}_1)$ of the receiving node is updated in an arbitrarily small amount of time. \item (A9): at most one signal can be received and processed by a given node at any given time instant. \end{itemize} Note that (A5)-(A6) imply the communication process is \emph{a priori} unknown at every node. The assumption (A7) implies that the node knowledge set $\mathcal{K}_i(t)$ can only change if a signal is received at node $i$, and (A8) implies that any signal $S^{ij}(t^{ij}_0,t^{ij}_1)$ transmitted from node $j$ is allowed to contain information that has been updated by a signal received at node $j$ at any time preceding $t^{ij}_0$. Notice that (A8) is realistic since all four of the proposed algorithms will require only a few arithmetic operations in the update process. Assumption (A9) is a technical requirement that simplifies the proofs of convergence. Together (A8) and (A9) imply the knowledge set updating rule $f_\mathcal{K} \{ \cdot \}$ defined in $(\ref{update10})$ reduces to, \begin{equation}\label{update1} f_\mathcal{K} \ : \  \mathcal{K}_i(t^{ij}_1) \cup S^{ij}(t^{ij}_0,t^{ij}_1) \ \longmapsto  \ \mathcal{K}_i(t^{ij}_1 (+)) \ . \end{equation} Each algorithm we propose can be easily adapted if (A8)-(A9) were relaxed, however this would be at the expense of simplicity in both our framework and analysis.

As our final update condition, we require that the local consensus estimate $\hat{s}_i(t^{ij}_1)$ at the receiving node $i$, as defined above $(\ref{a0a})$, is updated based on the updated normal consensus estimate $\hat{v}_i(t^{ij}_1(+))$ via the relation, \begin{equation}\label{a10} \hat{s}_i(t^{ij}_1(+)) = \mathbf{S} \hat{v}_i(t^{ij}_1(+)) \ , \end{equation} where $\mathbf{S} = [ s_1(0) ; s_2(0) ; \cdots ; s_n(0) ] \in \mathbb{R}^{d \times n}$. Under $(\ref{a10})$ it is clear that $\hat{s}_i(t^{ij}_1(+)) = \bar{s}(0)$ if $\hat{v}_i(t^{ij}_1(+)) = \frac{1}{n} \mathbf{1}_n$, where $\mathbf{1}_n \in \mathbb{R}^n$ denotes a vector that consists only of unit-valued elements. Thus, in terms of Definition $\ref{def10}$, under $(\ref{a10})$ and $(\ref{a0b})$ the average-consensus problem is solved at time $t$ if $\hat{v}_i(t) = \frac{1}{n} \mathbf{1}_n$ for all nodes $i \in \mathcal{V}$. Motivated by this fact, we propose for each of the four algorithms that the normal consensus estimate $\hat{v}_i(t)$ is updated based on the following optimization problem, \begin{equation}\label{optllca}\begin{array}{llll} \hat{v}_i ( t^{ij}_1(+)) = & \arg_{\mbox{$\tilde{v}$}} \ \min \ \big( \tilde{v} - \frac{1}{n} \mathbf{1}_n \big) ^2 , \\ & \ \mbox{s.t. $(\ref{a10})$ holds, given $\mathcal{K}_i(t^{ij}_1) \bigcup S^{ij}(t^{ij}_0,t^{ij}_1 )$.} \end{array}\end{equation} Note that from $(\ref{a0a})$ and $(\ref{a0b})$ each node $i \in \mathcal{V}$ will know it has obtained the true average-consensus value $\hat{s}_i(t) = \bar{s}(0)$ when the local normal consensus estimate $\hat{v}_i(t)$ satisfies the condition $\hat{v}_i(t) = \frac{1}{n} \mathbf{1}_n$. In the next section we will define the knowledge set updating rule $f_{\mathcal{K}} \{ \cdot \}$ and signal specification rule $f_S \{ \cdot \}$ for each of the four algorithms. We shall find that for each algorithm the update problem $(\ref{optllca})$ reduces to a least-squares optimization and a closed-form expression can be obtained for updated normal consensus estimate $\hat{v}_i(t^{ij}_1(+))$.

\section{Distributed Average-Consensus Algorithms}\label{sec:alg} With the above definitions, we are now ready to describe the four distributed algorithms that achieve average-consensus. The details of the two comparison algorithms can be found in Sec.$\mbox{VII-B}$ of Appendix VII. All six algorithms (the four presented below and the two comparison algorithms in Sec.$\mbox{VII-B}$) are defined using the abstract definition of the ``distributed algorithm'' given in Sec.II, that is $(\ref{update0}),(\ref{update1})$. This section sets the stage for the convergence theorems presented in Sec.IV.

\subsection{Algorithm 1: Bench-Mark (BM)}\label{sec:ben} The BM algorithm obtains average-consensus trivially; we propose it formally because the remaining three algorithms are specific reductions of it, and also because it requires communication conditions that will be used in our main result. The BM algorithm implies that the initial consensus vectors $\{ s_i(0) \ : \ i \in \mathcal{V} \}$ are essentially flooded through-out the network. In \cite{SB05}, \cite{mu07} the general methodology of the BM algorithm is discussed wherein it is referred to as ``flooding''. For completeness of our results, we show in Theorems $\ref{bmthm}$,$\ref{bmcor}$ that regardless of the communication pattern between nodes, there exists no distributed algorithm $(\ref{update10}),(\ref{update0})$ that can obtain average-consensus before the BM algorithm. This is why we have named it the ``bench-mark'' algorithm.

Let $\delta [ \cdot ]$ denote the Kronecker delta function applied element-wise, and $e_i$ denote the $i^{th}$ standard unit vector in $\mathbb{R}^n$. The BM signal specification $(\ref{update0})$ and knowledge set update $(\ref{update1})$ are respectively defined as $(\ref{bm1})$ and $(\ref{bm2})$ below.   \begin{framed} $$ \mbox{Algorithm 1: Bench-Mark (BM)} $$ \begin{equation}\label{bm1} \mbox{Signal Specification: } \ \ S^{ij}(t^{ij}_0,t^{ij}_1 ) = \left \{ \begin{array}{ll} \mathcal{K}_j(t^{ij}_0) \setminus \{ j, n , \hat{s}_j(t^{ij}_0 ) \} & \mbox{if $\hat{v}_j(t^{ij}_0) \neq \frac{1}{n} \mathbf{1}_n$} \\ \{ \hat{v}_j(t^{ij}_0) , \hat{s}_j(t^{ij}_0) \} & \mbox{if $\hat{v}_j(t^{ij}_0) = \frac{1}{n} \mathbf{1}_n$} \end{array} \right.  \ \ \ \ \ \ \ \ \ \end{equation} $$ \text{Knowledge Set Update:} \ \ \ \ \ \ \ \ \ \ \ \ \ \ \ \ \ \ \ \ \ \ \ \ \ \ \ \ \ \ \ \ \ \ \ \ \ \ \ \ \ \ \ \ \ \ \ \ \ \ \ \ \ \ \ \ \ \ \ \ \ \ \ \ \ \ \ \ \ \ \ \ \ \ \ \ \ \ \ \ \ \ $$ $$ v^{ij} ( t^{ij}_1 , t^{ij}_0) \equiv \mathbf{1}_n - \delta [ \hat{v}_i( t^{ij}_1) + \hat{v}_j( t^{ij}_0) ] $$ \begin{equation}\label{bm2} \mathcal{K}_i(t^{ij}_1(+)) = \left \{ \begin{array}{ll} \{ i, n , \hat{v}_i(t^{ij}_1(+)) , \hat{s}_i(t^{ij}_1(+)) , s_\ell( 0) \ \forall \ \ell \ \mbox{s.t. } v^{ij}_\ell ( t^{ij}_1 , t^{ij}_0) = 1 \} \ \ \mbox{if $\hat{v}_i(t^{ij}_1(+)) \neq \frac{1}{n} \mathbf{1}_n$} \\ \{ \hat{v}_i(t^{ij}_1(+)) , \hat{s}_i(t^{ij}_1(+)) \} \ \ \ \ \ \ \ \ \ \ \ \ \ \ \ \ \ \ \ \ \ \ \ \ \ \ \ \ \ \ \ \ \ \ \ \ \ \ \ \ \ \ \ \mbox{if $\hat{v}_i(t^{ij}_1(+)) = \frac{1}{n} \mathbf{1}_n$} \end{array} \right.\end{equation}  \begin{equation}\label{bm3} \mbox{Normal Consensus Estimate Update: } \ \ \hat{v}_i (t^{ij}_1(+)) = \frac{1}{n} v^{ij} ( t^{ij}_1 , t^{ij}_0 )  \ \ \ \ \ \ \ \ \ \ \ \ \ \ \ \ \ \ \ \ \ \ \ \ \ \ \ \ \ \ \ \ \end{equation}  \begin{equation}\label{bmconsest} \mbox{Consensus Estimate Update: } \hat{s}_i(t^{ij}_1(+)) = \left\{ \begin{array}{l l} \sum_{\ell = 1}^n \hat{v}_{i \ell} (t^{ij}_1(+)) s_\ell (0) & \mbox{if $\hat{v}_i(t^{ij}_1(+)) \neq \frac{1}{n} \mathbf{1}_n$} \\ \bar{s}(0) & \mbox{if $\hat{v}_i(t^{ij}_1(+)) = \frac{1}{n} \mathbf{1}_n$} \end{array} \right. \end{equation}   \begin{equation}\label{bminit1} \mbox{Estimate Initialization: } \ \ \hat{v}_i(0) = \frac{1}{n} e_i \ , \ \hat{s}_i(0) = \frac{1}{n} s_i(0) \ . \ \ \ \ \ \ \ \ \ \ \ \ \ \ \ \ \ \ \ \ \ \ \ \ \ \ \ \ \ \ \ \ \ \ \ \ \ \ \ \ \ \end{equation} \end{framed} In Lemma $\ref{lembm}$ and Lemma $\ref{lembm0}$ of Appendix VII we will prove that $(\ref{bm3})$ is the unique solution to $(\ref{optllca})$ under $(\ref{bm1})$ and $(\ref{bm2})$, and that $(\ref{bminit1})$ is the unique solution to $(\ref{optllca})$ under the initial knowledge set $(A4)$. The update $(\ref{bmconsest})$ follows immediately from $(\ref{bm3})$ and the relation $(\ref{a10})$. Notice that the BM algorithm updates $(\ref{bm2})$, $(\ref{bm3})$ together with the signal specification $(\ref{bm1})$ imply $\bar{s}(0) \in \mathcal{K}_i(t^{ij}_1(t^{ij}_1(+))$ iff $\hat{v}_i (t^{ij}_1(+)) = \frac{1}{n} \mathbf{1}_n$, and likewise $s_\ell(0) \in \mathcal{K}_i(t^{ij}_1(+))$ iff $\hat{v}_{i \ell}(t^{ij}_1(+)) = \frac{1}{n}$ and $\hat{v}_i(t^{ij}_1(+)) \neq \frac{1}{n} \mathbf{1}_n$. It thus follows that $(\ref{bm1})$, $(\ref{bm2})$ and $(\ref{bm3})$ imply that the consensus estimate update $\hat{s}_i(t^{ij}_1(+))$ defined in $(\ref{bmconsest})$ can be locally computed at node $i$ based only on $\mathcal{K}_i(t^{ij}_1)$ and the received signal $S^{ij}( t^{ij}_0, t^{ij}_1)$.

Besides flooding the initial consensus vectors, the BM algorithm $(\ref{bm1})-(\ref{bminit1})$ has an additional feature that is not necessary but is rather practical: if $\hat{v}_j(t) = \frac{1}{n} \mathbf{1}_n$ then all $n$ of the initial consensus vectors are stored in $\mathcal{K}_j(t)$, when this occurs the consensus estimate $\hat{s}_j(t)$ defined in $(\ref{bmconsest})$ will equal $\bar{s}(0)$ and thus node $j$ has obtained average-consensus. In this case all of the initial consensus vectors are deleted from $\mathcal{K}_j(t)$, and any signal transmitted from $j$ contains only the consensus estimate $\hat{s}_j(t) = \bar{s}(0)$ and the vector $\hat{v}_j(t) = \frac{1}{n} \mathbf{1}_n$. Upon reception of a signal containing $\hat{v}_j(t) = \frac{1}{n} \mathbf{1}_n$, the receiving node $i$ can delete all of their locally stored initial consensus vectors and set $\hat{s}_i(t) = \hat{s}_j(t) = \bar{s}(0)$. In this way the average-consensus value $\bar{s}(0)$ can be propagated through-out the network without requiring all $n$ initial consensus vectors to be contained in every signal. See Conjecture $\ref{conject0}$ in Sec.$\ref{sec:app1a}$ for a conjecture regarding the BM algorithm.

\subsection{Algorithm 2: Distributed Averaging (DA)}\label{sec:kda} The DA algorithm that we now introduce, to the best of our knowledge, is new. To define the DA update procedure, let $V^+$ denote the pseudo-inverse of an arbitrary matrix V. The DA signal specification $(\ref{update0})$ and knowledge set update $(\ref{update1})$ are defined as $(\ref{kda1})$ and $(\ref{kda2})$ below. \begin{framed} $$ \mbox{Algorithm 2: Distributed Averaging (DA)} $$  \begin{equation}\label{kda1}\mbox{Signal Specification: } \ \ \ \ \ S^{ij}(t^{ij}_0,t^{ij}_1 ) = \mathcal{K}_j(t^{ij}_0) \setminus \{ j, n , s_j(0 ) \} \ \ \ \ \ \ \ \ \ \ \ \ \ \ \ \ \ \ \ \ \ \ \ \ \ \ \ \ \ \ \ \ \ \ \ \ \ \ \ \ \ \ \ \ \ \end{equation} \begin{equation}\label{kda2} \mbox{Knowledge Set Update: } \ \  \mathcal{K}_i(t^{ij}_1(+)) = \{ i, n , \hat{v}_i(t^{ij}_1(+)) , \hat{s}_i(t^{ij}_1(+)) , s_i(0) \} \ \ \ \ \ \ \ \ \ \ \ \ \ \ \ \ \ \ \ \ \ \ \ \ \ \ \ \ \ \ \end{equation} \begin{equation}\label{optkda2}  \mbox{Normal Consensus Estimate Update: }  \hat{v}_i ( t^{ij}_1(+)) = V_{(DA)} V_{(DA)}^+ \frac{1}{n} \mathbf{1}_n \ , \ V_{(DA)} = [ \hat{v}_i(t^{ij}_1) , \hat{v}_j(t^{ij}_0), \frac{1}{n} e_i ] \end{equation}  \begin{equation}\label{rr1a}\mbox{Consensus Estimate Update: } \ \ \hat{s}_i(t^{ij}_1(+)) = V_s V_{(DA)}^{+} \frac{1}{n} \mathbf{1}_{n} \ , \ V_s = [ \hat{s}_i( t^{ij}_1 ), \hat{s}_j( t^{ij}_0 ) , \frac{1}{n} s_i(0) ] \ \ \ \ \ \ \ \ \ \ \ \ \ \end{equation}  \begin{equation}\label{dainit1} \mbox{Estimate Initialization: } \ \ \ \ \ \hat{v}_i(0) = \frac{1}{n} e_i \ , \ \hat{s}_i(0) = \frac{1}{n} s_i(0) \ . \ \ \ \ \ \ \ \ \ \ \ \ \ \ \ \ \ \ \ \ \ \ \ \ \ \ \ \ \ \ \ \ \ \ \ \ \ \ \ \ \ \ \ \ \ \ \ \end{equation} \end{framed} The Lemma $\ref{lembm0}$ and Lemma $\ref{lemkda}$ in Appendix VII, respectively, prove that $(\ref{dainit1})$ is the unique solution of $(\ref{optllca})$ under $(A4)$, and that $(\ref{optkda2})$ is the unique solution to $(\ref{optllca})$ under $(\ref{kda1})$ and $(\ref{kda2})$. Based on the normal consensus update $(\ref{optkda2})$ the relation $(\ref{a10})$ reduces to $(\ref{rr1a})$. From $(\ref{rr1a})$ it is clear that the DA consensus estimate update $\hat{s}_i(t^{ij}_1(+))$ can be locally computed at node $i$ based only on $\mathcal{K}_i(t^{ij}_1)$ and the received signal $S^{ij}( t^{ij}_0, t^{ij}_1)$. Notice that $V_{(DA)}$ is a $n \times 3$ matrix, thus the pseudo-inverse in $(\ref{optkda2})$ will have an immediate closed-form expression, see $(\ref{fin1})$, $(\ref{matinv1})$ and $(\ref{pp1condproof})$ in Lemma $\ref{lemkda3}$. Also note that under the DA algorithm every signal contains only the local consensus estimate $\hat{s}_j(t)$ together with the local normal consensus estimate $\hat{v}_j(t)$.

\subsection{Algorithm 3: One-Hop (OH)}\label{aaasec:sda2} Under the OH algorithm each signal $S^{ij}(t^{ij}_0,t^{ij}_1 )$ will either contain the local initial consensus vector $s_j(0)$ and transmitting node identity $j$, or the average-consensus value $\bar{s}(0)$ and a scalar $0$ to indicate that the transmitted vector is the true average-consensus value. For this reason the conditions for average-consensus under the OH algorithm are relatively straight-forward to derive (see Theorem $\ref{aaacor5}$ in Sec.IV and the proof in Appendix VII). The OH algorithm signal specification and knowledge set update are respectively defined by $(\ref{aaasdasig5})$ and $(\ref{aaais2})$ below. \begin{framed} $$ \mbox{Algorithm 3: One-Hop (OH)} $$  \begin{equation}\label{aaasdasig5} \mbox{Signal Specification: } \ \ S^{ij}(t^{ij}_0,t^{ij}_1 ) = \left\{ \begin{array}{l l} \{ j , s_j(0) \} \ \ \ \ \ \mbox{if $\hat{v}_j(t^{ij}_0) \neq \frac{1}{n} \mathbf{1}_n$} \\ \{ 0 , \hat{s}_j(t^{ij}_0) \} \ \ \ \mbox{if $\hat{v}_j(t^{ij}_0) = \frac{1}{n} \mathbf{1}_n$} \end{array} \right. \ \ \ \ \ \ \ \ \ \ \ \ \ \ \ \ \ \ \ \ \ \ \ \ \ \ \ \ \ \ \ \ \ \end{equation}   \begin{equation}\label{aaais2}\mbox{Knowledge Set Update: } \mathcal{K}_i(t^{ij}_1(+)) = \left \{ \begin{array}{ll} \{ i, n , \hat{v}_i(t^{ij}_1(+)) , \hat{s}_i(t^{ij}_1(+)) , s_i(0) \} \ \ \ \mbox{if $\hat{v}_i(t^{ij}_1(+)) \neq \frac{1}{n} \mathbf{1}_n$} \\ \{ \hat{v}_i(t^{ij}_1(+)), \hat{s}_i(t^{ij}_1(+)) \} \ \ \ \ \ \ \ \ \ \ \ \ \ \ \ \mbox{if $\hat{v}_i(t^{ij}_1(+)) = \frac{1}{n} \mathbf{1}_n$} \end{array} \right.\end{equation} $$ \text{Normal Consensus Estimate Update: } \ \ \ \ \ \ \ \ \ \ \ \ \ \ \ \ \ \ \ \ \ \ \ \ \ \ \ \ \ \ \ \ \ \ \ \ \ \ \ \ \ \ \ \ \ \ \ \ \ \ \ \ \ \ \ \ \ \ \ \ \ \ \ \ \ \ \ \ \ \ \ \ \ \ \ \ \ \ \ \ $$ $$ v^{ij} ( t^{ij}_1 , t^{ij}_0) \equiv \mathbf{1}_n - \delta [ \hat{v}_i( t^{ij}_1) + e_j ] $$ \begin{equation}\label{aaaisvhat} \hat{v}_i( t^{ij}_1 (+) ) = \left\{ \begin{array}{l l} \frac{1}{n} v^{ij} ( t^{ij}_1 , t^{ij}_0 ) \ \ \ \ \mbox{if $0 \neq S_1^{ij}(t^{ij}_0,t^{ij}_1 )$} \\ \frac{1}{n} \mathbf{1}_n \ \ \ \ \ \ \ \ \ \ \ \ \ \ \mbox{if $0 = S^{ij}_1(t^{ij}_0,t^{ij}_1 )$} \end{array} \right. \end{equation}  \begin{equation}\label{aaaredsda2} \mbox{Consensus Estimate Update: } \ \ \ \hat{s}_i( t^{ij}_1 (+) ) = \left\{ \begin{array}{l l} \hat{s}_i ( t^{ij}_1 ) + \big( \frac{1}{n} - \hat{v}_{ij} ( t^{ij}_1 ) \big) s_j(0) \ \ \ \ \mbox{if $0 \neq S^{ij}_1(t^{ij}_0,t^{ij}_1 )$} \\ \bar{s}(0) \ \ \ \ \ \ \ \ \ \ \ \ \ \ \ \ \ \ \ \ \ \ \ \ \ \ \ \ \ \ \ \ \ \mbox{if $0 = S^{ij}_1(t^{ij}_0,t^{ij}_1 )$} \end{array} \right. \end{equation}  \begin{equation}\label{ohinit1} \mbox{Estimate Initialization: } \ \ \ \ \ \ \ \ \ \hat{v}_i(0) = \frac{1}{n} e_i \ , \ \hat{s}_i(0) = \frac{1}{n} s_i(0) \ . \ \ \ \ \ \ \ \ \ \ \ \ \ \ \ \ \ \ \ \ \ \ \ \ \ \ \ \ \ \ \ \ \ \ \ \ \ \ \ \ \ \ \ \ \end{equation} \end{framed} The Lemma $\ref{lembm0}$ and Lemma $\ref{aaalemsda2}$ in Appendix VII respectively prove that $(\ref{ohinit1})$ is the unique solution to $(\ref{optllca})$ under $(A4)$, and that $(\ref{aaaisvhat})$ is the unique solution to $(\ref{optllca})$ under $(\ref{aaasdasig5})$ and $(\ref{aaais2})$. Notice that $(\ref{aaaisvhat})$ implies the relation $(\ref{a10})$ reduces to $(\ref{aaaredsda2})$. Given $(\ref{aaasdasig5})$, $(\ref{aaais2})$ and $(\ref{aaaisvhat})$ it follows that the OH consensus estimate update $\hat{s}_i(t^{ij}_1(+))$ defined in $(\ref{aaaredsda2})$ can be locally computed at node $i$ based only on $\mathcal{K}_i(t^{ij}_1)$ and the received signal $S^{ij}( t^{ij}_0, t^{ij}_1)$.

\subsection{Algorithm 4: Discretized Distributed-Averaging (DDA)}\label{aaasec:kda1} For the DDA algorithm let the discrete set of vectors $\mathbb{R}^n_{0,\frac{1}{n}}$ be defined, \begin{equation}\label{aaaaa2} \mathbb{R}^n_{0,\frac{1}{n}} = \{ v \in \mathbb{R}^n \ : \ v_\ell \in \{ 0, 1/n \} \ , \ \forall \ \ell = 1, 2, \ldots, n \ \} . \end{equation} The discretized version of $(\ref{optllca})$ that we consider under the DDA algorithm is, \begin{equation}\label{aaaoptdisc}\begin{array}{llll} \hat{v}_i ( t^{ij}_1(+)) = & \arg_{\mbox{$\tilde{v}$}} \ \min \big( \tilde{v} - \frac{1}{n} \mathbf{1}_n \big) ^2 , \\ & \mbox{s.t. $(\ref{a10})$ holds and $\tilde{v} \in \mathbb{R}^n_{0,\frac{1}{n}}$, given $\mathcal{K}_i(t^{ij}_1) \bigcup S^{ij}(t^{ij}_0,t^{ij}_1 )$.} \end{array}\end{equation} To define the DDA normal consensus update it is convenient to denote $v^{-i} \in \mathbb{R}^{n-1}$ as the vector $v \in \mathbb{R}^n$ with element $v_i$ deleted. The DDA signal specification and knowledge set update are defined below, \begin{framed}$$ \mbox{Algorithm 4: Discretized Distributed-Averaging (DDA)} $$   \begin{equation}\label{dda1}\mbox{Signal Specification: } \ \ \ \ \ S^{ij}(t^{ij}_0,t^{ij}_1 ) = \mathcal{K}_j(t^{ij}_0) \setminus \{ j, n , s_j(0 ) \} \ \ \ \ \ \ \ \ \ \ \ \ \ \ \ \ \ \ \ \ \ \ \ \ \ \ \ \ \ \ \ \ \ \ \ \ \ \ \ \  \end{equation}  \begin{equation}\label{dda2} \mbox{Knowledge Set Update: } \ \ \ \mathcal{K}_i(t^{ij}_1(+)) = \{ i, n , \hat{v}_i(t^{ij}_1(+)) , \hat{s}_i(t^{ij}_1(+)) , s_i(0) \} \ \ \ \ \ \ \ \ \ \ \ \ \ \ \ \ \ \ \ \ \ \ \ \ \end{equation} $$ \text{Normal Consensus Estimate Update: } \ \ \ \ \ \ \ \ \ \ \ \ \ \ \ \ \ \ \ \ \ \ \ \ \ \ \ \ \ \ \ \ \ \ \ \ \ \ \ \ \ \ \ \ \ \ \ \ \ \ \ \ \ \ \ \ \ \ \ \ \ \ \ \ \ \ \ \ \ \ \ \ \ \ $$ \begin{equation}\label{aaaddaveqa} \hat{v}_i ( t^{ij}_1 (+) ) = \hat{a} \hat{v} _i ( t^{ij}_1 ) + \hat{b} \hat{v}_j( t^{ij}_0 ) + \hat{c} e_i \end{equation} \begin{equation}\label{aaaddaveqb} \begin{array}{llll} ( \hat{a}, \hat{b}, \hat{c} ) = \big( \ 1 , \ 1 , \ - \hat{v}_{ji} (t^{ij}_0 ) \big) \ & \ \mbox{if $\hat{v}^{-i} _i ( t^{ij}_1 ) ' \hat{v}^{-i} _j ( t^{ij}_0 ) = 0$} \\ ( \hat{a}, \hat{b}, \hat{c} ) = \big( \ 0 , \ 1 , \ \frac{1}{n} - \hat{v}_{ji} (t^{ij}_0 ) \big) \ & \ \mbox{if $\hat{v}^{-i} _i ( t^{ij}_1 ) ' \hat{v}^{-i} _j ( t^{ij}_0 ) > 0 \ , \ \hat{v}^{-i} _i ( t^{ij}_1 ) ^2 < \hat{v}^{-i} _j ( t^{ij}_0 )^2 $} \\ ( \hat{a}, \hat{b}, \hat{c} ) = \big( \ 1 , \ 0 , \ 0 \big) \ & \ \mbox{if $\hat{v}^{-i} _i ( t^{ij}_1 ) ' \hat{v}^{-i} _j ( t^{ij}_0 ) > 0 \ , \ \hat{v}^{-i} _i ( t^{ij}_1 ) ^2 \geq \hat{v}^{-i} _j ( t^{ij}_0 )^2 $} \end{array} \end{equation} \begin{equation}\label{aaaredkda1} \mbox{Consensus Estimate Update: } \ \ \  \hat{s}_i( t^{ij}_1 (+) ) = \hat{a} \hat{s} _i ( t^{ij}_1 ) + \hat{b} \hat{s}_j( t^{ij}_0 ) + \hat{c} s_i(0) \ \ \ \ \ \ \ \ \ \ \ \ \ \ \ \ \ \ \ \ \ \ \ \ \ \ \ \end{equation} \begin{equation}\label{ddainit1} \mbox{Estimate Initialization: } \ \ \ \ \ \ \hat{v}_i(0) = \frac{1}{n} e_i \ , \ \hat{s}_i(0) = \frac{1}{n} s_i(0) \ . \ \ \ \ \ \ \ \ \ \ \ \ \ \ \ \ \ \ \ \ \ \ \ \ \ \ \ \ \ \ \ \ \ \ \ \ \ \ \ \ \ \end{equation} \end{framed} Notice that the DDA signal specification $(\ref{dda1})$ and knowledge set update $(\ref{dda2})$ are identical to those of the DA algorithm. Also notice that under $(\ref{aaaddaveqa})$ the relation $(\ref{a10})$ implies $(\ref{aaaredkda1})$. The Lemma $\ref{lembm0}$ and Lemma $\ref{aaalemkda1}$ in Appendix VII respectively prove that $(\ref{ddainit1})$ is the unique solution to $(\ref{aaaoptdisc})$ under $(A4)$, and that $(\ref{aaaddaveqa})-(\ref{aaaddaveqb})$ is a global solution to $(\ref{aaaoptdisc})$ under $(\ref{dda1})$ and $(\ref{dda2})$. From $(\ref{dda1})$ and $(\ref{dda2})$ it is clear that the DDA consensus estimate update $\hat{s}_i(t^{ij}_1(+))$ defined in $(\ref{aaaredkda1})$ can be locally computed at node $i$ based only on $\mathcal{K}_i(t^{ij}_1)$ and the received signal $S^{ij}( t^{ij}_0, t^{ij}_1)$.

\emph{Summary:} We have defined four algorithms that in Sec.IV will be shown to achieve average-consensus. The algorithms were derived as special cases of the distributed algorithm $(\ref{update0}),(\ref{update1})$, where $(\ref{update1})$ is a special case of $(\ref{update10})$. Computationally, each algorithm requires only a few elementary arithmetic operations. In Sec.$\ref{sec:rescost}$ we discuss the storage and communication costs of these four algorithms.

\section{Main Results: Analysis of the Average-Consensus Algorithms}\label{sec:main} This section proves that Algorithms $1-4$ described in Sec.III achieve average-consensus under different communication conditions, we state this below in five theorems. Under assumptions (A1)-(A9) listed in Sec.$\ref{sec:probstat}-\ref{sec:commod1}$, these theorems provide necessary and sufficient conditions on the communication between nodes for average-consensus formation under Algorithms $1-4$ (recall Definition $\ref{def10}$ defines the solution to the average-consensus problem). The main implication of these results is that, with or without flooding, an average-consensus can be achieved in the presence of arbitrary communication delays and link failures, provided only that there is a uni-directional connectivity condition among nodes.

Each theorem below will assume a certain condition on the communication among nodes. To specify these conditions we require the following two definitions. For any $t_1 \geq t_0 \geq 0$, an arbitrary ``communication sequence'' $C_{(t_0,t_1 )}$ is defined as the set of all signals transmitted after time $t_0$ and received before time $t_1$, that is, $$ C_{[t_0, t_1 ]} = \{ S^{i_1 j_1} , S^{i_2 j_2 }  , S^{i_3 j_3 } , \ldots \} $$ where we have omitted the time indices but it is understood that the transmission time $t^{i_\ell j_\ell}_0 $ and reception time $t^{i_\ell j_\ell}_1$ of each signal $S^{i_\ell j_\ell}$ belong to the interval $[t_0, t_1]$. Recall that a signal $S^{i j} ( t^{i j}_0, t^{i j}_1 )$ as stated in (A2) implies that a well-defined subset of $\mathcal{K}_j(t^{i j}_0)$ leaves node $j$ at time $t^{i j}_0$ and is received at node $i$ at time $t^{i j}_1$ (the specific subset depends on the algorithm considered).

Next we define the notion of a ``communication path''. Intuitively, a communication path from node $j$ to $i$ implies that node $j$ transmits a signal received by some node $\ell_1$, and \emph{then} node $\ell_1$ sends a signal received by some node $\ell_2$, and \emph{then} node $\ell_2$ sends a signal received by some node $\ell_3$, and so on, until node $i$ receives a signal from node $\ell_{k(ij)}$. Technically, we say that $C_{[t_0,t_1]}$ contains a ``communication path'' from node $j$ to node $i$ iff $C_{[t_0,t_1]}$ contains a sub-sequence $C^{ij}_{[t_0(ij),t_1(ij)]}$ with the following connectivity property, \begin{equation}\label{pairs}\begin{array}{llll} & C^{ij}_{[t_0(ij),t_1(ij)]} \supseteq \{ S^{\ell_1 j}  , S^{\ell_2 \ell_1} ,  S^{\ell_3 \ell_2} , \\ & \ \ \ \ \ \ \ \ \ \ \ \ \ \ \ \ \ \ \ \ \ldots, S^{\ell_{k(ij)} \ell_{k(ij)-1}} , S^{i \ell_{k(ij)}}  \}  \end{array}\end{equation} where again we have omitted the time indices but it is understood that the transmission time $t^{\ell_{q+1} \ell_{q}}_0$ of each signal $S^{\ell_{q+1} \ell_{q}}$ occurs \emph{after} the reception time $t^{\ell_q \ell_{q-1}}_1$ of the preceding signal $S^{\ell_q \ell_{q-1}}$. Note that the communication path $C^{ij}_{[t_0(ij),t_1(ij)]}$ has a cardinality $|C^{ij}_{[t_0(ij),t_1(ij)]} | \geq k(ij) + 1$.

\subsection{Necessary and Sufficient Conditions for Average-Consensus}\label{sec:nscon} We are now ready to present the main convergence theorems for Algorithms $1-4$.

\subsubsection{Algorithm 1 (BM)}\label{sec:sbm} To prove convergence of the BM algorithm, consider the following communication condition. Let us denote $\mathcal{V}_{-i} = \mathcal{V} \setminus \{ i \}$ for an arbitrary node $i \in \mathcal{V}$. \begin{defn} (S$\mathcal{V}$SC) A communication sequence $C_{[0,t_1]}$ is ``singly $\mathcal{V}$-strongly connected'' (S$\mathcal{V}$SC) iff there exists a communication path from each node $i \in \mathcal{V}$ to every node $j \in \mathcal{V}_{-i}$. $\quad \quad \quad \quad \quad \quad \ \ \ \ \ \ \ \ \ \ \ \ \ \ \ \ \ \ \ \ \ \ \ \ \ \ \ \ \ \ \Box$ \end{defn} We will let ``$C_{[t_0,t_1]} \in \mbox{S$\mathcal{V}$SC}$'' denote that a sequence $C_{[t_0,t_1]}$ is S$\mathcal{V}$SC. The definition $(\ref{pairs})$ implies that $C_{[t_0,t_1]} \in \mbox{S$\mathcal{V}$SC}$ iff, \begin{equation}\label{scdef} C^{ij}_{[t_0(ij),t_1(ij)]} \subset C_{[0,t_1]} \ , \ \ \forall \ j \in \mathcal{V}_{-i} \ , \ \ \forall \ i \in \mathcal{V} \ . \end{equation} The following theorem establishes the sufficient communication conditions for the BM algorithm $(\ref{bm1})-(\ref{bminit1})$.

\begin{thrm}\label{bmthm}(BM Sufficient Conditions) Consider Algorithm 1, namely the BM algorithm $(\ref{bm1})-(\ref{bminit1})$. Then the average-consensus $(\ref{acdef})$ is achieved at time $t=t_1(+)$ for any communication sequence $C_{[0,t_1]}$ satisfying the $\mbox{S$\mathcal{V}$SC}$ condition $(\ref{scdef})$. \end{thrm} \begin{proof} See Appendix VII. \end{proof}

The BM algorithm implies a ``flooding'' of the initial consensus vectors $\{ s_i(0) \ : \ i \in \mathcal{V} \}$ through-out the network. There is no other algorithm in the literature that does not use a protocol ``equivalent'' to the flooding technique and still guarantees average-consensus for every communication sequence $C_{[0,t_1]}$ satisfying $(\ref{scdef})$, we state this formally as Conjecture $\ref{conject0}$ in Sec.$\ref{sec:app1a}$. Also related to this is Conjecture $\ref{conject1}$ in Sec.$\mbox{VII-B}$ of Appendix VII.

When combined with Theorem $\ref{bmthm}$, the following theorem implies that the communication condition \emph{sufficient} for average-consensus under the BM algorithm is the exact same communication condition that is \emph{necessary} for average-consensus under any algorithm. This is why the BM algorithm can be said to possess optimal convergence properties.

\begin{thrm}\label{bmcor}(BM,DA,OH and DDA Necessary Conditions) If a communication sequence $C_{[0,t_1]}$ does not satisfy the $\mbox{S$\mathcal{V}$SC}$ condition $(\ref{scdef})$, then no distributed algorithm $(\ref{update10}),(\ref{update0})$ can achieve average-consensus $(\ref{acdef})$ at time $t=t_1(+)$. \end{thrm} \begin{proof} See Appendix VII. \end{proof}

Although Theorem $\ref{bmcor}$ is somewhat obvious, it is also a valid necessary communication condition and will be referred to through-out the paper. Note that Theorem $\ref{aaacor5}$ below states that the OH algorithm requires even stronger necessary conditions than the $\mbox{S$\mathcal{V}$SC}$ condition $(\ref{scdef})$ implied by Theorem $\ref{bmcor}$. However, the DA and DDA algorithms do not; there exist many $\mbox{S$\mathcal{V}$SC}$ sequences under which the DA and DDA algorithms will obtain average-consensus at the same instant as the BM algorithm. An example is the ``unit-delay double cycle sequence'' defined in $(\ref{double})$ of Appendix VII. Together with Theorem $\ref{bmcor}$, the example $(\ref{double})$ implies that the DA and DDA algorithms possess the weakest possible necessary conditions for average-consensus that any algorithm can have, this is illustrated in Fig.$\ref{simpfig2}$ below.

\subsubsection{Algorithm 2 (DA)}\label{sec:sda} The convergence of Algorithm 2 depends on the following I$\mathcal{V}$SC condition.

\begin{defn} (I$\mathcal{V}$SC) A communication sequence $C_{[0,t_1]}$ is ``infinitely $\mathcal{V}$-strongly connected'' (I$\mathcal{V}$SC) iff for each time instant $t \in [0,t_1)$ there exists a finite span of time $T_t \in (0, t_1 - t)$ such that $C_{[t,t + T_t]}$ satisfies the S$\mathcal{V}$SC condition $(\ref{scdef})$. $ \quad \quad \quad \quad \quad \quad \quad \ \ \ \ \ \ \ \ \ \ \ \ \ \ \ \ \ \ \ \ \ \ \ \ \ \ \ \ \ \ \ \ \ \ \ \ \ \ \ \ \ \ \ \ \ \ \ \ \ \ \ \ \ \ \ \ \ \ \ \ \ \ \ \ \ \ \ \ \ \ \ \ \ \ \ \ \ \ \ \ \ \ \ \ \ \ \ \ \ \ \ \ \ \ \ \ \ \ \ \ \ \ \ \ \ \ \ \ \ \Box$ It follows that a sequence $C_{[0,t_1]}$ is I$\mathcal{V}$SC iff we can define the infinite set of ordered pairs $\{ (t^\ell_0, t^\ell_1 )  \ : \ \ell \in \mathbb{N} \}$, $$ \begin{array}{llll} & t^0_0 = 0 \ , \ t^\ell_1 = \mbox{arg min}_t \{ C_{[t^\ell_0,t]} \in \mbox{S$\mathcal{V}$SC} \} \\ & t^\ell_0 = t^{\ell-1}_1 (+) \ , \ \ell = 1,2, \ldots . \end{array} $$ A sequence $C_{[0,t_1]}$ being I$\mathcal{V}$SC is then equivalent to the condition, \begin{equation}\label{star1ga}\begin{array}{llll} & C_{[0,t_1]} = \bigcup_{  \ell \in \mathbb{N} }  \ C_{[t^\ell_0, t^\ell_1]}  \\ &  C_{[t^\ell_0, t^\ell_1]} \in \mbox{S$\mathcal{V}$SC} \ , \ \forall \ \ell \in \mathbb{N} \ . \end{array}\end{equation} \end{defn} We now proceed to the convergence of the DA algorithm. The following theorem is our main result.

\begin{thrm}\label{thm6}(DA Sufficient Conditions) Consider Algorithm 2, namely the DA algorithm $(\ref{kda1})-(\ref{dainit1})$. Then average-consensus $(\ref{acdef})$ is achieved at time $t=t_1(+)$ for any communication sequence $C_{[0,t_1]}$ satisfying the I$\mathcal{V}$SC condition $(\ref{star1ga})$. \end{thrm} \begin{proof} See Appendix VII. \end{proof}

The above result is interesting because the $\mbox{I$\mathcal{V}$SC}$ condition $(\ref{star1ga})$ assumes a weak recurring connectivity between nodes, and also because the resource costs of the DA algorithm are significantly lower than the BM algorithm (see Sec.$\ref{sec:rescost}$). Many papers on average-consensus formation such as \cite{mo06},\cite{SH06},\cite{HA05},\cite{ne09},\cite{MU04},\cite{ja08} assume communication conditions that are special cases of the IVSC condition. See also Conjecture $\ref{conject2}$ in Sec.$\mbox{VII-B}$ of Appendix VII.

\subsubsection{Algorithm 3 (OH)}\label{sec:soh} The following S$\mathcal{V}$CC condition will be shown sufficient for convergence of the OH algorithm. \begin{defn} (S$\mathcal{V}$CC) A communication sequence $C_{[0,t_1]}$ is ``singly $\mathcal{V}$-completely connected'' (S$\mathcal{V}$CC) iff there exists a node $\hat{i} \in \mathcal{V}$ and a time instant $t_{1/2} \in (0, t_1)$ such that, \begin{equation}\label{symcondef2}\begin{array}{llll} & S^{\hat{i} j}( t^{\hat{i} j}_0, t^{\hat{i} j}_1 ) \in C_{[0, t_{1/2} ] } \ \ , \ \ \forall \ j \in \mathcal{V}_{-\hat{i}} \\ & S^{j \hat{i}}( t^{j \hat{i} }_0, t^{j \hat{i}}_1) \in C_{(t_{1/2} ,t_1]} \ , \ \ \forall \ j \in \mathcal{V}_{-\hat{i}} \ . \end{array}\end{equation} $\quad \quad \quad \quad \quad \quad \quad \quad \quad \quad \quad \quad \quad \quad \quad \quad \quad \quad \quad \quad \quad \quad \quad \quad \quad \quad \quad \quad \quad \quad \quad \quad \quad \quad \quad \quad \quad \ \ \ \ \ \ \ \ \ \ \ \ \ \ \ \ \ \ \ \ \ \ \ \ \ \ \ \ \ \ \ \Box$ \end{defn} The first line in $(\ref{symcondef2})$ implies that during the interval $[0, t_{1/2} ]$ every node $j \in \mathcal{V}_{-\hat{i}}$ will have transmitted a signal that was received by node $\hat{i}$. The second line in $(\ref{symcondef2})$ implies that during the interval $(t_{1/2},t_1]$ the node $\hat{i}$ will have transmitted a signal that is received by each node $j \in \mathcal{V}_{-\hat{i}}$. We will let ``$C_{[t_0,t_1]} \in \mbox{S$\mathcal{V}$CC}$'' denote that a sequence $C_{[t_0,t_1]}$ is S$\mathcal{V}$CC. Note that the identity of node $\hat{i}$ need not be known by any node. We now consider convergence of the OH algorithm.

\begin{thrm}\label{aaacor5}(OH Necessary and Sufficient Conditions) Consider Algorithm 3, namely the OH algorithm $(\ref{aaasdasig5})-(\ref{ohinit1})$. Then average-consensus $(\ref{acdef})$ is achieved at time $t = t_1(+)$ iff the communication sequence $C_{[0,t_1]}$ satisfies the following condition: \begin{itemize} \item (C): for each node $i$ for which there exists a node $j \in \mathcal{V}_{-i}$ such that $S^{ij}( t^{ij}_0, t^{ij}_1) \notin C_{[0,t_1]}$ for all $(t^{ij}_0, t^{ij}_1) \in \mathbb{R}^2$, there exists a communication path $C^{i \ell}_{[t_0(i \ell),t_1(i \ell)]} \in C_{[0,t_1]}$ from at least one node $\ell$ such that $S^{\ell j}( t^{\ell j}_0, t^{\ell j}_1) \in C_{[0,t_0(i \ell))}$ for all $j \in \mathcal{V}_{-\ell}$.\end{itemize} \end{thrm} \begin{proof} See Appendix VII. \end{proof}

Notice that any communication sequence that satisfies the S$\mathcal{V}$CC condition $(\ref{symcondef2})$ will also satisfy the condition (C) stated in Theorem $\ref{aaacor5}$; take $\ell = \hat{i}$, $t_0(j \ell) = t_{1/2}(+)$, and $C^{j \ell}_{[t_0(j \ell),t_1(j \ell)]} = S^{j \hat{i}}( t^{j \hat{i} }_0, t^{j \hat{i}}_1)$ for all $j \in \mathcal{V}_{-\ell}$. This relation motivates the position of the OH algorithm in the Venn diagrams presented in Fig.$\ref{simpfig2}$. See also Remark $\ref{rem1}$ in Sec.$\ref{sec:app1a}$. The condition (C) is more general than S$\mathcal{V}$CC, it implies that each node $i \in \mathcal{V}$ will either receive a signal directly from every other node $j \in \mathcal{V}_{-i}$, or have a communication path from some node $\ell$ after the node $\ell$ has received a signal directly from every other node $j \in \mathcal{V}_{-\ell}$. We have defined the S$\mathcal{V}$CC condition not only because it is sufficient for average-consensus under the OH algorithm, but also because it is necessary for the definition of the communication condition I$\mathcal{V}$CC described next.

\subsubsection{Algorithm 4 (DDA)}\label{sec:sdda} The sufficient conditions for convergence under Algorithm 4 requires the following definition. \begin{defn} (I$\mathcal{V}$CC) A communication sequence $C_{[0,t_1]}$ is ``infinitely $\mathcal{V}$-completely connected'' (I$\mathcal{V}$CC) iff for each time instant $t \in [0,t_1)$ there exists a finite span of time $T_t \in (0, t_1 - t)$ such that $C_{[t,t + T_t]}$ satisfies the $\mbox{S$\mathcal{V}$CC}$ condition $(\ref{symcondef2})$. $\quad \quad \quad \quad \quad \quad \quad \quad \quad \quad \quad \quad \quad \quad \quad \quad \quad \quad \quad \quad  \quad \quad \quad \quad \quad \quad  \quad \quad \quad \quad \quad \quad \ \ \ \ \ \ \ \ \ \ \ \ \ \ \ \ \ \ \ \ \ \ \ \ \ \ \Box$ \end{defn} It follows that a sequence $C_{[0,t_1]}$ is I$\mathcal{V}$CC iff we can define the infinite set of ordered pairs $\{ (t^\ell_0, t^\ell_1 )  \ : \ \ell \in \mathbb{N} \}$, $$ \begin{array}{llll} & t^0_0 = 0 \ , \ t^\ell_1 = \mbox{arg min}_t \{ C_{[t^\ell_0,t]} \in \mbox{S$\mathcal{V}$CC} \} \\ & t^\ell_0 = t^{\ell-1}_1 (+) \ , \ \ell = 1,2, \ldots . \end{array} $$ A sequence $C_{[0,t_1]}$ being I$\mathcal{V}$CC is then equivalent to the condition, \begin{equation}\label{star1g}\begin{array}{llll} & C_{[0,t_1]} = \bigcup_{  \ell \in \mathbb{N} }  C_{[t^\ell_0, t^\ell_1]} \\ &  C_{[t^\ell_0, t^\ell_1]} \in \mbox{S$\mathcal{V}$CC} \ , \ \forall \ \ell \in \mathbb{N} \ . \end{array}\end{equation} We note that the specific node $\hat{i}$ can vary between each $C_{[t^{\ell}_0,t^{\ell}_1]}$ sequence, and furthermore we do not assume that any node knows the identity of $\hat{i}$. The following theorem deals with the convergence of Algorithm 4.

\begin{thrm}\label{aaathm5}(DDA Sufficient Conditions) Consider Algorithm 4, namely the DDA algorithm $(\ref{dda1})-(\ref{ddainit1})$. Then average-consensus $(\ref{acdef})$ is achieved at some time $t \in (0, t_1)$ for any communication sequence $C_{[0,t_1]}$ satisfying the $\mbox{I$\mathcal{V}$CC}$ condition $(\ref{star1g})$. \end{thrm} \begin{proof} See Appendix VII. \end{proof}

The above result is not quite as interesting as Theorem $\ref{thm6}$, since even though the DDA algorithm is a discretized version of the DA algorithm, the $\mbox{I$\mathcal{V}$CC}$ condition $(\ref{star1g})$ is far stronger than the $\mbox{I$\mathcal{V}$SC}$ condition $(\ref{star1ga})$. Also observe that the OH algorithm obtains average-consensus for any $C_{[0,t_1]} \in \mbox{S$\mathcal{V}$CC}$, whereas the DDA algorithm will obtain average-consensus under the much stronger condition that $C_{[0,t_1]}$ satisfies $(\ref{star1g})$. However, as mentioned above, due to example $(\ref{double})$ of Appendix VII the DDA and DA algorithms both can obtain average-consensus under suitable $\mbox{S$\mathcal{V}$SC}$ sequences, thus they possess much weaker necessary conditions than the OH algorithm. In fact, any communication sequence $C_{[0,t_1]}$ that strictly satisfies $(\ref{symcondef2})$ implies Algorithms $1-4$ all obtain average-consensus at the exact same time instant. This is noteworthy because many past algorithms can only achieve average-consensus asymptotically (e.g. most iterative averaging schemes), in contrast all four of the algorithms considered here can achieve the (finite) bench-mark time for average-consensus under appropriate communication sequences (e.g. any sequence $C_{[0,t_1]}$ that strictly satisfies $(\ref{symcondef2})$).

\emph{Summary.} Theorems $\ref{bmthm},\ref{bmcor},\ref{thm6},\ref{aaacor5},\ref{aaathm5}$ above state necessary and sufficient communication conditions for average-consensus under Algorithms $1-4$ given the assumptions (A1)-(A9). Each theorem is associated with one of four connectivity assumptions on the communication sequence $C_{[0,t_1]}$, denoted by $\{$ $\mbox{S$\mathcal{V}$SC}$, $\mbox{I$\mathcal{V}$SC}$, $\mbox{S$\mathcal{V}$CC}$, $\mbox{I$\mathcal{V}$CC}$ $\}$ defined in $(\ref{scdef})$,$(\ref{star1ga})$,$(\ref{symcondef2})$, and $(\ref{star1g})$. Observe that assumptions $\mbox{I$\mathcal{V}$SC}$ and $\mbox{S$\mathcal{V}$CC}$ are sufficient conditions for $\mbox{S$\mathcal{V}$SC}$. Furthermore, $\mbox{I$\mathcal{V}$CC}$ implies that both $\mbox{I$\mathcal{V}$SC}$ and $\mbox{S$\mathcal{V}$CC}$ are satisfied, see Fig.$\ref{simpfig2}$. Notice that each connectivity condition assumes a set of directed signals with an arbitrary delay in the transmission time of each signal. This is significant because, apart from the flooding technique, no other consensus protocol in the current literature can ensure average-consensus in the presence of arbitrary delays in the transmission time of each signal for \emph{a priori} unknown communication sequences. Of course if the communication sequence is known \emph{a priori}, then specific update protocols can always be constructed that guarantee average-consensus at the same instant as the BM algorithm. Current results on average-consensus that do allow communication delays assume the delays have pre-determined upper-bounds, and also require the use of averaging weights that are globally balanced and pre-determined off-line (see for example \cite{MU04}, \cite{li08}, \cite{zh10}). On a related note, besides flooding there appears to be no consensus protocol in the literature that has been proven to guarantee average-consensus for \emph{every} communication sequence $C_{[0 , t_1]}$ satisfying any one of the conditions $(\ref{scdef})$,$(\ref{star1ga})$,$(\ref{symcondef2})$, or $(\ref{star1g})$. The majority of past results on average-consensus either require special cases of the $\mbox{I$\mathcal{V}$SC}$ condition, or else can only  guarantee approximate average-consensus under the $\mbox{S$\mathcal{V}$SC}$ condition \cite{SH08}, \cite{co04}, \cite{ba03}. On the other hand, the two non-trivial algorithms DA and DDA require a set with cardinality upper-bounded bounded by $O(n+d)$ to be communicated and stored at each node, and all previous algorithms besides the flooding and randomized protocols do not possess this drawback.

The Venn diagrams in Fig.$\ref{simpfig2}$ summarize Theorems $\ref{bmthm},\ref{bmcor},\ref{thm6},\ref{aaacor5},\ref{aaathm5}$, as well as their relation to the Gossip and ARIS algorithms that will be used as comparisons to the four proposed algorithms. It remains an open problem whether any protocol exists that guarantees average-consensus for all $\mbox{I$\mathcal{V}$SC}$ sequences without requiring a set with cardinality upper bounded by $O(n+d)$ to be stored and communicated at each node.

\begin{figure}[htb]
\includegraphics[width=0.7\linewidth]{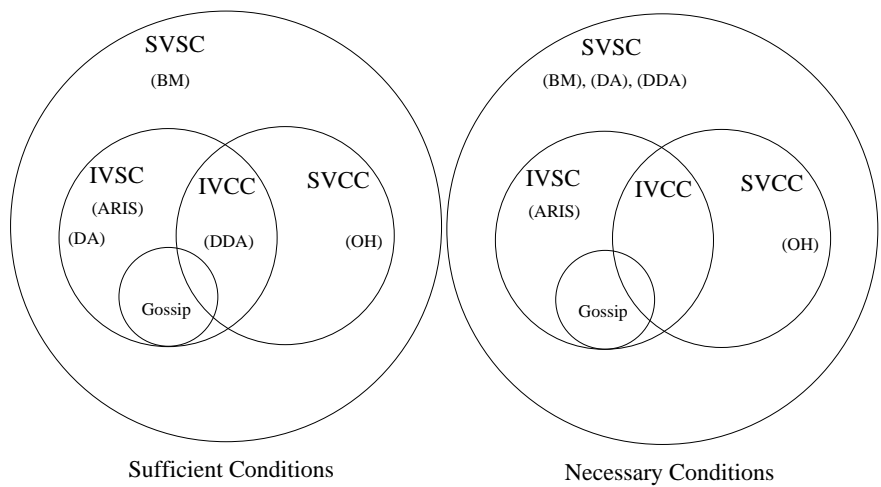}
\caption{Venn Diagram of Sufficient and Necessary Conditions for Algorithms 1-4 as well as the Comparison Algorithms Gossip and ARIS to Achieve Average-Consensus. The condition (C) in Theorem $\ref{aaacor5}$ is omitted for simplicity of presentation.}
\label{simpfig2}
\end{figure}

\section{Numerical Implementation and Examples}\label{sec:numimp} This section presents the communication and storage costs of the various algorithms proposed above. In Sec.$\ref{sec:rescost}$ below we define the resource costs for the six algorithms compared in this paper (four proposed in Sec.III and two others in Sec.$\mbox{VII-B}$). In Sec.$\ref{sec:figures}$ we present numerical simulations of the six algorithms under various randomized full network graphs when assuming instantaneous communication.

\subsection{Algorithm Resource Costs}\label{sec:rescost} Each of the four average-consensus algorithms as well as the two comparison algorithms presented in Sec.$\mbox{VII-B}$ require that the knowledge set $\mathcal{K}_i(t)$ and signal set $S^{i j }(t^{i j}_0, t^{i j}_1 )$ are respectively defined by a set of scalars, where each scalar has a particular meaning. We are thus motivated to quantify the ``resource cost'' of each algorithm in terms of the total number of scalar values that are required to define the two sets $\mathcal{K}_i(t)$ and $S^{i j }(t^{i j}_0, t^{i j}_1 )$. In particular, for any $t \geq 0$ we define the ``storage cost'' of an arbitrary node $i \in \mathcal{V}$ as $\phi_i(t)$, \begin{equation}\label{phidef}\begin{array}{llll} & \phi_i(t) = \mbox{the minimum number of scalars} \\ & \ \ \ \ \ \ \ \mbox{required to define the knowledge set $\mathcal{K}_i(t)$.} \end{array}\end{equation} Likewise we define the ``communication cost'' $\rho_{ij} (t^{ij}_0)$ of an arbitrary signal $S^{ij}(t^{ij}_0, t^{ij}_1)$ as follows, \begin{equation}\label{xidef}\begin{array}{llll} & \rho_{ij} (t^{ ij}_0) = \mbox{the minimum number of scalars} \\ & \ \ \ \ \ \ \ \ \ \mbox{required to define the set $S^{ij}(t^{ij}_0, t^{ij}_1)$.} \end{array}\end{equation} Next define the total resource cost of node $i$ at time $t \geq 0$ as $\omega_i(t)$,\begin{equation}\label{phixidef} \omega_i (t) = \phi_i(t) + \rho_{ij} (t) \ . \end{equation} Based on the knowledge set and signal specifications defined in Sec.III, the Table I presents the resource cost computations of each algorithm when using the definitions $(\ref{phidef})$, $(\ref{xidef})$, and $(\ref{phixidef})$. The entries of Table I are derived in Sec.$\ref{sec:rescostapp}$ of Appendix VII.

\begin{table}\label{tab1}\begin{tabular}{| l | l | l | l | l | l | l |} \hline Algorithm & BM & DA & OH & DDA & Gossip & ARIS \\ \hline $\min (\phi_i(t))$ & $4d+5$ & $4d+6$ & $4d+5$ & $4d+5$ & $2d$ & $7+2(r+2)d$ \\ \hline $\max (\phi_i(t))$ & $2nd+4+\lfloor \frac{n}{2} \rfloor$ & $4d+2n+4$ & $4d+4+\lfloor \frac{n}{2} \rfloor$ & $4d+4+\lfloor \frac{n}{2} \rfloor$ & $2d$ & $\lfloor \frac{n}{2} \rfloor+6+2(r+2)d$ \\ \hline $\min (\rho_{ij}(t))$ & $2d+1$ & $2d+1$ & $2d+2$ & $2d+1$ & $2d$ & $3+2(r+1)d$ \\ \hline $\max (\rho_{ij} (t))$ & $2(n-1)d+\lfloor \frac{n}{2} \rfloor$ & $2d+2n$ & $2d+2$ & $2d+\lfloor \frac{n}{2} \rfloor$ & $2d$ & $\lfloor \frac{n}{2} \rfloor+2+2(r+1)d$ \\ \hline $\min (\omega_i(t))$ & $6d+6$ & $6d+8$ & $6d+7$ & $6d+6$ & $4d$ & $10+4(r+\frac{3}{2})d$ \\ \hline $\max (\omega_i(t))$ & $2(2n-1)d+4+2\lfloor \frac{n}{2} \rfloor$ & $6d+4n+4$ & $6d+6+\lfloor \frac{n}{2} \rfloor$ & $6d+4+2\lfloor \frac{n}{2} \rfloor$ & $4d$ & $2 \lfloor \frac{n}{2} \rfloor+8+4(r+\frac{3}{2})d$ \\ \hline $O( \phi )$ & $nd$ & $n+d$ & $n+d$ & $n+d$ & $d$ & $rd+n$ \\ \hline $O(\rho)$ & $nd$ & $n+d$ & $d$ & $n+d$ & $d$ & $rd+n$ \\ \hline \end{tabular} \caption{Minimum and Maximum Resource Costs obtained by applying each Algorithm in Sec.III to $(\ref{phidef})$, $(\ref{xidef})$, and $(\ref{phixidef})$. Here $d$ denotes the dimension of the initial consensus vectors, $n$ denotes the number of nodes, $\rho_{ij}(t)$ denotes the signal set dimension, $\phi_{i}(t)$ the knowledge set dimension, and $\omega_i(t)$ is the sum of $\rho_{ij}(t)$ and $\phi_{i}(t)$. We let $\lfloor \cdot \rfloor$ denote the ``floor'' operation.} \end{table}

Note that the storage cost $\phi_{i}(t)$ is defined per node, and the communication cost $\rho_{ij}(t)$ is defined per signal. The total maximum resource costs of the BM algorithm increase on the order $O(nd)$, whereas the total maximum resource costs of the DA and DDA algorithm increase on the order $O(n+d)$. Although the total maximum resource costs of the OH algorithm increase on the order $O(n+d)$, the maximum resource cost of each signal under the OH algorithm increases only on the order $O(d)$. However, the communication conditions necessary for average-consensus under the OH algorithm are much stronger than the conditions necessary under the BM, DA, and DDA algorithms. This disparity makes it difficult to state definitive results regarding the least costly algorithm under general communication sequences. If condition (C) in Theorem $\ref{aaacor5}$ is known to hold \emph{a priori}, then the OH algorithm may be preferable to the DA and DDA due to the lower communication costs $O(d) \leq O(d+n)$. On the other hand, under example $(\ref{double})$ in Appendix VII the DDA algorithm is preferable to both the DA and OH algorithm, since the former implies larger resource costs and the latter will not obtain average-consensus.

The Gossip algorithm in \cite{SH06} has total resource costs that increase on the order $O(d)$, however this algorithm requires strictly bi-directional and instantaneous communication, see Sec.$\mbox{VII-B}$ as well as Fig.$\ref{simpfig1}$ in Sec.$\ref{sec:figures}$. The total resource costs of the ARIS algorithm increase on the order $O(rd + n)$, where $r$ is an ARIS algorithm parameter explained in Sec.$\mbox{VII-B}$. If $r \geq n$ then ARIS is more costly than the BM algorithm. For this reason, the simulations presented in Sec.$\ref{sec:figures}$ assume $r = n$. The RIS algorithm proposed in \cite{SH08} requires $r$ random variables to be initially generated at each node for each element of the respective local initial consensus vector $s_i(0) \in \mathbb{R}^d$, the RIS algorithm also requires these random variables to be communicated between nodes, thus both the storage and communication costs of the RIS algorithm increase on the order $O(rd)$. 

\subsection{Numerical Results}\label{sec:figures}

We present here numerical simulations of the four proposed algorithms together with the two comparison algorithms. The algorithm parameters were chosen as $n = 80$ (number of nodes), $d = 1$ (dimension of the initial consensus vectors), $r = n$ (number of ARIS random variables generated per initial consensus element), and $s_i(0) = i , \ \forall \ i \in \mathcal{V}$ (initial consensus vector values). The node communication is assumed to be instantaneous and in discrete-time, using the following randomized protocol: \begin{itemize} \item at each $k \in \mathbb{N}$, two nodes $(i,j) \in \mathcal{V}^2$ are uniformly chosen at random such that $i \neq j$, \item with probability one, the node $i$ sends a signal to the node $j$ at time $k$, \item with probability $p$ the node $j$ sends a signal to the node $i$ at time $k$. \end{itemize} We compare four choices of $p$, namely $p \in \{ 1,  \frac{1}{2},  \frac{1}{4}, 0 \}$. Note that $p=1$ implies instantaneous bi-directional communication. As $p$ decreases there will be fewer expected signals per time instant, thus we expect that each algorithm will have slower convergence for lower values of $p$. 

\begin{figure}[htbp!]
\includegraphics[width=1 \linewidth]{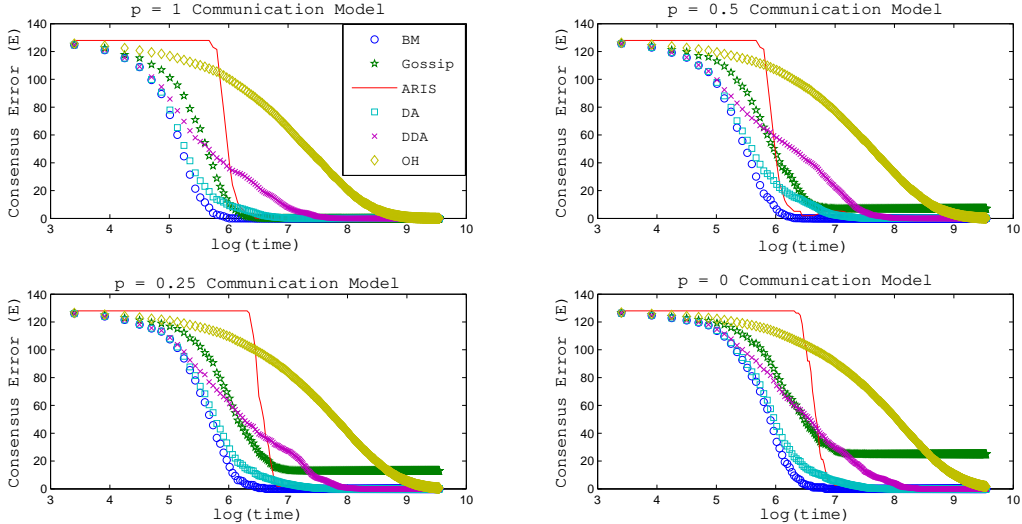}
\caption{Convergence of Consensus Algorithms for Communication Probabilities $p \in \{ 1, \frac{1}{2} ,\frac{1}{4} ,0 \}$ (top left to bottom right). The consensus error $E$ is defined above $(\ref{acdef})$. Observe that the Gossip consensus error vanishes only for $p = 1$, because this is the only model that implies strictly bi-directional communication. All algorithms converge moderately slower as $p$ decreases.}
\label{simpfig1}
\end{figure}

Fig.$\ref{simpfig1}$ shows the network consensus error under each algorithm for each value of $p \in \{ 1,  \frac{1}{2},  \frac{1}{4}, 0 \}$. It is clear that the Gossip algorithm only converges to average-consensus for $p=1$. For $p \in \{ \frac{1}{2} , \frac{1}{4} , 0 \}$ the Gossip algorithm converges to a consensus that is increasingly distant from the average-consensus. The BM algorithm converges fastest in all simulations, as expected from Theorems $\ref{bmthm},\ref{bmcor}$. Initially the DA algorithm can be observed to almost match the BM algorithm in all simulations, however as time proceeds for $p =1$ the Gossip and ARIS algorithm eventually overcome DA, and in $p \in \{ \frac{1}{2},\frac{1}{4}, 0\}$ the ARIS algorithm eventually overcomes DA. However, the resource costs of ARIS in our simulations is greater than even the BM algorithm, thus further research is needed to objectively evaluate the trade-off between resource cost and network consensus error for each algorithm under various communication conditions.

\section{Conclusion}\label{sec:conc} This paper has described and analyzed four distributed algorithms designed to solve the average-consensus problem under general uni-directional connectivity conditions. We derived necessary and sufficient conditions for the convergence to average-consensus under each respective algorithm. The conditions for convergence were based on two types of connectivity conditions, namely the singly $\mathcal{V}$-strongly connected sequence (defined in Sec.$\ref{sec:sbm}$), and the singly $\mathcal{V}$-completely connected sequence (defined in Sec.$\ref{sec:soh}$). Both connectivity conditions allow arbitrary delays in the transmission time of each signal, and we did not assume that the sending node knows the identity of the receiving node. The resource costs for each of the algorithms were derived and shown to differ in regard to their order of magnitude with respect to the parameters $n$ (number of nodes) and $d$ (dimension of the initial consensus vectors). Comparisons were made with two known consensus algorithms, referred to as the Gossip algorithm and the adapted randomized information spreading (ARIS) algorithm. Simulations were provided as well as Venn diagrams of the connectivity conditions required for average-consensus under each algorithm. The non-trivial algorithms considered here are relatively advantageous under weak communication conditions if the dimension $d$ of the initial consensus vectors exceeds the network size $n$. The works \cite{SB05} and \cite{ol05} provide two practical examples of when $d \geq n$ might typically be the case, e.g. distributed inference regarding a $\sqrt{n}$ dimensional process or parameter in noise.

The four communication conditions we proposed were deterministic; there were no stochastic properties assumed in regard to the signal process between nodes. However, our framework allowed directed signals as well as arbitrary delays in transmission time, hence \emph{every} causal signal sequence is a special case of our framework. This suggests that the four proposed algorithms can be applied to stochastic communication models for which there is a non-zero probability of consensus under any distributed algorithm, however future work is needed in this direction.

Future work is also needed to obtain a lower bound on the convergence rate to average-consensus for given characterizations of the communication sequence under each proposed algorithm, as well as improved algorithms designed specifically for particular communication sequences. For an objective evaluation of the various average-consensus algorithms, additional research is needed to compare the evolution of a cost function of the resource cost and network consensus error under a variety of communication sequences. Lastly, an interesting generalization of the average-consensus problem is to allow the initial consensus vectors $s_i(0)$ to vary with time, as discussed for instance by \cite{ol05}, \cite{cy89}, \cite{SP05}, \cite{VK09}, \cite{zh08}, \cite{mu07}. Applying the algorithms and communication conditions proposed in this work could yield further results with regard to ensuring the average $\bar{s}(t)$ is contained in the knowledge set $\mathcal{K}_i(t)$ of each node $i \in \mathcal{V}$ at some time instants $t$ for given dynamic models of the set $\{ s_i(t) \ : \ i \in \mathcal{V} \}$.

As a final note, it is worth mentioning that each of the four proposed algorithms can obtain a consensus on any \emph{linear combination} of the initial consensus vectors $\{s_i(0) \in \mathbb{R}^d \ : \ i \in \mathcal{V} \}$ under the exact same communication conditions as stated in the main results. In other words, suppose a vector $w \in \mathbb{R}^n$ was initially known at each node, then if each algorithm updated the normal consensus estimate based on $(\ref{optllca})$ with $\frac{1}{n} \mathbf{1}_n$ replaced by the vector $w$, the same conditions stated in Sec.IV will imply the respective algorithms ensure $\hat{s}_i(t_1(+)) = \sum_{i \in \mathcal{V}} s_i(0) w_i$ for all $i \in \mathcal{V}$. The proofs of these results follow by identical arguments to those presented in this work, simply by replacing $\frac{1}{n} \mathbf{1}_n$ by the vector $w$.

\section{Appendix}\label{sec:app1} In Sec.$\ref{sec:app1a}$ of this appendix we derive the consensus estimate initialization for each of the four proposed algorithms, as well as the proofs for Theorems $\ref{bmthm}$,$\ref{bmcor}$,$\ref{thm6}$,$\ref{aaacor5}$, and $\ref{aaathm5}$. In Sec.$\ref{sec:compalg}$ we define the Gossip algorithm proposed in \cite{SH06} in terms of the class of distributed algorithms $(\ref{update0}),(\ref{update1})$, and then we define the ARIS algorithm in these terms as well. In Sec.$\ref{sec:rescostapp}$ we derive the resource costs presented in Table I of Sec.$\ref{sec:rescost}$, and in Sec.$\ref{sec:examp}$ we define the ``unit-delay double cycle sequence'' as an example of a $\mbox{S$\mathcal{V}$SC}$ sequence that implies the DA, DDA, and BM algorithm all obtain average-consensus at the same instant. We present two conjectures in Sec.$\ref{sec:app1a}$, and two conjectures in Sec.$\ref{sec:compalg}$.

Through-out the appendix we denote the ``error'' of the normal consensus estimate $\hat{v}_i(t)$ as, \begin{equation}\label{newdef2} E_{\frac{1}{n} \mathbf{1}_n } (\hat{v}_i(t)) = \big| \sqrt{ \big( \hat{v}_i(t) - \frac{1}{n} \mathbf{1}_n \big) ^2 } \big| \ . \end{equation} The total reduction in normal consensus squared error resulting from the sequence $C_{[t_0,t_1]}$ is then, \begin{equation}\label{newdef0} \mathbf{E}^2 ( C_{[t_0,t_1]} ) \equiv \sum_{S^{i j} ( t^{i j}_0, t^{i j}_1 ) \in C_{[t_0,t_1]} } \mathbf{E}^2 \big( S^{i j} ( t^{i j}_0, t^{i j}_1 ) \big)  \ , \end{equation} where $\mathbf{E}^2 \big( S^{i j} ( t^{i j}_0, t^{i j}_1 ) \big)$ is defined using the normal consensus error $E_{\frac{1}{n} \mathbf{1}_n} \big( \hat{v}(t) \big)$ in $(\ref{newdef2})$, \begin{equation}\label{newdef1} \mathbf{E}^2 \big( S^{i j} ( t^{i j}_0, t^{i j}_1 ) \big) \equiv E^2_{\frac{1}{n} \mathbf{1}_n} \big( \hat{v}_i(t^{ij}_1) \big) - E^2_{\frac{1}{n} \mathbf{1}_n} \big( \hat{v}_i(t^{ij}_1(+) \big) \ .  \end{equation} 

It is convenient to use the following definition.

\begin{defn}\label{def11} Under $(\ref{a10})$, $(\ref{a0a})$, and $(\ref{a0b})$, the average-consensus problem is solved at time $t$ iff the ``network normal consensus error'' $\sum_{i=1}^n E_{\frac{1}{n} \mathbf{1}_n } \big( \hat{v}_i(t) \big)$ is zero, that is, \begin{equation}\label{acdef1} \sum_{i=1}^n E_{\frac{1}{n} \mathbf{1}_n } \big( \hat{v}_i(t) \big) = 0 \ , \end{equation} where $E_{\frac{1}{n} \mathbf{1}_n } \big( \hat{v}_i(t) \big)$ is defined in $(\ref{newdef2})$. \end{defn}

\subsection{Algorithm Convergence Proofs}\label{sec:app1a}

\begin{lem}\label{lembm0}(Consensus Estimate Initialization) Given the initial knowledge state (A4), the unique solution to both $(\ref{optllca})$ and $(\ref{aaaoptdisc})$ is, \begin{equation}\label{initsol} \hat{v}_i ( 0 ) = \frac{1}{n} e_i \ , \ \hat{s}_i ( 0 ) = \frac{1}{n} s_i(0) \ . \end{equation} \end{lem} \begin{proof} Under (A4) the update $(\ref{optllca})$ becomes, \begin{equation}\label{bminit}\begin{array}{llll} \hat{v}_i ( 0 ) = & \arg_{\mbox{$\tilde{v}$}} \ \min \ \big( \tilde{v} - \frac{1}{n} \mathbf{1}_n \big)^2 \ , \\ & \mbox{s.t. $(\ref{a10})$ holds, given $\mathcal{K}_i(0) \supseteq\{ i, n, s_i(0) \}$.} \end{array}\end{equation} Observe that the vector $e_i$ can be locally constructed at each node $i \in \mathcal{V}$ based only on the data $\{ i, n \}$. Next observe that $s_i(0) = \mathbf{S} e_i$ for all $i \in \mathcal{V}$. Given that $s_i (0) = \mathbf{S} e_i$ is known by node $i$ at $t = 0$, the set of vectors $\hat{v}_i(0)$ for which $(\ref{a10})$ holds is $\mbox{span} \{ e_i \}$, thus $(\ref{bminit})$ becomes, \begin{equation}\label{optbm2n}\begin{array}{llll} \hat{v}_i (0) & = \arg_{\mbox{$\tilde{v}$}} \ \min_{\mbox{$\tilde{v} \in \mbox{span} \{ e_i \}$}} \ \big( \tilde{v} - \frac{1}{n} \mathbf{1}_n \big) ^2 \ , \\ & = V V^{+} \frac{1}{n} \mathbf{1}_n \ \ , \ \ V = [ e_i ] \\ & = V ( V ' V)^{-1} V ' \frac{1}{n} \mathbf{1}_n  = \frac{1}{n} e_i \ . \end{array}\end{equation} If $\hat{v}_i (0) = \frac{1}{n} e_i$ then $(\ref{a10})$ implies $\hat{s}_i (0) = \frac{1}{n} s_i(0)$. Note that $\frac{1}{n} e_i$ and $\frac{1}{n} s_i(0)$ can both be initially computed at node $i$ given that $\mathcal{K}_i(0) \supseteq\{ i, n, s_i(0) \}$. The solution $(\ref{initsol})$ implies $\hat{v}_i ( 0 ) \in \mathbb{R}^n_{0,\frac{1}{n}}$ and thus $(\ref{initsol})$ is a feasible solution to $(\ref{aaaoptdisc})$. It then follows that $(\ref{initsol})$ is also the unique solution to $(\ref{aaaoptdisc})$ since the feasible space of $(\ref{aaaoptdisc})$ is contained in the feasible space of $(\ref{optllca})$, and both $(\ref{optllca})$ and $(\ref{aaaoptdisc})$ minimize the same objective function. \end{proof}

\begin{lem}\label{lembm}(BM Normal Consensus Estimate Update) Applying $(\ref{bm1})$ and $(\ref{bm2})$ to $(\ref{optllca})$ yields the BM normal consensus estimate update $(\ref{bm3})$. \end{lem} \begin{proof} If $\hat{v}_i(t^{ij}_1) = \frac{1}{n} \mathbf{1}_n$ then $(\ref{a10})$ implies $\hat{s}_i(t^{ij}_1) = \bar{s}(0)$ and thus node $i$ need not update its knowledge state regardless of the signal $S^{ij}(t^{ij}_0,t^{ij}_1 )$. If $\hat{v}_j(t^{ij}_0) = \frac{1}{n} \mathbf{1}_n$ then $(\ref{bm1})$ together with $(\ref{a10})$ implies, \begin{equation}\label{acbm1}\begin{array}{llll} S^{ij}(t^{ij}_0,t^{ij}_1 ) & \supseteq \hat{s}_j(t^{ij}_0) \ , \\ & \supseteq \bar{s}(0) \ . \end{array} \end{equation} In this case $(\ref{optllca})$ can be re-written, \begin{equation}\label{optbmac1}\begin{array}{llll}  \hat{v}_i ( t^{ij}_1(+)) = & \arg_{\mbox{$\tilde{v}$}} \ \min \ \big( \tilde{v} - \frac{1}{n} \mathbf{1}_n \big)^2 \ , \\ & \mbox{s.t. $(\ref{a10})$ holds, given $\hat{s}_j(t^{ij}_0) = \bar{s}(0)$.} \end{array}\end{equation} Since $\hat{s}_j(t^{ij}_0) = \bar{s}(0)$ is known we can let $\hat{s}_i(t^{ij}_1(+)) = \bar{s}(0)$ and thus obtain $\hat{v}_i(t^{ij}_1(+)) = \frac{1}{n} \mathbf{1}_n$ as the unique global solution to $(\ref{optbmac1})$ (note that $\frac{1}{n} \mathbf{1}_n$ can be computed since (A4) and (A7) imply $\mathcal{K}_i(t^{ij}_1) \supseteq \{ n \}$). Notice that this solution coincides with $(\ref{bm3})$. Next assume that $\hat{v}_i(t^{ij}_1) \neq \frac{1}{n} \mathbf{1}_n$ and $\hat{v}_j(t^{ij}_0) \neq \frac{1}{n} \mathbf{1}_n$. Under $(\ref{bm1})$ and $(\ref{bm2})$ we can then re-write $(\ref{optllca})$ as, \begin{equation}\label{optbm1}\begin{array}{llll} \hat{v}_i ( t^{ij}_1(+)) = & \arg_{\mbox{$\tilde{v}$}} \ \min \ \big( \tilde{v} - \frac{1}{n} \mathbf{1}_n \big)^2 , \\ & \mbox{s.t. $(\ref{a10})$ holds, given $s_\ell (0) = \mathbf{S} e_\ell$,} \\ & \mbox{for all $\{ \ell \ : \ v^{ij}_\ell ( t^{ij}_1 , t^{ij}_0) = 1 \}$.} \end{array}\end{equation} Given that $s_\ell (0) = \mathbf{S} e_\ell$ for all $\{ \ell \ : \ v^{ij}_\ell ( t^{ij}_1 , t^{ij}_0) = 1 \}$ is known, the set of vectors $\hat{v}_i(t^{ij}_1(+))$ for which $(\ref{a10})$ holds is $\mbox{span} \{ e_\ell \ : \ v^{ij}_\ell ( t^{ij}_1 , t^{ij}_0) = 1 \}$, thus $(\ref{optbm1})$ becomes, \begin{equation}\label{optbm2}\begin{array}{llll} \hat{v}_i ( t^{ij}_1(+)) & = \arg_{\mbox{$\tilde{v}$}} \ \min_{\mbox{$\tilde{v} \in \mbox{span} \{ e_\ell \ : \ v^{ij}_\ell ( t^{ij}_1 , t^{ij}_0) = 1 \}$}} \ \big( \tilde{v} - \frac{1}{n} \mathbf{1}_n \big)^2 \ , \\ & = V_{(BM)} V_{(BM)}^+ \frac{1}{n} \mathbf{1}_n \ , \\ & \ \ \ \ \ \ \ \ \ \ \ \ V_{(BM)} = [ e_\ell \ : \ v^{ij}_\ell ( t^{ij}_1 , t^{ij}_0) = 1 ] \ . \end{array}\end{equation} Let $L = v^{ij} ( t^{ij}_1 , t^{ij}_0) ' \mathbf{1}_n$ denote cardinality of the set $\{ \ell \ : \ v^{ij}_\ell ( t^{ij}_1 , t^{ij}_0) = 1 \}$. Since the columns of $V_{(BM)}$ are linearly independent, the right-hand side (RHS) of $(\ref{optbm2})$ can be computed as, $$ \begin{array}{llll} V_{(BM)} V_{(BM)}^+ \frac{1}{n} \mathbf{1}_n & = \frac{1}{n} V_{(BM)} (V_{(BM)} ' V_{(BM)})^{-1} V_{(BM)} ' \mathbf{1}_n \\ & = \frac{1}{n} V_{(BM)} ( I_{L \times L})^{-1} \mathbf{1}_L ' \ , \\ & = \frac{1}{n} v^{ij} ( t^{ij}_1 , t^{ij}_0) \ , \end{array} $$ where $I_{L \times L}$ denotes the identity matrix of dimension $L$. \end{proof}

\emph{Theorem $\ref{bmthm}$}(BM Network Convergence to Average-Consensus) \begin{proof} By the BM update $(\ref{bm3})$ when any node receives a signal $S^{ij}(t^{ij}_0,t^{ij}_1 )$, the receiving node $i$ either receives the desired average-consensus value $\bar{s}(0)$, or receives the initial consensus vector $s_\ell (0)$ of every node that has a communication path to node $j$ within the time span $[0, t^{ij}_0)$. If $C_{[0,t_1]}$ satisfies $(\ref{scdef})$ then every node has a communication path to every other node within the time span $[0, t_1]$, thus (A7) and $(\ref{bm3})$ together imply that at time $t_1(+)$ every node $i$ will compute $\hat{v}_i(t_1(+)) = \frac{1}{n} \mathbf{1}_n$. This implies $(\ref{acdef1})$ holds at $t_1(+)$ and hence by Definition $(\ref{def11})$ a network average-consensus is obtained at time $t = t_1(+)$. \end{proof}

\emph{Theorem $\ref{bmcor}$}(BM,DA,OH and DDA Necessary Conditions) \begin{proof} If $C_{[0,t_1]} \notin \mbox{S$\mathcal{V}$SC}$ then there exists a node $i \in \mathcal{V}$ that does not have a communication path to some node $j \in \mathcal{V}_{-i}$ within the time span $[0,t_1]$. At time $t_1$ the node $j$ thus cannot have any knowledge that is contained in $\mathcal{K}_i(t)$ for any $t \leq t_1$, regardless of the knowledge set update rule $(\ref{update10})$ and signal specification $(\ref{update0})$. Hence $\hat{s}_j(t_1(+))$ cannot be a function of $s_i(0)$ for an arbitrary initial consensus vector $s_i(0)$. It then follows that no distributed algorithm $(\ref{update10}),(\ref{update0})$ can imply $(\ref{acdef})$ is satisfied at time $t \leq t_1(+)$ for an arbitrary set of initial consensus vectors $\{s_i(0) \ : \ i \in \mathcal{V} \}$. \end{proof}

We now present a conjecture regarding the BM algorithm. From Theorem $\ref{bmthm}$ and $(\ref{bm1})-(\ref{bminit1})$ it follows that the BM algorithm implies the following property (P) of each knowledge set $\mathcal{K}_i(t)$: \begin{framed} \begin{itemize} \item (P): $$ C^{ij}_{[t_0(ij),t_1(ij)]} \subset C_{[0,t_1]} \Rightarrow \left\{ \begin{array}{l l} \mathcal{K}_i(t_1(+)) \supseteq \{ s_j(0) \} \ & , \ \mbox{if $\exists \ \ell$ s.t. $C^{i \ell}_{[t_0(i \ell),t_1(i \ell)]} \nsubseteq C_{[0,t_1]}$} \\ \mathcal{K}_i(t_1(+)) \supseteq \{ \bar{s}(0) \} \ & , \ \mbox{if $C^{i \ell}_{[t_0(i \ell),t_1(i \ell)]} \subset C_{[0,t_1]} \ \forall \ \ell \in \mathcal{V}_{-i}$} \end{array} \right. $$  \end{itemize} \end{framed} The condition (P) forms an equivalence class among all algorithms $\mathcal{A}$ defined by $(\ref{update10})$ and $(\ref{update0})$. From the second line in (P) it follows that any algorithm $\mathcal{A}$ that satisfies (P) will imply $(\ref{acdef})$ holds at time $t=t_1(+)$ for any communication sequence $C_{[0,t_1]}$ satisfying $(\ref{scdef})$. We now conjecture that if an algorithm $\mathcal{A}$ implies $(\ref{acdef})$ holds at time $t=t_1(+)$ for every communication sequence $C_{[0,t_1]}$ satisfying $(\ref{scdef})$, then algorithm $\mathcal{A}$ must have resource costs at least as great as any algorithm 
$\mathcal{A}$ that satisfies (P).

\begin{conj}\label{conject0} If an algorithm $\mathcal{A}$ guarantees $(\ref{acdef})$ holds at time $t=t_1(+)$ for every communication sequence $C_{[0,t_1]}$ satisfying $(\ref{scdef})$, then the algorithm $\mathcal{A}$ will require that a set with cardinality upper bounded by at least $O(nd)$ can be communicated and stored at each node. \end{conj}

The above conjecture implies that any algorithm $\mathcal{A}$ that satisfies (P) will require that a set with cardinality upper bounded by $O(nd)$ can be communicated and stored at each node, this is why searching for less costly algorithms is of importance. The problem is, less costly algorithms tend to require stronger communication conditions than the BM algorithm, and they also do not guarantee average-consensus is obtained as quickly as the BM algorithm. We note that due to the resource costs associated with the RIS and ARIS algorithms, the Conjecture $\ref{conject1}$ in Sec.$\ref{sec:ris}$ does not contradict Conjecture $\ref{conject0}$.

\emph{Proof.} (Theorem $\ref{thm6}$) Lemmas $\ref{lemkda}$ - $\ref{thmkda7c}$. \emph{Overview of Proof.} To prove Theorem $\ref{thm6}$ we initially show in Lemma $\ref{lemkda1}$ that the update $(\ref{optkda2})$ implies each normal consensus estimate $\hat{v}_i(t)$ satisfies the normalization property $(\ref{norm1})$. The Lemma $\ref{lemkda3}$ proves that each normal consensus estimate $\hat{v}_i(t)$ also satisfies the ``zero local error'' property $\hat{v}_{ii}(t) = 1/n$. Without the latter, the following lemmas would still imply convergence of all consensus estimates, but the final consensus value would not necessarily equal the average $\bar{s}(0)$ defined in $(\ref{avg0})$. The essence of the convergence proof is that the reduction in error that results from any signal will eventually vanish if $C_{[0,t_1]} \in \mbox{I$\mathcal{V}$SC}$, see Lemma $\ref{lemkda4b}$. Applying this result to the DA lower bounds on the reduction in error derived in Lemmas $\ref{lemkda4}$ and $\ref{lemkda4c}$, we can show that each normal consensus vector will necessarily converge to a common vector, see Lemma $\ref{lemkda6}$. Together with the two technical results derived in Lemmas $\ref{lemkda7b}$ and $\ref{lemkda7c}$, we then combine the triangle inequality and the ``zero local error'' property to prove that the common vector will approach $\frac{1}{n} \mathbf{1}_n$ in the $L^2$ norm as time approaches $t_1$. This implies $(\ref{acdef1})$ holds at $t_1(+)$ and hence by Definition $(\ref{def11})$ a network average-consensus is obtained at time $t = t_1(+)$.

\begin{lem}\label{lemkda}(DA Normal Consensus Estimate Update) Applying $(\ref{kda1})$ and $(\ref{kda2})$ to $(\ref{optllca})$ yields the DA normal consensus estimate update $\hat{v}_i(t^{ij}_1(+))$ defined in $(\ref{optkda2})$. \end{lem} \begin{proof} Under $(\ref{kda1})$ and $(\ref{kda2})$ we can re-write $(\ref{optllca})$ as, \begin{equation}\label{optkda1}\begin{array}{llll} \hat{v}_i ( t^{ij}_1(+)) = & \arg_{\mbox{$\tilde{v}$}} \ \min \ \big( \tilde{v} - \frac{1}{n} \mathbf{1}_n \big)^2 \ , \\ & \mbox{s.t. $(\ref{a10})$ holds, given $\hat{s}_i(t^{ij}_1) = \mathbf{S} \hat{v}_i(t^{ij}_1),$} \\ & \mbox{   $\hat{s}_j(t^{ij}_0) = \mathbf{S} \hat{v}_j(t^{ij}_0)$, and $s_i(0) = \mathbf{S} e_i$.} \end{array}\end{equation} Given that $\hat{s}_i(t^{ij}_1) = \mathbf{S} \hat{v}_i(t^{ij}_1) , \hat{s}_j(t^{ij}_0) = \mathbf{S} \hat{v}_j(t^{ij}_0)$, and $s_i(0) = \mathbf{S} e_i$ are known, the set of vectors $\hat{v}_i(t^{ij}_1(+))$ for which $(\ref{a10})$ holds is $\mbox{span} \{ \hat{v}_i(t^{ij}_1) , \hat{v}_j(t^{ij}_0) , e_i \}$, thus $(\ref{optkda1})$ can be re-written as \begin{equation}\label{optkda3} \hat{v}_i ( t^{ij}_1(+)) = \arg_{\mbox{$\tilde{v}$}} \min_{\mbox{$\tilde{v} \in \mbox{span} \{ \hat{v}_i(t^{ij}_1) , \hat{v}_j(t^{ij}_0) , e_i \}$}} \big( \tilde{v} - \frac{1}{n} \mathbf{1}_n \big)^2 . \end{equation} The update $(\ref{optkda2})$ follows immediately from $(\ref{optkda3})$. \end{proof}

\begin{lem}\label{lemkda1}(DA Consensus Estimate Normalization) Every normal consensus estimate $\hat{v}_i(t)$ satisfies \begin{equation}\label{norm1} \hat{v}_i(t)^2 = \frac{1}{n} \hat{v}_i (t) ' \mathbf{1}_n \ , \ \forall \ i \in \mathcal{V} \ , \ \forall \ t \geq 0 \ . \end{equation} \end{lem} \begin{proof} Note that $\hat{v}_i(0) = \frac{1}{n} e_i$ satisfies $(\ref{norm1})$ for each $i \in \mathcal{V}$. Next observe that under $(A7)$ the estimate $\hat{v}_i(t)$ will not change unless a signal $S^{i j} ( t^{i j}_0, t^{i j}_1 )$ is received at node $i$. If a signal is received then by Lemma $\ref{lemkda}$ the estimate $\hat{v}_i(t)$ is updated to the unique solution of $(\ref{optkda3})$. Thus to finish the proof it suffices to show that if a vector $v \in \mathbb{R}^n$ does not satisfy $(\ref{norm1})$ then the vector $v$ is not the solution to $(\ref{optkda3})$. To prove this we show that if $(\ref{norm1})$ does not hold then the vector $w$ defined, $$ w = v \big( v ' \mathbf{1}_n \big) / \big( n v^2 \big) \ , $$ will satisfy the inequality \begin{equation}\label{wnorm1} \big( w - \frac{1}{n} \mathbf{1}_n \big) ^2 < \big( v - \frac{1}{n} \mathbf{1}_n \big) ^2 \ . \end{equation} Notice that since $w$ is contained in $\mbox{span}(v)$, the inequality $(\ref{wnorm1})$ implies that $v$ is not the solution to $(\ref{optkda3})$. Next observe that if a vector $v$ does not satisfy $(\ref{norm1})$ then, \begin{equation}\label{wnorm4} \big( v^2 - \frac{1}{n} v ' \mathbf{1}_n \big)^2 > 0 \ . \end{equation} Expanding $(\ref{wnorm4})$ yields, \begin{equation}\label{wnorm5} ( v^2 )^2 - 2 \frac{1}{n} v ' \mathbf{1}_n v^2 + \big( \frac{1}{n} v ' \mathbf{1}_n \big)^2 > 0 \ . \end{equation} Re-arranging $(\ref{wnorm5})$ then implies $(\ref{wnorm1})$, $$ \begin{array}{llll} & ( v^2 )^2 - 2 \frac{1}{n} v ' \mathbf{1}_n v^2 > - \big( \frac{1}{n} v ' \mathbf{1}_n \big) ^2 \ , \\ & v^2 - 2 \frac{1}{n} v ' \mathbf{1}_n + \frac{1}{n} > \big( \frac{v ' \mathbf{1}_n }{n v^2 } \big) ^2 v^2 - 2 \frac{( v ' \mathbf{1}_n )^2}{n^2 v^2 } + \frac{1}{n} \ , \\ & \big( v - \frac{1}{n} \mathbf{1}_n \big)^2 > \big( v \frac{v ' \mathbf{1}_n }{n v^2} - \frac{1}{n} \mathbf{1}_n \big)^2 \ . \end{array} $$ \end{proof}

\begin{lem}\label{lemkda2}(DA Non-Decreasing Normal Consensus Estimate Magnitude) Each magnitude $\hat{v}_i(t)^2$ is a non-decreasing function of $t \geq 0$ for all $i \in \mathcal{V}$.\end{lem} \begin{proof} Note that under (A7) the estimate $\hat{v}_i(t)$ will not change unless a signal is received at node $i$. If a signal $S^{i j} ( t^{i j}_0, t^{i j}_1 )$ is received then the DA update problem $(\ref{optkda3})$ implies the update $\hat{v}_i(t^{ij}_1(+))$ must satisfy, \begin{equation}\label{gkda1what}\begin{array}{llll} & \big( \hat{v}_i(t^{ij}_1(+)) - \frac{1}{n} \mathbf{1}_n \big)^2 \leq \big( w - \frac{1}{n} \mathbf{1}_n \big)^2 \ , \\ & \ \ \ \forall \ w \in \mbox{span $\{ \hat{v}_i(t^{ij}_1) , \hat{v}_j(t^{ij}_0) , e_i \}.$} \end{array}\end{equation} Since $\{ \hat{v}_i(t^{ij}_1) , \hat{v}_j(t^{ij}_0) \} \in \mbox{span $\{ \hat{v}_i(t^{ij}_1) , \hat{v}_j(t^{ij}_0) , e_i \}$}$ the inequality $(\ref{gkda1what})$ implies, \begin{equation}\label{gkda2}\begin{array}{llll} & \big( \hat{v}_i(t^{ij}_1(+)) - \frac{1}{n} \mathbf{1}_n \big)^2 \leq \mbox{min} \{ \big( \hat{v}_i(t^{ij}_1) - \frac{1}{n} \mathbf{1}_n \big)^2 \ , \\ & \ \ \ \ \ \ \ \ \ \ \ \ \ \ \ \ \ \ \ \ \ \ \ \ \ \ \ \ \ \ \ \ \ \ \ \big( \hat{v}_j(t^{ij}_0) - \frac{1}{n} \mathbf{1}_n \big)^2 \}. \end{array}\end{equation} Next observe that if a vector $v \in \mathbb{R}^n$ satisfies $(\ref{norm1})$ then, \begin{equation}\label{gkda0}\begin{array}{llll} \big( v - \frac{1}{n} \mathbf{1}_n \big)^2 & = v^2 - 2 \frac{1}{n} v ' \mathbf{1}_n + \frac{1}{n} \ , \\ & = \frac{1}{n} - v^2 \ .\end{array}\end{equation} Due to Lemma $\ref{lemkda1}$, all normal consensus estimates satisfy $(\ref{norm1})$, thus we can apply $(\ref{gkda0})$ to $(\ref{gkda2})$ and obtain, \begin{equation}\label{gkda3} \frac{1}{n} - \hat{v}_i(t^{ij}_1(+))^2 \leq \mbox{min$\{ \frac{1}{n} - \hat{v}_i(t^{ij}_1) ^2 \ , \ \frac{1}{n} - \hat{v}_j(t^{ij}_0)^2 \}$. } \end{equation} Subtracting both sides of $(\ref{gkda3})$ from $\frac{1}{n}$ then yields, \begin{equation}\label{gkda4}\begin{array}{llll} \hat{v}_i(t^{ij}_1(+))^2 & \geq \mbox{max$\{ \hat{v}_i(t^{ij}_1) ^2 \ , \ \hat{v}_j(t^{ij}_0) ^2 \}$,} \\ & \geq \hat{v}_i(t^{ij}_1) ^2 \ , \end{array}\end{equation} thus each magnitude $\hat{v}_i(t)^2$ is a non-decreasing function of $t \geq 0$ for all $i \in \mathcal{V}$. \end{proof}

\begin{lem}\label{lemkdanew1}(Equality of Normalized Linear Dependent Vectors) If two linearly dependent vectors $\hat{v}_i$, $\hat{v}_j \in \mathbb{R}^n$ both satisfy $(\ref{norm1})$, then $\hat{v}_j ^2 = \hat{v}_i ' \hat{v}_j = \hat{v}_i ^2$. \end{lem} \begin{proof} If $\hat{v}_i$ and $\hat{v}_j$ are linearly dependent then there exists some $k \neq 0$ such that $\hat{v}_i = k \hat{v}_j$, thus if both vectors also satisfy $(\ref{norm1})$ then, \begin{equation}\label{eqnormlin1}\begin{array}{llll} & \frac{1}{n} \hat{v}_i ' \mathbf{1}_n = k \frac{1}{n} \hat{v}_j ' \mathbf{1}_n \ \Rightarrow \ k = \big( \hat{v}_j ' \mathbf{1}_n \big) ^{-1} \big( \hat{v}_i ' \mathbf{1}_n \big) \\ & \ \ \ \ \ \ \ \ \ \ \ \ \ \ \ \ \ \ \ \ \ \ \ \ \ \ \ \ \ \ = \hat{v}_i ^2 / \hat{v}_j ^2 , \\ & \ \ \ \ \ \ \ \ \ \hat{v}_i ^2 = k^2 \hat{v}_j ^2 \ \Rightarrow \ k ^2 = \hat{v}_i ^2 / \hat{v}_j ^2 \ . \end{array}\end{equation} Combining the RHS of the first and second lines in $(\ref{eqnormlin1})$ implies $k^2 = k$ and thus $k = \pm 1$ since $k \neq 0$. We then obtain $k=1$ since $\hat{v}_i ^2 / \hat{v}_j ^2 > 0$, thus $\hat{v}_i = \hat{v}_j$ and the result follows. \end{proof}

\begin{lem}\label{lemkda3}(DA Local Zero Error Property) Every normal consensus estimate $\hat{v}_i(t)$ satisfies \begin{equation}\label{pp1} \hat{v}_{ii} (t) = \frac{1}{n} \ , \ \forall \ i \in \mathcal{V} \ , \ \forall \ t \geq 0 \ . \end{equation} \end{lem} \begin{proof} By Lemma $\ref{lembm0}$, $\hat{v}_{ii} (0) = \frac{1}{n}$ for each $i \in \mathcal{V}$. Next observe that under $(A7)$ the estimate $\hat{v}_i(t)$ will not change unless a signal $S^{i j} ( t^{i j}_0, t^{i j}_1 )$ is received at node $i$. If a signal is received then $\hat{v}_i(t)$ is updated to the solution of $(\ref{optkda3})$. We now show that under the assumption $(\ref{pp1})$ the solution $\hat{v}_i ( t^{ij}_1(+))$ to $(\ref{optkda3})$ will imply $(\ref{pp1})$ for every set vectors $\{ \hat{v}_i ( t^{ij}_1), \hat{v}_j ( t^{ij}_0) , e_i\}$. First assume $\hat{v}_i ( t^{ij}_1)$, $\hat{v}_j ( t^{ij}_0)$, and $e_i$ are linearly dependent. In this case the update problem $(\ref{optkda3})$ reduces to the RHS of $(\ref{optbm2n})$ and thus implies $(\ref{pp1})$. Next assume that any two vectors in the set $\{ \hat{v}_i ( t^{ij}_1), \hat{v}_j ( t^{ij}_0) , e_i \}$ are linearly dependent. In this case $(\ref{optkda3})$ reduces to \begin{equation}\label{optkda2a}\begin{array}{llll} \hat{v}_i ( t^{ij}_1(+)) & = \arg_{\mbox{$\tilde{v}$}} \ \min_{\mbox{$\tilde{v} \in \mbox{span} \{ e_i , v \}$}} \ \big( \tilde{v} - \frac{1}{n} \mathbf{1}_n \big)^2 , \\ & = \arg_{\mbox{$\tilde{v} = a \frac{1}{n} e_i + b v$}} \min_{\mbox{$(a , b)$}} \ \big( a \frac{1}{n} e_i + b v - \frac{1}{n} \mathbf{1}_n \big)^2 \ , \end{array}\end{equation} where $v \in \{ \hat{v}_i ( t^{ij}_1), \hat{v}_j ( t^{ij}_0) \}$ is linearly independent of $e_i$. The objective function in $(\ref{optkda2a})$ is \begin{equation}\label{optfunct}\begin{array}{llll} f(a,b) & = \big( a \frac{1}{n} e_i + b v - \frac{1}{n} \mathbf{1}_n \big)^2 \ , \\ & = a^2 (\frac{1}{n} e_i) ^2 + b^2 v^2 + 2 a b \frac{1}{n} e_i ' v \\ & \ \ \ \ - 2 a \frac{1}{n} e_i ' \frac{1}{n} \mathbf{1}_n - 2 b v ' \frac{1}{n} \mathbf{1}_n + \frac{1}{n} \ . \end{array}\end{equation} The Lemma $\ref{lemkda1}$ implies that $v$ satisfies $(\ref{norm1})$. Note also that $\frac{1}{n} e_i$ satisfies $(\ref{norm1})$, thus the objective function $(\ref{optfunct})$ can be simplified, $$ f(a,b) = (a^2 - 2a) \big( \frac{1}{n} e_i \big) ^2 + ( b^2 - 2 b ) v^2 + 2 a b \frac{1}{n} e_i ' v + \frac{1}{n} \ . $$ The first-order partial derivatives of $f(a,b)$ are, \begin{equation}\label{optfunct2}\begin{array}{llll} & \frac{\partial f(a,b)}{\partial a} = 2 (a - 1) \big( \frac{1}{n} e_i \big) ^2 + 2 b \frac{1}{n} e_i ' v \ , \\ & \frac{\partial f(a,b)}{\partial b} = 2 (b - 1) v^2 + 2 a \frac{1}{n} e_i ' v \ . \end{array}\end{equation} The second-order partial derivatives of $f(a,b)$ are, $$ \begin{array}{llll} & \frac{\partial^2 f(a,b)}{\partial a^2} = 2 \big( \frac{1}{n} e_i \big) ^2 \ , \\ & \frac{\partial^2 f(a,b)}{\partial b^2} = 2 v ^2 \ , \\ & \frac{\partial^2 f(a,b)}{\partial a \partial b} = \frac{\partial^2 f(a,b)}{\partial b \partial a} = 2 \frac{1}{n} e_i ' v \ . \end{array} $$ The determinant of the Hessian matrix of $f(a,b)$ is thus, \begin{equation}\label{hessian} | H \big(f(a,b) \big) | = 2 \big( (\frac{1}{n} e_i ) ^2 v^2 - (\frac{1}{n} e_i ' v)^2 \big) \ . \end{equation} Since $\frac{1}{n} e_i$ and $v$ are linearly independent, the determinant $(\ref{hessian})$ is strictly positive by the Cauchy-Schwartz inequality. This implies the Hessian matrix $H \big(f(a,b) \big)$ is positive-definite, thus setting the RHS of $(\ref{optfunct2})$ to zero and solving for $(a,b)$ yields the unique optimal values $(\hat{a}, \hat{b})$, \begin{equation}\label{optfunct2sol}\hat{a} = \frac{v ^2 \frac{1}{n} e_i ' ( \frac{1}{n} e_i - v) }{( \frac{1}{n} e_i) ^2 v^2 - ( \frac{1}{n} e_i ' v )^2} \ , \ \ \hat{b} = \frac{ ( \frac{1}{n} e_i ) ^2 v ' (v - \frac{1}{n} e_i) }{ (\frac{1}{n} e_i) ^2 v^2 - ( \frac{1}{n} e_i ' v )^2} \ . \end{equation} From $(\ref{optfunct2sol})$ the unique solution to $(\ref{optkda2a})$ is thus obtained, \begin{equation}\label{fin1}\begin{array}{llll} & \hat{v}_i ( t^{ij}_1(+)) = \frac{v ^2 \frac{1}{n} e_i ' ( \frac{1}{n} e_i - v) }{ ( \frac{1}{n} e_i) ^2 v^2 - ( \frac{1}{n} e_i ' v )^2} \frac{1}{n} e_i \\ & \ \ \ \ \ \ \ \ \ \ \ \ \ \ \ \ \ \ \ + \frac{ ( \frac{1}{n} e_i ) ^2 v ' (v - \frac{1}{n} e_i) }{ ( \frac{1}{n} e_i ) ^2 v^2 - ( \frac{1}{n} e_i ' v )^2} v \ . \end{array}\end{equation} Based on $(\ref{fin1})$ the element $\hat{v}_{ii} ( t^{ij}_1(+))$ can be expressed, \begin{equation}\label{fin1a}\begin{array}{llll} \hat{v}_{ii} ( t^{ij}_1(+)) & = \frac{v ^2 \frac{1}{n} e_i ' ( \frac{1}{n} e_i - v) }{(\frac{1}{n} e_i ) ^2 v^2 - ( \frac{1}{n} e_i ' v )^2} \frac{1}{n} \\ & \ \ \ \ \ \ \ + \frac{ ( \frac{1}{n} e_i) ^2 v ' (v - \frac{1}{n} e_i) }{ (\frac{1}{n} e_i) ^2 v^2 - ( \frac{1}{n} e_i ' v )^2} v_i \ , \\ & = \frac{v ^2 ( \frac{1}{n^2} - \frac{1}{n} v_i ) \frac{1}{n} + \frac{1}{n^2} (v^2 - \frac{1}{n}v_i ) v_i }{ \frac{1}{n^2} v^2 - \frac{1}{n^2} v_i^2} \ , \\ & = \frac{\frac{1}{n} v^2 - v_i v^2 + v_i v^2 - \frac{1}{n} v_i^2 }{v^2 - v_i^2} = \frac{1}{n} . \end{array}\end{equation} Note that the last equality in $(\ref{fin1a})$ follows since Lemma $\ref{lemkdanew1}$ implies $v^2 \neq v_i^2$ under the assumption that $v$ satisfies $(\ref{norm1})$ and is linearly independent from $\frac{1}{n} e_i$. Next assume that $\hat{v}_i ( t^{ij}_1)$, $\hat{v}_j ( t^{ij}_0)$ and $e_i$ are linearly independent. In this case the solution $(\ref{optkda2})$ can be expressed, $$ \begin{array}{llll} & \hat{v}_i ( t^{ij}_1(+)) = V_{(DA)} (V_{(DA)} ' V_{(DA)})^{-1} V_{(DA)}' \frac{1}{n} \mathbf{1}_n , \\ & \ \ \ \ \ V_{(DA)} = [ \hat{v}_i ( t^{ij}_1) , \hat{v}_j ( t^{ij}_0) , \frac{1}{n} e_i ] \ . \end{array} $$ For notational convenience we denote $\hat{v}_i = \hat{v}_i ( t^{ij}_1)$ and $\hat{v}_j = \hat{v}_j ( t^{ij}_0)$. Note that $(V_{(DA)} ' V_{(DA)})$ has the inverse $(\ref{matinv1})$ below, \begin{equation}\label{matinv1}\begin{array}{llll} (V_{(DA)} ' V_{(DA)})^{-1} & = \left( \begin{array}{ccc} \hat{v}_i^2 & \hat{v}_i ' \hat{v}_j & \frac{1}{n} \hat{v}_{ii} \\ \hat{v}_i ' \hat{v}_j & \hat{v}_j^2 & \frac{1}{n} \hat{v}_{ji} \\ \frac{1}{n} \hat{v}_{ii} & \frac{1}{n} \hat{v}_{ji} & \frac{1}{n^2} \end{array} \right) ^{-1} \ , \\ & = \frac{1}{|V_{(DA)} ' V_{(DA)}|} \left( \begin{array}{ccc} \frac{1}{n^2} \big( \hat{v}_j^2 - \hat{v}_{ji}^2 \big) & \frac{1}{n^2} \big( \hat{v}_{ii} \hat{v}_{ji} - \hat{v}_i ' \hat{v}_j \big) & \frac{1}{n} \big( \hat{v}_{ji} \hat{v}_i ' \hat{v}_j - \hat{v}_j^2 \hat{v}_{ii} \big) \\ \frac{1}{n^2} \big( \hat{v}_{ii} \hat{v}_{ji} - \hat{v}_i ' \hat{v}_j \big) & \frac{1}{n^2} \big( \hat{v}_i^2 - \hat{v}_{ii}^2 \big) & \frac{1}{n} \big( \hat{v}_{ii} \hat{v}_i ' \hat{v}_j - \hat{v}_i^2 \hat{v}_{ji} \big) \\ \frac{1}{n} \big( \hat{v}_{ji} \hat{v}_i ' \hat{v}_j - \hat{v}_j^2 \hat{v}_{ii} \big) & \frac{1}{n} \big( \hat{v}_{ii} \hat{v}_i ' \hat{v}_j - \hat{v}_i^2 \hat{v}_{ji} \big) & \big( \hat{v}_i^2 \hat{v}_j^2 - ( \hat{v}_i ' \hat{v}_j )^2 \big) \end{array} \right) \end{array}\end{equation} where the determinant $|V_{(DA)}'V_{(DA)}|$ can be computed as, $$ \begin{array}{llll} & |V_{(DA)}'V_{(DA)}| = \frac{1}{n^2} ( \hat{v}_i^2 \hat{v}_j^2 - ( \hat{v}_i ' \hat{v}_j )^2 + 2 \hat{v}_{ii} \hat{v}_{ji} \hat{v}_i ' \hat{v}_j \\ & \ \ \ \ \ \ \ \ \ \ \ \ \ \ \ \ \ \ \ \ \ \ \ \ \ \ \ \ - \hat{v}_{ii}^2 \hat{v}_j^2 - \hat{v}_i^2 \hat{v}_{ji}^2 ) \ . \end{array} $$ Right-multiplying $(\ref{matinv1})$ by $V_{(DA)}' \frac{1}{n} \mathbf{1}_n$ and then left-multiplying by the first row of $V_{(DA)}$ yields $\hat{v}_{ii} ( t^{ij}_1(+))$, \begin{equation}\label{pp1condproof}\begin{array}{llll} & \hat{v}_{ii} ( t^{ij}_1(+)) = \frac{1}{|V_{(DA)}' V_{(DA)}|} \frac{1}{n^2} \big( \hat{v}_i^2 \hat{v}_j^2 \hat{v}_{ii} - \hat{v}_i^2 \hat{v}_{ji}^2 \hat{v}_{ii} \\ & \ \ \ + \hat{v}_{ii}^2 \hat{v}_{ji} \hat{v}_j^2 - \hat{v}_{ii} \hat{v}_i ' \hat{v}_j \hat{v}_j^2 + \frac{1}{n} \hat{v}_{ji} \hat{v}_i ' \hat{v}_j \hat{v}_{ii} - \frac{1}{n} \hat{v}_{ii}^2 \hat{v}_j^2 \\ & \ \ \ + \hat{v}_i^2 \hat{v}_{ii} \hat{v}_{ji}^2 - \hat{v}_i ' \hat{v}_j \hat{v}_i^2 \hat{v}_{ji} + \hat{v}_i^2 \hat{v}_j^2 \hat{v}_{ji} - \hat{v}_{ii}^2 \hat{v}_j^2 \hat{v}_{ji} \\ & \ \ \ + \frac{1}{n} \hat{v}_{ii} \hat{v}_{ji} \hat{v}_i ' \hat{v}_j - \frac{1}{n} \hat{v}_i^2 \hat{v}_{ji}^2 + \hat{v}_i^2 \hat{v}_i ' \hat{v}_j \hat{v}_{ji} - \hat{v}_i^2 \hat{v}_j^2 \hat{v}_{ii} \\ & \ \ \ + \hat{v}_{ii} \hat{v}_i ' \hat{v}_j \hat{v}_j^2 - \hat{v}_i^2 \hat{v}_j^2 \hat{v}_{ji} + \frac{1}{n} \hat{v}_i^2 \hat{v}_j^2 - \frac{1}{n} (\hat{v}_i ' \hat{v}_j)^2 \big) \ , \\ & \ \ \ \ \ = \frac{1}{|V_{(DA)}'V_{(DA)}|} \frac{1}{n^3} \big( \hat{v}_i^2 \hat{v}_j^2 - ( \hat{v}_i ' \hat{v}_j )^2 \\ & \ \ \ \ \ \ \ \ \ \ \ \ \ \ \ + 2 \hat{v}_{ii} \hat{v}_{ji} \hat{v}_i ' \hat{v}_j - \hat{v}_{ii}^2 \hat{v}_j^2 - \hat{v}_i^2 \hat{v}_{ji}^2 \big) \ , \\ & \ \ \ \ \ = \frac{1}{n} \ . \end{array}\end{equation} The last line in $(\ref{pp1condproof})$ follows since $|V_{(DA)}'V_{(DA)}|$ is non-zero if $e_i, \hat{v}_i$ and $\hat{v}_j$ are linearly independent. \end{proof}

\begin{lem}\label{lemkda4}(DA Lower Bound On Signal Reduction in Error) Upon reception of any signal $S^{i j} ( t^{i j}_0, t^{i j}_1 )$ the decrease in the updated normal consensus squared error has the following lower bound, \begin{equation}\label{lowbound}\begin{array}{llll} & E_{\frac{1}{n} \mathbf{1}_n } ^2 \big(\hat{v}_i(t^{ij}_1) \big) - E_{\frac{1}{n} \mathbf{1}_n } ^2 \big( \hat{v}_i(t^{ij}_1(+)) \big) \geq \\ & \ \ \ \ \ \ \ \ \ \ \ \ \mbox{max} \{ \hat{v}_j(t^{ij}_0) ^2 - \hat{v}_i(t^{ij}_1) ^2 \ , \\ & \ \ \ \ \ \ \ \ \ \ \ \ \ \ \ \ \ \ \ n \big( \hat{v}_j(t^{ij}_0) ' ( \hat{v}_j(t^{ij}_0) - \hat{v}_i(t^{ij}_1) ) \big)^2 \}. \end{array}\end{equation} \end{lem} \begin{proof} By Lemma $\ref{lemkda1}$ we can apply $(\ref{gkda0})$ to the left-hand side (LHS) of $(\ref{lowbound})$, \begin{equation}\label{lowbound1} E_{\frac{1}{n} \mathbf{1}_n } ^2 \big( \hat{v}_i(t^{ij}_1) \big) - E_{\frac{1}{n} \mathbf{1}_n } ^2 \big( \hat{v}_i(t^{ij}_1(+)) \big) = \hat{v}_i(t^{ij}_1(+))^2 - \hat{v}_i(t^{ij}_1)^2 . \end{equation} Applying the first line of $(\ref{gkda4})$ to the RHS of $(\ref{lowbound1})$ then yields, \begin{equation}\label{lowbound2}\begin{array}{llll} & E_{\frac{1}{n} \mathbf{1}_n } ^2 \big(\hat{v}_i(t^{ij}_1) \big) - E_{\frac{1}{n} \mathbf{1}_n } ^2 \big( \hat{v}_i(t^{ij}_1(+)) \big) \\ & \ \ \ \ \ \geq \max \{ \hat{v}_i(t^{ij}_1) ^2 , \hat{v}_j(t^{ij}_0)^2 \} - \hat{v}_i(t^{ij}_1)^2 \ , \\ & \ \ \ \ \ \geq \max \{ 0 \ , \ \hat{v}_j(t^{ij}_0) ^2 - \hat{v}_i(t^{ij}_1) ^2 \} \ , \\ & \ \ \ \ \ \geq \hat{v}_j(t^{ij}_0) ^2 - \hat{v}_i(t^{ij}_1) ^2 \ . \end{array}\end{equation} Next observe that for any two vectors $\hat{v}_i,\hat{v}_j$, $$ \mbox{span $\{\hat{v}_i, \hat{v}_j \}$} \subseteq \mbox{span $\{ \hat{v}_i, \hat{v}_j , e_i \}$} \ . $$ It thus follows that, \begin{equation}\label{optkda3new}\begin{array}{llll} & E_{\frac{1}{n} \mathbf{1}_n} \big( \arg_{\mbox{$\tilde{v}$}} \ \min_{\mbox{$\tilde{v} \in \mbox{span} \{ \hat{v}_i , \hat{v}_j , e_i \}$}} \ ( \tilde{v} - \frac{1}{n} \mathbf{1}_n )^2 \ \big) \\ & \leq E_{\frac{1}{n} \mathbf{1}_n} \big( \arg_{\mbox{$\tilde{v}$}} \ \min_{\mbox{$\tilde{v} \in \mbox{span} \{ \hat{v}_i , \hat{v}_j \}$}} \ ( \tilde{v} - \frac{1}{n} \mathbf{1}_n )^2 \ \big) . \end{array}\end{equation} Let $\hat{v}_i = \hat{v}_i ( t^{ij}_1)$, $\hat{v}_j = \hat{v}_j ( t^{ij}_0)$ and $\hat{v}_i ( t^{ij}_1(+)) = \hat{v}_i (+)$ for notational convenience. From $(\ref{optkda3new})$ we have, \begin{equation}\label{lowbound3}\begin{array}{llll} & E_{\frac{1}{n} \mathbf{1}_n } ^2 (\hat{v}_i) - E_{\frac{1}{n} \mathbf{1}_n } ^2 ( \hat{v}_i (+)) \\ & \geq \big( \arg_{\mbox{$\tilde{v}$}} \ \min_{\mbox{$\tilde{v} \in \mbox{span} \{ \hat{v}_i , \hat{v}_j \}$}} \ ( \tilde{v} - \frac{1}{n} \mathbf{1}_n )^2 \big)^2 - \hat{v}_i^2 \ , \\ & = \hat{w} ^2 - \hat{v}_i ^2 \ , \\ & = \frac{1}{n} \big( \hat{w} ' \mathbf{1}_n - \hat{v}_i ' \mathbf{1}_n \big) \ , \end{array}\end{equation} where the last line holds due to Lemma $\ref{lemkda1}$, and $\hat{w}$ is defined, \begin{equation}\label{what}\begin{array}{llll} \hat{w} & = \arg_{\mbox{$\tilde{v}$}} \ \min_{\mbox{$\tilde{v} \in \mbox{span} \{ \hat{v}_i , \hat{v}_j \}$}} \ ( \tilde{v} - \frac{1}{n} \mathbf{1}_n )^2 \ , \\ & = \arg_{\mbox{$\tilde{v} = a \hat{v}_i + b \hat{v}_j$}} \ \min_{\mbox{$(a,b)$}} \ ( a \hat{v}_i + b \hat{v}_j - \frac{1}{n} \mathbf{1}_n )^2 \ . \end{array}\end{equation} Note that Lemma $\ref{lemkda2}$ together with the initialization $(\ref{initsol})$ implies $\hat{v}_i$ and $\hat{v}_j$ are non-zero. Next assume that $\hat{v}_i$ and $\hat{v}_j$ are linearly dependent. In this case $(\ref{what})$ reduces to, \begin{equation}\label{what2}\begin{array}{llll} \hat{w} & = \arg_{\mbox{$\tilde{v}$}} \ \min_{\mbox{$\tilde{v} \in \mbox{span} \{ \hat{v}_i \}$}} \ ( \tilde{v} - \frac{1}{n} \mathbf{1}_n )^2 \ , \\ & = V (V ' V)^{-1} V' \frac{1}{n} \mathbf{1}_n \ , \ V = [\hat{v}_i] \ , \\ & = \hat{v}_i \frac{\hat{v}_i ' \mathbf{1}_n}{n \hat{v}_i^2} \ , \\ & = \hat{v}_i \ , \end{array}\end{equation} where the last line follows due to Lemma $\ref{lemkda1}$. Applying $(\ref{what2})$ to $(\ref{lowbound3})$ implies, \begin{equation}\label{lowbound3a}\begin{array}{llll} E_{\frac{1}{n} \mathbf{1}_n } ^2 (\hat{v}_i) - E_{\frac{1}{n} \mathbf{1}_n } ^2 ( \hat{v}_i (+)) & \geq \hat{w} ^2 - \hat{v}_i ^2 \ , \\ & = \hat{v}_i ^2 - \hat{v}_i ^2 \\ & = 0 \ , \\ & = n \big( \hat{v}_j ^2 - \hat{v}_i ' \hat{v}_j \big)^2 \end{array}\end{equation} where the last line follows by Lemma $\ref{lemkdanew1}$ since Lemma $\ref{lemkda1}$ implies both $\hat{v}_i$ and $\hat{v}_j$ satisfy $(\ref{norm1})$, and we are assuming $\hat{v}_i$ and $\hat{v}_j$ are linearly dependent.

Next assume $\hat{v}_i$ and $\hat{v}_j$ are linearly independent. In this case $(\ref{what})$ can be solved analogous to $(\ref{fin1})$ based on the optimization problem $(\ref{optkda2a})$, \begin{equation}\label{what3}\begin{array}{llll} \hat{w} & = \arg_{\mbox{$\tilde{v} = a \hat{v}_i + b \hat{v}_j$}} \ \min_{\mbox{$(a,b)$}} \ ( a \hat{v}_i + b \hat{v}_j - \frac{1}{n} \mathbf{1}_n )^2 \ , \\ & = \frac{\hat{v}_j ^2 \hat{v}_i ' ( \hat{v}_i - \hat{v}_j) }{ \hat{v}_i^2 \hat{v}_j^2 - ( \hat{v}_i ' \hat{v}_j )^2} \hat{v}_i + \frac{ \hat{v}_i ^2 \hat{v}_j ' (\hat{v}_j - \hat{v}_i) }{ \hat{v}_i ^2 \hat{v}_j^2 - ( \hat{v}_i ' \hat{v}_j )^2} \hat{v}_j \ . \end{array}\end{equation} Substituting the second line of $(\ref{what3})$ for $\hat{w}$ in the third line of $(\ref{lowbound3})$ then yields, \begin{equation}\label{lowbound4}\begin{array}{llll} & \frac{1}{n} \big( \hat{w} ' \mathbf{1}_n - \hat{v}_i ' \mathbf{1}_n \big) \\ & \ \ \ = \frac{1}{n} \big( \frac{\hat{v}_j ^2 \hat{v}_i ' ( \hat{v}_i - \hat{v}_j) \hat{v}_i ' \mathbf{1}_n + \hat{v}_i ^2 \hat{v}_j ' (\hat{v}_j - \hat{v}_i) \hat{v}_j ' \mathbf{1}_n }{ \hat{v}_i ^2 \hat{v}_j^2 - ( \hat{v}_i ' \hat{v}_j )^2} - \hat{v}_i ' \mathbf{1}_n \big) \\ & \ \ \ = \frac{\hat{v}_j ^2 \hat{v}_i ^2 \hat{v}_i ' ( \hat{v}_i - \hat{v}_j) + \hat{v}_i ^2 \hat{v}_j^2 \hat{v}_j ' (\hat{v}_j - \hat{v}_i) - \big( \hat{v}_i ^2 \hat{v}_j^2 - ( \hat{v}_i ' \hat{v}_j )^2 \big) \hat{v}_i ^2 }{ \hat{v}_i ^2 \hat{v}_j^2 - ( \hat{v}_i ' \hat{v}_j )^2} \\ & \ \ \ = \frac{\hat{v}_i ^2}{ \hat{v}_i ^2 \hat{v}_j^2 - ( \hat{v}_i ' \hat{v}_j )^2} \big( \hat{v}_j ^2 ( \hat{v}_i^2 - \hat{v}_i ' \hat{v}_j) + \hat{v}_j^2 ( \hat{v}_j^2 - \hat{v}_i ' \hat{v}_j) \\ & \ \ \ \ \ \ \ \ \ \ \ \ \ \ \ \ \ \ \ \ \ \ \ \ \ \ \ \ \ \ \ \ \ \ \ - \hat{v}_i ^2 \hat{v}_j^2 + ( \hat{v}_i ' \hat{v}_j )^2 \big) , \\ & \ \ \ \geq \big( \hat{v}_j ^4 - 2 \hat{v}_i ' \hat{v}_j \hat{v}_j ^2 + ( \hat{v}_i ' \hat{v}_j )^2 \big) \frac{\hat{v}_i ^2}{ \hat{v}_i ^2 \hat{v}_j^2 } \\ & \ \ \ = \big( \hat{v}_j ^2 - \hat{v}_i ' \hat{v}_j \big)^2 \frac{1}{ \hat{v}_j^2 } \\ & \ \ \ \geq n \big( \hat{v}_j ^2 - \hat{v}_i ' \hat{v}_j \big)^2 \ , \end{array}\end{equation} where the last line follows since $E^2_{\frac{1}{n} \mathbf{1}_n}( \hat{v}_j) \geq 0$ implies $\hat{v}_j^2 \leq \frac{1}{n}$. Combining $(\ref{lowbound3})$ and $(\ref{lowbound4})$ implies, \begin{equation}\label{lowbound5} E_{\frac{1}{n} \mathbf{1}_n } ^2 (\hat{v}_i) - E_{\frac{1}{n} \mathbf{1}_n } ^2 ( \hat{v}_i (+) ) \geq n \big( \hat{v}_j ^2 - \hat{v}_i ' \hat{v}_j \big)^2 \ . \end{equation} Together $(\ref{lowbound2})$, $(\ref{lowbound3a})$ and $(\ref{lowbound5})$ imply $(\ref{lowbound})$. \end{proof}

\begin{lem}\label{lemkda4c}(DA Non-Decreasing Magnitude for any Communication Path) Any communication path $C^{ij}_{[t_0(ij),t_1(ij)]} \in C_{[0,t_1]}$ implies, \begin{equation}\label{lbcp} \hat{v}_i(t_1(ij)(+)) ^2 \geq \hat{v}_j(t_0(ij)) ^2 \ . \end{equation} \end{lem} \begin{proof} The first line in $(\ref{gkda4})$ implies that for any signal $S^{i j} ( t^{i j}_0, t^{i j}_1 )$, $$ \hat{v}_i(t^{ij}_1(+)) ^2 \geq \hat{v}_j(t^{ij}_0) ^2 \ . $$ Thus for any communication path $C^{ij}_{[t_0(ij),t_1(ij)]}$ defined as in $(\ref{pairs})$, $$ \begin{array}{llll} & \hat{v}_i (t^{i \ell_{k(ij)}}_1(+))^2 \geq \hat{v}_{\ell_{k(ij)}}(t^{i \ell_{k(ij)}}_0) ^2 \\ & \geq \hat{v}_{\ell_{k(ij)}}(t^{\ell_{k(ij)} \ell_{k(ij)-1}}_1(+)) ^2 \\ & \geq \hat{v}_{\ell_{k(ij)-1}}(t^{\ell_{k(ij)} \ell_{k(ij)-1}}_0) ^2 \\ & \geq \hat{v}_{\ell_{k(ij)-1}}(t^{\ell_{k(ij)-1} \ell_{k(ij)-2}}_1(+)) ^2 \geq \cdots \\ & \geq \hat{v}_{\ell_1}(t^{\ell_2 \ell_1}_0) ^2 \geq \ \hat{v}_{\ell_1}(t^{\ell_1 j}_1(+)) ^2 \geq \hat{v}_{j}(t^{\ell_1 j}_0) ^2 \ , \end{array} $$ where every second inequality holds due to Lemma $\ref{lemkda2}$ since $t^{\ell_{q+1} \ell_{q}}_0 > t^{\ell_{q} \ell_{q-1}}_1$ for each $q = 1,2, \ldots, k(ij)$, where $\ell_0 = j$ and $\ell_{k(ij)+1} = i$. We then obtain $(\ref{lbcp})$ again by Lemma $\ref{lemkda2}$ since $\hat{v}_i(t_1(ij)(+)) ^2 \geq \hat{v}_i (t^{i \ell_{k(ij)}}_1(+))^2$ and $\hat{v}_{j}(t^{\ell_1 j}_0)^2 \geq \hat{v}_{j}(t_0(ij)) ^2 $ because $t_1(ij) \geq t^{i \ell_{k(ij)}}_1$ and $t_0(ij) \leq t^{\ell_1 j}_0$ respectively. \end{proof}

\begin{lem}\label{lemkda4a}(Error Expression for $C_{[0,t_1]}$ satisfying $(\ref{star1ga})$) For any communication sequence $C_{[0,t_1]}$ satisfying $(\ref{star1ga})$ the total reduction in normal consensus squared error $\mathbf{E}^2 ( C_{[0,t_1]} )$ defined in $(\ref{newdef0})$ is, \begin{equation}\label{errred}\begin{array}{llll} \mathbf{E}^2 ( C_{[0,t_1]} ) & = \sum_{i=1}^n \big( E^2_{\frac{1}{n} \mathbf{1}_n} ( \hat{v}_i(0) ) - E^2_{\frac{1}{n} \mathbf{1}_n} ( \hat{v}_i(t_1(+)) ) \big) \\ & = \frac{n-1}{n} - \sum_{i=1}^n E^2_{\frac{1}{n} \mathbf{1}_n} ( \hat{v}_i(t_1(+) ) \\ & = \sum_{\ell \in \mathbb{N}} \mathbf{E}^2 ( C_{[t^{\ell}_0,t^{\ell}_1 ]} ) \ , \ C_{[t^{\ell}_0,t^{\ell}_1 ]} \in \mbox{S$\mathcal{V}$SC} \\ & \leq \frac{n-1}{n} \ . \end{array}\end{equation} \end{lem} \begin{proof} The first line follows from (A7),(A8) and the definitions $(\ref{newdef2})-(\ref{newdef1})$. The second line in $(\ref{errred})$ is due to the initialization $(\ref{initsol})$. The third line in $(\ref{errred})$ follows since any communication sequence $C_{[0,t_1]}$ satisfying $(\ref{star1ga})$ can be partitioned into an infinite number of disjoint sequences $C_{[t^{\ell}_0,t^{\ell}_1 ]} \in \mbox{S$\mathcal{V}$SC}$. The last line in $(\ref{errred})$ follows since the minimum error of any normal consensus estimate is $0$. \end{proof}

\begin{lem}\label{lemkda4b}(Vanishing Reduction in Error for $C_{[0,t_1]}$ satisfying $(\ref{star1ga})$) For any communication sequence $C_{[0,t_1]}$ satisfying $(\ref{star1ga})$ there exists an integer $\ell_\varepsilon$ such that, \begin{equation}\label{lab1} \mathbf{E}^2 ( C_{[t^{\ell}_0,t^{\ell}_1 ]} ) \leq \varepsilon \ , \ \forall \ \ell \geq \ell_\varepsilon \ , \end{equation} for any $\varepsilon > 0$, where $\mathbf{E}^2 ( C_{[t^{\ell}_0,t^{\ell}_1 ]} )$ is defined by $(\ref{newdef0})$. \end{lem} \begin{proof} The third line of $(\ref{errred})$ implies that for any $C_{[0,t_1]}$ satisfying $(\ref{star1ga})$ the quantity $\mathbf{E}^2 ( C_{[0,t_1]} )$ is the sum of an infinite number of non-negative terms $\mathbf{E}^2 ( C_{[t^{\ell}_0,t^{\ell}_1 ]} )$, the fourth line in $(\ref{errred})$ implies $\mathbf{E}^2 ( C_{[0,t_1]} )$ is bounded above, thus $(\ref{lab1})$ follows by the monotonic sequence theorem. \end{proof}

\begin{lem}\label{lemkda4d}(DA Lower Bound on Reduction in Error for any $C_{[0,t_1]}$ satisfying $(\ref{scdef})$) Any communication sequence $C_{[t_0,t_1]}$ satisfying $(\ref{scdef})$ implies, \begin{equation}\label{red2}\begin{array}{llll} & \mathbf{E}^2 ( C_{[t_0,t_1]} ) \geq n \big( \min_{i \in \mathcal{V}} \{ \hat{v}_i(t_1(+)) ^2 \} \\ & \ \ \ \ \ \ \ \ \ \ \ \ \ \ \ \ \ \ \ \ \ \ \ \ \ \ \ \ \ - \max_{i \in \mathcal{V}} \{ \hat{v}_i(t_0) ^2 \} \big) \ , \\ & \ \ \ \ \ \ \ \ \ \ \ \ \ \ \geq 0 \ . \end{array}\end{equation} \end{lem} \begin{proof} The first line in $(\ref{red2})$ holds for any communication sequence $C_{[t_0,t_1]}$, $$ \begin{array}{llll} & \mathbf{E}^2 ( C_{[t_0,t_1]} ) = \sum_{i = 1}^n \big( E^2_{\frac{1}{n}\mathbf{1}_n} ( \hat{v}_i(t_0)) - E^2_{\frac{1}{n}\mathbf{1}_n} ( \hat{v}_i(t_1(+)) ) \big) \\ & \ \ \ \ = \sum_{i = 1}^n \big( \hat{v}_i(t_1(+)) ^2 - \hat{v}_i(t_0) ^2 \big) \\ & \ \ \ \ \geq n \big( \ \min_{i \in \mathcal{V}} \{ \hat{v}_i(t_1(+)) ^2 \} - \max_{i \in \mathcal{V}} \{ \hat{v}_i(t_0) ^2 \} \ \big) \ . \end{array} $$ To prove the second line in $(\ref{red2})$ it is required that $C_{[t_0,t_1]}$ satisfies $(\ref{scdef})$. We define, $$ \begin{array}{llll} & \underline{\ell} = \arg_i \ \min_{i \in \mathcal{V}} \{ \hat{v}_i(t_1(+)) ^2 \} \ , \\ & \overline{\ell} = \arg_i \ \max_{i \in \mathcal{V}} \{ \hat{v}_i(t_0) ^2 \} \ . \end{array} $$ Since $C_{[t_0,t_1]}$ satisfies $(\ref{scdef})$ there exists a communication path $C^{\underline{\ell} \overline{\ell}}_{[t_0(\underline{\ell} \overline{\ell}),t_1(\underline{\ell} \overline{\ell})]} \in C_{[t_0,t_1]}$. The Lemma $\ref{lemkda4c}$ then implies, $$ \hat{v}_{\underline{\ell}} ( t_1 ( \underline{\ell} \overline{\ell} ) (+) ) ^2 \geq \hat{v}_{\overline{\ell}} ( t_0 ( \underline{\ell} \overline{\ell} ) ) ^2 \ . $$ The second line in $(\ref{red2})$ then follows because, $$ \begin{array}{llll} & \min_{i \in \mathcal{V}} \{ \hat{v}_i(t_1(+))^2 \} = \hat{v}_{\underline{\ell}} (t_1(+))^2 \\ & \ \ \ \ \ \geq \hat{v}_{\underline{\ell}} (t_1 ( \underline{\ell} \overline{\ell} )(+) )^2 \\ & \ \ \ \ \ \geq \hat{v}_{\overline{\ell}} ( t_0 ( \underline{\ell} \overline{\ell} ) ) ^2 \\ & \ \ \ \ \ \geq \hat{v}_{\overline{\ell}} ( t_0 ) ^2 = \max_{i \in \mathcal{V}} \{ \hat{v}_i(t_0)^2 \} \ , \end{array} $$ where the first and third inequality hold due to Lemma $\ref{lemkda2}$ because $t_1 \geq t_1 ( \underline{\ell} \overline{\ell} )$ and $t_0 ( \underline{\ell} \overline{\ell} ) \geq t_0$ respectively. \end{proof}

\begin{lem}\label{lemkda6}(DA Local Convergence of Normal Consensus Estimates) For any communication sequence $C_{[0,t_1]}$ satisfying $(\ref{star1ga})$ there exists an integer $\ell_\chi$ such that for all $\ell \geq \ell_\chi$, \begin{equation}\label{lab1a} \big( \hat{v}_i( t^{i j}_1) - \hat{v}_j( t^{i j}_0 ) \big)^2 \leq \chi \ , \ \forall \ S^{i j} ( t^{i j}_0, t^{i j}_1 ) \in C_{[t^{\ell}_0,t^{\ell}_1 ]} \ , \end{equation} for any $\chi > 0$. \end{lem} \begin{proof} For any communication sequence $C_{[0,t_1]}$ satisfying $(\ref{star1ga})$ the Lemma $\ref{lemkda4b}$ implies there exists an integer $\ell_\varepsilon$ such that $(\ref{lab1})$ holds for any $\varepsilon > 0$. For all $\ell \geq \ell_\varepsilon$ we thus have for any signal $S^{ij} ( t^{i j}_0, t^{i j}_1 ) \in C_{[t^{\ell}_0,t^{\ell}_1 ]}$, \begin{equation}\label{lab3}\begin{array}{llll} & \hat{v}_i (t^{ij}_1) ^2 - \hat{v}_j (t^{ij}_0) ^2 \leq \hat{v}_i (t^\ell_1(+)) ^2 - \hat{v}_j (t^\ell_0) ^2 \ , \\ & \ \ \ \ \ \ \ \ \ \ \ \ \leq \sum_{r=1}^n \big( \hat{v}_r (t^\ell_1(+)) ^2 - \hat{v}_r (t^\ell_0) ^2 \big) \ , \\ & \ \ \ \ \ \ \ \ \ \ \ \ = \mathbf{E}^2 ( C_{[t^{\ell}_0,t^{\ell}_1 ]} ) \leq \varepsilon \ . \end{array}\end{equation} Note that the first inequality in $(\ref{lab3})$ holds by Lemma $\ref{lemkda2}$ since $t^{ij}_0 \geq t^\ell_0$ and $t^{ij}_1 \leq t^\ell_1$ for any $S^{i j} ( t^{i j}_0, t^{i j}_1 ) \in C_{[t^{\ell}_0,t^{\ell}_1 ]}$. The second inequality in $(\ref{lab3})$ holds since, \begin{equation}\label{lab3a}\begin{array}{llll} & \sum_{r=1}^n \big( \hat{v}_r (t^\ell_1(+)) ^2 - \hat{v}_r (t^\ell_0) ^2 \big) \\ & \ \ \ \ \ \ = \big( \hat{v}_i (t^\ell_1(+)) ^2 - \hat{v}_j (t^\ell_0) ^2 \big) + \big( \hat{v}_j (t^\ell_1(+)) ^2 - \hat{v}_i (t^\ell_0) ^2 \big) \\ & \quad \quad \ \ \ \ \ \ \ + \sum_{r \in \mathcal{V} \setminus \{ i, j \}} \big( \hat{v}_r (t^\ell_1(+)) ^2 - \hat{v}_r (t^\ell_0) ^2 \big) \ , \end{array}\end{equation} where \begin{equation}\label{ar1} \sum_{r \in \mathcal{V} \setminus \{ i, j \}} \big( \hat{v}_r (t^\ell_1(+)) ^2 - \hat{v}_r (t^\ell_0) ^2 \big) \geq 0 \end{equation} holds due to Lemma $\ref{lemkda2}$, and, \begin{equation}\label{ar2}\begin{array}{llll} \hat{v}_j (t^\ell_1(+)) ^2 - \hat{v}_i (t^\ell_0) ^2 & \geq \min_{r \in \mathcal{V}} \{ \hat{v}_r( t^\ell_1(+)) ^2 \} \\ & \ \ \ \ \ \ \ \ - \max_{r \in \mathcal{V}} \{ \hat{v}_r( t^\ell_0) ^2 \} \\ & \geq 0 \ , \end{array}\end{equation} where the second inequality in $(\ref{ar2})$ holds due to Lemma $\ref{lemkda4d}$ since $C_{[t^{\ell}_0,t^{\ell}_1 ]} \in \mbox{S$\mathcal{V}$SC}$ and thus satisfies $(\ref{scdef})$. Together $(\ref{lab3a})$, $(\ref{ar1})$ and $(\ref{ar2})$ imply the second inequality in $(\ref{lab3})$. Applying Lemma $\ref{lemkda4}$ to Lemma $\ref{lemkda4b}$ implies that for all $S^{i j} ( t^{i j}_0, t^{i j}_1 ) \in C_{[t^{\ell}_0,t^{\ell}_1 ]}$ and $\ell \geq \ell_\varepsilon$, \begin{equation}\label{ar3}\begin{array}{llll} \varepsilon & \geq \mathbf{E}^2 (C_{[t^{\ell}_0,t^{\ell}_1 ]}) \\ & \geq \mathbf{E}^2 \big( S^{ij} ( t^{i j}_0, t^{i j}_1 ) \big) \ , \\ & = E_{\frac{1}{n} \mathbf{1}_n } ^2 \big( \hat{v}_i(t^{ij}_1) \big) - E_{\frac{1}{n} \mathbf{1}_n } ^2 \big( \hat{v}_i(t^{ij}_1(+)) \big) \ , \\ & \geq \mbox{max} \{ \hat{v}_j(t^{ij}_0) ^2 - \hat{v}_i(t^{ij}_1) ^2 \ , \\ & \ \ \ \ \ \ \ \ \ \ \ \ \ \ \ n \big( \hat{v}_j(t^{ij}_0) ' ( \hat{v}_j(t^{ij}_0) - \hat{v}_i(t^{ij}_1) ) \big)^2 \} \ , \end{array}\end{equation} for any $\varepsilon > 0$. For notational convenience denote $\hat{v}_i = \hat{v}_i ( t^{ij}_1)$ and $\hat{v}_j = \hat{v}_j ( t^{ij}_0)$. Combining $(\ref{lab3})$ and $(\ref{ar3})$ implies that for any $\varepsilon > 0$ there exist an integer $\ell_\varepsilon$ such that, \begin{equation}\label{ar4}\begin{array}{llll} \ \ \hat{v}_i ^2 - \hat{v}_j ^2 & \leq \varepsilon \ , \\ \hat{v}_j ' ( \hat{v}_j - \hat{v}_i ) & \leq \sqrt{ \varepsilon / n } \ , \end{array}\end{equation} for any $S^{i j} ( t^{i j}_0, t^{i j}_1 ) \in C_{[t^{\ell}_0,t^{\ell}_1 ]}$ and $\ell \geq \ell_\varepsilon$. To obtain $(\ref{lab1a})$ note that $(\ref{ar4})$ implies, \begin{equation}\label{chil1}\begin{array}{llll} ( \hat{v}_i - \hat{v}_j )^2 & = \hat{v}_i^2 - 2 \hat{v}_i ' \hat{v}_j + \hat{v}_j^2 \ , \\ & = \hat{v}_i^2 - \hat{v}_j^2 + 2 \hat{v}_j ' ( \hat{v}_j - \hat{v}_i ) \ , \\ & \leq \varepsilon + 2 \sqrt{\varepsilon/n}\ , \\ & \leq \sqrt{ \varepsilon} ( 1 + 2 / \sqrt{n} ) \ , \ \mbox{$\forall \ \epsilon \in (0, 1]$.} \end{array}\end{equation} We thus define $\overline{\varepsilon}(\chi)$, \begin{equation}\label{chil} \overline{\varepsilon}(\chi) = \Big( \frac{\chi}{1 + 2 / \sqrt{n} } \Big) ^2 \ . \end{equation} For any $\chi \in (0,1)$ and $\varepsilon \in (0, \overline{\varepsilon}(\chi) ]$ the result $(\ref{lab1a})$ then follows from $(\ref{chil1})$. \end{proof}

\begin{lem}\label{lemkda7a}(DA Properties of the Normal Consensus Update) Upon reception of any signal $S^{ij} ( t^{i j}_0, t^{i j}_1 )$ the normal consensus estimate $\hat{v}_i(t^{ij}_1(+))$ that results from the update problem $(\ref{optkda3})$ will satisfy, \begin{equation}\label{aa5} \hat{v}_i(t^{ij}_1) ' \big( \hat{v}_i(t^{ij}_1) - \hat{v}_i(t^{ij}_1(+)) \big) \leq 0 \ . \end{equation} \end{lem}\begin{proof} Let us define $\tilde{w}$, \begin{equation}\label{def1w} \tilde{w} = \arg_{\mbox{$\tilde{v}$}} \ \min_{\mbox{$\tilde{v} \in \mbox{span} \{ \hat{v}_i(t^{ij}_1(+)) , \hat{v}_i(t^{ij}_1) , e_i \}$}} \ \big( \tilde{v} - \frac{1}{n} \mathbf{1}_n \big)^2 \ , \end{equation} where $\hat{v}_i(t^{ij}_1(+))$ is given by $(\ref{optkda3})$. Note that $\hat{v}_i(t^{ij}_1(+)) \in \mbox{span} \{ \hat{v}_i(t^{ij}_1) , \hat{v}_j(t^{ij}_0) , e_i \}$ implies, \begin{equation}\label{gg2} \mbox{span} \{ \hat{v}_i(t^{ij}_1(+)) , \hat{v}_i(t^{ij}_1) , e_i \} \subseteq \mbox{span} \{ \hat{v}_i(t^{ij}_1) , \hat{v}_j(t^{ij}_0) , e_i \} \ . \end{equation} From $(\ref{gg2})$ we have, \begin{equation}\label{gg3} E_{\frac{1}{n} \mathbf{1}_n} ( \tilde{w} ) \geq E_{\frac{1}{n} \mathbf{1}_n} \big( \hat{v}_i(t^{ij}_1(+)) \big) \ , \end{equation} and thus combining $(\ref{gg3})$ and $(\ref{gkda0})$ implies, \begin{equation}\label{gkda1} \tilde{w}^2 \leq \hat{v}_i(t^{ij}_1(+)) ^2 \ . \end{equation} Next observe that since $\tilde{w}$ is defined by $(\ref{def1w})$, if $(\ref{gkda1})$ is applied to the signal $S^{ii}( t^{i j}_1 , t^{i j}_1(+) )$ then, \begin{equation}\label{gg3w}\begin{array}{llll} \mathbf{E} ^2 \big( S^{ii}( t^{i j}_1 , t^{i j}_1 (+)) \big) & \equiv E_{\frac{1}{n} \mathbf{1}_n } ^2 \big( \hat{v}_i(t^{ij}_1(+)) \big) - E_{\frac{1}{n} \mathbf{1}_n } ^2 (\tilde{w}) \ , \\ & = \tilde{w}^2 - \hat{v}_i(t^{ij}_1(+))^2 \leq 0 \ . \end{array} \end{equation} Applying Lemma $\ref{lemkda4}$ to the signal $S^{ii}( t^{i j}_1 , t^{i j}_1 (+))$ then implies, \begin{equation}\begin{array}{llll} 0 & \geq \mathbf{E} ^2 \big( S^{ii}( t^{i j}_1 , t^{i j}_1 (+) ) \big) \\ & \geq \mbox{max} \{ \hat{v}_i(t^{ij}_1) ^2 - \hat{v}_i(t^{ij}_1(+)) ^2 \ , \\ & \ \ \ \ \ \ \ \ \ \ \ \ \ n \big( \hat{v}_i(t^{ij}_1) ' ( \hat{v}_i(t^{ij}_1) - \hat{v}_i(t^{ij}_1(+)) ) \big)^2 \} \end{array} \end{equation} where the first line follows from $(\ref{gg3w})$ and the last line implies $(\ref{aa5})$. \end{proof}

\begin{lem}\label{lemkda7b}(DA Vanishing Change in Normal Consensus Update) For any communication sequence $C_{[0,t_1]}$ satisfying $(\ref{star1ga})$ there exists an integer $\ell_\varepsilon$ such that for all $\ell \geq \ell_\varepsilon$, \begin{equation}\label{ch1a} \big( \hat{v}_i( t^{i j}_1(+)) - \hat{v}_i( t^{i j}_1) \big)^2 \leq \varepsilon \ , \ \forall \ S^{i j} ( t^{i j}_0, t^{i j}_1 ) \in C_{[t^{\ell}_0,t^{\ell}_1 ]} \ , \end{equation} for any $\varepsilon > 0$. \end{lem}\begin{proof} Recall that Lemma $\ref{lemkda4b}$ implies that for any communication sequence $C_{[0,t_1]}$ satisfying $(\ref{star1ga})$ there exists an integer $\ell_\varepsilon$ such that $(\ref{lab1})$ holds for any $\varepsilon > 0$, we thus observe for $\ell \geq \ell_\varepsilon$, \begin{equation}\label{bal1}\begin{array}{llll} \varepsilon & \geq \mathbf{E}^2 \big( C_{[t^{\ell}_0,t^{\ell}_1 ]} \big) \\ & \geq \mathbf{E} ^2 ( S^{i j} ( t^{i j}_0, t^{i j}_1 ) ) \ , \ \forall \ S^{i j} ( t^{i j}_0, t^{i j}_1 ) \in C_{[t^{\ell}_0,t^{\ell}_1 ]} \\ & = \hat{v}_i (t^{ij}_1(+) )^2 - \hat{v}_i (t^{ij}_1 )^2 \ . \end{array}\end{equation} For all $\ell \geq \ell_\varepsilon$ and signals $S^{i j} ( t^{i j}_0, t^{i j}_1 ) \in C_{[t^{\ell}_0,t^{\ell}_1 ]}$ we then have, $$ \begin{array}{llll} & \big( \hat{v}_i (t^{ij}_1(+) ) - \hat{v}_i (t^{ij}_1) \big) ^2 = \hat{v}_i (t^{ij}_1(+)) ^2 - \hat{v}_i (t^{ij}_1)^2 \\ & \ \ \ \ \ \ \ \ \ \ \ \ \ \ \ \ \ \ + 2 \hat{v}_i (t^{ij}_1) ' \big( \hat{v}_i (t^{ij}_1) - \hat{v}_i (t^{ij}_1(+)) \big) \\ & \ \ \ \ \ \ \ \ \ \ \ \ \ \ \ \ \ \ \ \ \leq \mathbf{E}^2 \big( S^{i j} ( t^{i j}_0, t^{i j}_1 ) \big) \leq \varepsilon \ , \end{array} $$ where the first inequality follows from Lemma $\ref{lemkda7a}$ and the second inequality from Lemma $\ref{lemkda7b}$. \end{proof}

\begin{lem}\label{lemkda7c}(DA Vanishing Change between Normal Consensus Update and Signal) For any communication sequence $C_{[0,t_1]}$ satisfying $(\ref{star1ga})$ there exists an integer $\ell_\gamma$ such that, \begin{equation}\label{ch1} \big( \hat{v}_i( t^{i j}_1(+)) - \hat{v}_j( t^{i j}_0) \big)^2 \leq \gamma \ , \ \forall \ S^{i j} ( t^{i j}_0, t^{i j}_1 ) \in C_{[t^{\ell}_0,t^{\ell}_1 ]} , \end{equation} for all $\ell \geq \ell_\gamma$ and any $\gamma > 0$. \end{lem}\begin{proof} For any communication sequence $C_{[0,t_1]}$ satisfying $(\ref{star1ga})$ the Lemma $\ref{lemkda6}$ implies that $(\ref{lab1a})$ holds for any $\chi \in (0,1)$ and $\varepsilon \in (0, \overline{\varepsilon}(\chi) ]$, where $\overline{\varepsilon}(\chi)$ is given by $(\ref{chil})$. The Lemma $\ref{lemkda7b}$ implies that $(\ref{ch1a})$ holds for any $\varepsilon > 0$, thus for any $\ell \geq \ell_{\overline{\varepsilon}(\chi)}$ and $S^{i j} ( t^{i j}_0, t^{i j}_1 ) \in C_{[t^{\ell}_0,t^{\ell}_1 ]}$ the triangle inequality then implies, $$ \begin{array}{llll} & \sqrt{ \big( \hat{v}_i (t^{ij}_1(+)) - \hat{v}_j (t^{ij}_0) \big) ^2} \leq \sqrt{ \big( \hat{v}_i (t^{ij}_1(+)) - \hat{v}_i (t^{ij}_1) \big) ^2} \\ & \ \ \ \ \ \ \ \ \ \ \ \ \ \ \ \ \ \ \ \ \ \ \ \ \ \ \ \ \ \ \ \ \ \ + \sqrt{ \big( \hat{v}_i (t^{ij}_1) - \hat{v}_j (t^{ij}_0) \big) ^2} \\ & \ \ \ \ \ \ \ \ \ \ \ \ \ \ \ \ \ \ \ \ \ \ \ \ \leq \sqrt{\overline{\varepsilon}(\chi)} + \sqrt{\chi} \ , \\ & \ \ \ \ \ \ \ \ \ \ \ \ \ \ \ \ \ \ \ \ \ \ \ \ \leq 2 \sqrt{\chi} \ , \ \mbox{$\forall \ \chi \in ( 0 ,1 )$.} \end{array} $$ Any $\chi \in (0, \gamma/4]$ and $\varepsilon \in (0, \overline{\varepsilon}(\chi) ]$ thus yields $(\ref{ch1})$ for any $\gamma \in ( 0, 4]$. \end{proof}

\begin{lem}\label{thmkda7c}(DA Network Convergence to Average-Consensus) For any communication sequence $C_{[0,t_1]}$ satisfying $(\ref{star1ga})$ there exists an integer $\ell_\xi$ such that for all $\ell \geq \ell_\xi$, \begin{equation}\label{l2a} \sum_{i = 1}^n \big( \hat{v}_i( t^\ell_1 (+) ) - \frac{1}{n}\mathbf{1}_n \big)^2 \leq \xi \ , \ \forall \ C_{[t^{\ell}_0,t^{\ell}_1 ]} \in C_{[0,t_1]} , \end{equation} for any $\xi > 0$. \end{lem} \begin{proof} Since $C_{[0,t_1]}$ satisfies $(\ref{star1ga})$ we have $C_{[t^{\ell}_0,t^{\ell}_1 ]} \in \mbox{S$\mathcal{V}$SC}$ for each $\ell \in \mathbb{N}$, thus there exists a communication path $C^{ij}_{[t_0(ij),t_1(ij)]} \in C_{[t^{\ell}_0,t^{\ell}_1 ]}$ for any $i \in \mathcal{V}$, $j \in \mathcal{V}_{-i}$, and $\ell \in \mathbb{N}$. For any $i \in \mathcal{V}$, $j \in \mathcal{V}_{-i}$ and $\ell \in \mathbb{N}$ the triangle inequality then implies, \begin{equation}\label{new1}\begin{array}{llll} & \sqrt{ \big( \hat{v}_{ij} (t^\ell_1(+)) - \hat{v}_{jj} (t^{\ell_1 j}_0) \big) ^2} \\ & \ \ \ \ \ \ \ \ \ \ \leq \sum_{S^{rp}(t^{rp}_0 , t^{rp}_1 ) \in C^{ij}_{[t_0(ij),t_1(ij)]} } \\ & \ \ \ \ \ \ \ \ \ \ \ \ \ \ \ \ \ \ \ \ \ \sqrt{\big( \hat{v}_{r j}(t^{rp}_1(+)) - \hat{v}_{p j} ( t^{rp}_0 ) \big)^2} \\ & \ \ \ \ \ \ \ \ \ \ \ \ \ \ + \sum_{q = 1}^{k(ij)+1} \sum_{S^{rp}(t^{rp}_0 , t^{rp}_1 ) \in Q^\ell_q (ij)} \\ & \ \ \ \ \ \ \ \ \ \ \ \ \ \ \ \ \ \ \ \ \ \sqrt{ \big( \hat{v}_{r j}(t^{r p}_1(+)) - \hat{v}_{r j} ( t^{rp}_1 ) \big)^2 } \end{array}\end{equation} where we define $\bar{C}_{[t^\ell_0, t^\ell_1]} = C_{[t^\ell_0, t^\ell_1]} \setminus C^{ij}_{[t_0(ij),t_1(ij)]}$ and, \begin{equation}\label{qdef}\begin{array}{llll} & Q^\ell_q (ij) = \{ S^{\ell_q m}(t^{\ell_q m}_0, t^{\ell_q m}_1 ) \in \bar{C}_{[t^\ell_0, t^\ell_1]} \ : \\ & \ \ \ \ \ \ \ \ \ \ \ \ \ \ \ \ \ \ \ \ \ \ \ \ \ \ \ \ t^{\ell_q m}_1 \in ( t^{\ell_q \ell_{q-1}}_1 , t^{\ell_{q+1} \ell_q}_0 ) \} , \\ & \mbox{for $q = 1, \ldots, k(ij)$, where $\ell_0 = j$, $\ell_{k(ij)+1} = i$} , \\ & Q^\ell_{k(ij)+1}(ij) = \{ S^{i m}(t^{i m}_0, t^{i m}_1 ) \ \in C_{[t^\ell_0, t^\ell_1]} \ : \\ & \ \ \ \ \ \ \ \ \ \ \ \ \ \ \ \ \ \ \ \ \ \ \ \ \ \ \ \ t^{i m}_1 > t_1(ij) \} \ . \end{array}\end{equation} We clarify that the RHS of $(\ref{new1})$ includes the differences between the received normal consensus vector $\hat{v}_{\ell_{q-1}}( t^{\ell_q \ell_{q-1}}_0)$ and the updated normal consensus vector $\hat{v}_{\ell_q}( t^{\ell_q \ell_{q-1}}_1(+))$ that results from each signal contained in the communication path $C^{ij}_{[t_0(ij),t_1(ij)]} \in C_{[t^{\ell}_0,t^{\ell}_1 ]}$. Each set $Q^\ell_q(ij)$ defined in $(\ref{qdef})$ contains the signals received at each node \emph{after} the respective signal in communication path was received, but \emph{before} the respective signal in the communication path was sent, as is required for an application of the triangle inequality. The set $Q^\ell_{k(ij)+1}(ij)$ contains the signals received at node $i$ \emph{after} the last signal in the communication path $C^{ij}_{[t_0(ij),t_1(ij)]}$ was received, but \emph{before or at} the end of the communication sequence $C_{[t^{\ell}_0,t^{\ell}_1 ]}$, again this is required for application of the triangle inequality.

For any communication sequence $C_{[0,t_1]}$ satisfying $(\ref{star1ga})$ the Lemma $\ref{lemkda7b}$ implies there exists an integer $\ell_\varepsilon$ such that $(\ref{ch1a})$ holds for any $\varepsilon > 0$. Likewise, for any communication sequence $C_{[0,t_1]}$ satisfying $(\ref{star1ga})$ the Lemma $\ref{lemkda7c}$ implies there exists an integer $\ell_\gamma$ such that $(\ref{ch1})$ holds for any $\gamma > 0$. Thus for any $\gamma \in (0, 4]$ if we let $\chi \in (0, \gamma/4]$ and $\varepsilon \in (0, \overline{\varepsilon}(\chi)]$ then for any $\ell \geq \ell_{\overline{\varepsilon}(\chi)}$, \begin{equation}\label{sh1}\begin{array}{llll} & \sqrt{ \big( \hat{v}_{ij} (t^\ell_1(+)) - \hat{v}_{jj} (t^{\ell_1 j}_0) \big) ^2} \\ & \ \ \ \ \leq | C^{ij}_{[t_0(ij),t_1(ij)]} | \sqrt{\gamma} \\ & \ \ \ \ \ \ \ \ \ \ \ \ + \sum_{q = 1}^{k(ij)+1} |Q^\ell_q (ij)| \sqrt{\overline{\varepsilon}(\chi)} \ , \\ & \ \ \ \ \leq \big( | C^{ij}_{[t_0(ij),t_1(ij)]} | + \sum_{q = 1}^{k(ij)+1} |Q^\ell_q (ij)| \big) \sqrt{\gamma} \ , \\ & \ \ \ \ \leq |C_{[t^{\ell}_0,t^{\ell}_1 ]}| \sqrt{\gamma} \ , \end{array}\end{equation} where the second inequality holds since $\overline{\varepsilon}(\chi) < \gamma$, and the last inequality holds since every signal contained in $C_{[t^{\ell}_0,t^{\ell}_1 ]}$ is represented by at most one term on the RHS of $(\ref{new1})$, that is, $$ | C_{[t^{\ell}_0,t^{\ell}_1 ]} | \geq | C^{ij}_{[t_0(ij),t_1(ij)]} | + \sum_{q = 1}^{k(ij)+1} | Q^\ell_q(ij) | \ . $$ Due to $(\ref{sh1})$, any $\xi > 0$ and $\gamma \in (0, \xi/(n | C_{[t^{\ell}_0,t^{\ell}_1 ]} | )^2 ]$ then implies, \begin{equation}\label{sh1a}\begin{array}{llll} \big( \hat{v}_{ij} (t^\ell_1(+)) - \hat{v}_{jj} (t^{\ell_1 j}_0) \big) ^2 & = \big( \hat{v}_{ij} (t^\ell_1(+)) - \frac{1}{n} \big) ^2 \ , \\ & \leq \xi/n^2 \ , \end{array}\end{equation} for any $(i,j) \in \mathcal{V}^2$, where the first equality holds due to Lemma $\ref{lemkda3}$. The result $(\ref{l2a})$ then follows from $(\ref{sh1a})$ since $$ \begin{array}{llll} & \sum_{i = 1}^n \big( \hat{v}_{i} ( t^\ell_1 (+) ) - \frac{1}{n} \mathbf{1}_n \big) ^2 \\ & \ \ \ \ \ \ = \sum_{i = 1}^n \sum_{j = 1}^n \big( \hat{v}_{i j} ( t^\ell_1 (+) ) - \hat{v}_{jj} (t^{\ell_1 j}_0) \big) ^2 \ , \\ & \ \ \ \ \ \ \leq n^2 \big( \xi / n^2 \big) \ , \end{array} $$ where the first equality holds again due to Lemma $\ref{lemkda3}$. \end{proof}

\emph{Proof.}(Theorem $\ref{aaacor5}$) Lemma $\ref{aaalemsda2}$-$\ref{aaalemis2}$. \emph{Overview of Proof.} The Lemma $\ref{aaalemsda2}$ proves the OH normal consensus update $(\ref{aaaisvhat})$. As stated in Lemma $\ref{aaalemis1}$, the update $(\ref{aaaisvhat})$ implies the error of each element of every normal consensus vector is a non-increasing function of time, and thus $(\ref{aaaisvhat})$ will imply the conditions stated in Lemma $\ref{aaalemis2}$ are sufficient and necessary for any node to obtain average-consensus. The Theorem $\ref{aaacor5}$ then follows immediately from Lemma $\ref{aaalemis2}$.

\begin{lem}\label{aaalemsda2}(OH Normal Consensus Estimate Update) Applying $(\ref{aaasdasig5})$ and $(\ref{aaais2})$ to the optimization problem $(\ref{optllca})$ yields the OH normal consensus estimate update $\hat{v}_i( t^{ij}_1 (+) )$ defined in $(\ref{aaaisvhat})$. \end{lem} \begin{proof} If $\hat{v}_j(t^{ij}_0) = \frac{1}{n} \mathbf{1}_n$ then $0 = S^{ij}_1(t^{ij}_0,t^{ij}_1 )$, in this case $(\ref{aaasdasig5})$ together with $(\ref{a10})$ implies $(\ref{acbm1})$. The update $(\ref{optllca})$ can thus be re-written as $(\ref{optbmac1})$. Since $\hat{s}_j(t^{ij}_0) = \bar{s}(0)$ is known we can let $\hat{s}_i(t^{ij}_1(+)) = \bar{s}(0)$ and thus obtain $\hat{v}_i(t^{ij}_1(+)) = \frac{1}{n} \mathbf{1}_n$ as the unique solution to $(\ref{optbmac1})$. Note that this coincides with $(\ref{aaaisvhat})$. Next assume that $\hat{v}_j(t^{ij}_0) \neq \frac{1}{n} \mathbf{1}_n$ and thus $0 \neq S^{ij}_1(t^{ij}_0,t^{ij}_1 )$. Under $(\ref{aaasdasig5})$ and $(\ref{aaais2})$ we can re-write $(\ref{optllca})$ as, \begin{equation}\label{aaaoptis1}\begin{array}{llll} \hat{v}_i ( t^{ij}_1(+)) = & \arg_{\mbox{$\tilde{v}$}} \ \min \ \big( \tilde{v} - \frac{1}{n} \mathbf{1}_n \big )^2 , \\ & \mbox{s.t. $(\ref{a10})$ holds, given $\hat{s}_i(t^{ij}_1) = \mathbf{S} \hat{v}_i(t^{ij}_1)$,} \\ & \mbox{     $s_j(0) = \mathbf{S} e_j , \ s_i(0) = \mathbf{S} e_i$.} \end{array}\end{equation} Given that $\hat{s}_i(t^{ij}_1) = \mathbf{S} \hat{v}_i(t^{ij}_1) , s_j(0) = \mathbf{S} e_j$, and $s_i(0) = \mathbf{S} e_i$ are known, the set of vectors $\hat{v}_i(t^{ij}_1(+))$ for which $(\ref{a10})$ holds is $\mbox{span} \{ \hat{v}_i(t^{ij}_1) , e_j , e_i \}$, thus $(\ref{aaaoptis1})$ can be re-written as, \begin{equation}\label{aaaoptis3} \hat{v}_i ( t^{ij}_1(+)) = \arg_{\mbox{$\tilde{v}$}} \ \min_{\mbox{$\tilde{v} \in \mbox{span} \{ \hat{v}_i( t^{ij}_1 ) , e_j , e_i \}$}} \ \big( \tilde{v} - \frac{1}{n} \mathbf{1}_n \big)^2 . \end{equation} If $\hat{v}_i(t^{ij}_1) \in \mbox{span} \{ e_j , e_i \}$ then $(\ref{aaaoptis3})$ becomes, $$ \begin{array}{llll} \hat{v}_i ( t^{ij}_1(+)) & = \arg_{\mbox{$\tilde{v}$}} \ \min_{\mbox{$\tilde{v} \in \mbox{span} \{ e_j , e_i \}$}} \ \big( \tilde{v} - \frac{1}{n} \mathbf{1}_n \big)^2 \ , \\ & = V_{(OH1)} V^+_{(OH1)} \frac{1}{n} \mathbf{1}_n \ , \ \ V_{(OH1)} = [ e_j , e_i ] \ . \end{array} $$ Since $e_i$ is linearly independent of $e_j$ we then have, \begin{equation}\label{aaaoptis5}\begin{array}{llll} & \hat{v}_i ( t^{ij}_1(+)) = V_{(OH1)} ( V'_{(OH1)} V _{(OH1)})^{-1} V '_{(OH1)} \frac{1}{n} \mathbf{1}_n \\ & \ \ \ \ \ \ \ \ \ \ \ \ \ = \frac{1}{n} V _{(OH1)} ( I_{2} )^{-1} \mathbf{1}_2 \ , \\ & \ \ \ \ \ \ \ \ \ \ \ \ \ = \frac{1}{n} ( e_i + e_j ) \ . \end{array}\end{equation} Note that if $\hat{v}_i(t^{ij}_1) \in \mbox{span} \{ e_j , e_i \}$ then the last line in $(\ref{aaaoptis5})$ coincides with $(\ref{aaaisvhat})$. Next assume $\hat{v}_i(t^{ij}_1) \notin \mbox{span} \{ e_j , e_i \}$. In this case $(\ref{aaaoptis3})$ can be expressed, $$ \begin{array}{llll} & \hat{v}_i ( t^{ij}_1(+)) = V_{(OH)} V^+_{(OH)} \frac{1}{n} \mathbf{1}_n \ , \\ & \ \ \ \ \ \ \ \ \ \ \ \ \  = V_{(OH)} ( V'_{(OH)} V_{(OH)} )^{-1} V'_{(OH)} \frac{1}{n} \mathbf{1}_n \\ & \ \ \ \ \ V_{(OH)} = [ \hat{v}_i(t^{ij}_1), \frac{1}{n} e_j , \frac{1}{n} e_i ] \ . \end{array} $$ Recall the discrete set of vectors $\mathbb{R}^n_{0,\frac{1}{n}}$ is defined in $(\ref{aaaaa2})$. Note that the initialization $(\ref{initsol})$ implies $\hat{v}_{ii}(0) = \frac{1}{n}$ and thus $\hat{v}_i(0) \in \mathbb{R}^n_{0,\frac{1}{n}}$. Also note that $(\ref{aaaoptis5})$ implies $\hat{v}_{ii}(t^{ij}_1(+)) = \frac{1}{n}$ and $\hat{v}_i(t^{ij}_1(+)) \in \mathbb{R}^n_{0, \frac{1}{n}}$. We thus assume, \begin{equation}\label{aaaaais} \hat{v}_{i i} (t^{ij}_1) = \frac{1}{n} \ , \ \ \hat{v}_i (t^{ij}_1) \in \mathbb{R}^n_{0,\frac{1}{n}} \ . \end{equation} Observe that under $(\ref{aaaaais})$ if the result $(\ref{aaaisvhat})$ is proven, then the assumptions $(\ref{aaaaais})$ are valid. For notational convenience denote $\hat{v}_i = \hat{v}_i (t^{ij}_1)$. Given $(\ref{aaaaais})$ the matrix $(V'_{(OH)} V_{(OH)})$ has an inverse $(\ref{aaamatinv1})$ below, \begin{equation}\label{aaamatinv1}\begin{array}{llll} (V'_{(OH)} V_{(OH)})^{-1} & = n^{-4} \left( \begin{array}{ccc} \hat{v}_i ' \mathbf{1}_n & 1 & v_{ij} \\ 1 & 1 & 0 \\ v_{ij} & 0 & 1 \end{array} \right) ^{-1} \\ & = \frac{n^{-4}}{|V'_{(OH)} V_{(OH)}|} \left( \begin{array}{ccc} 1 & -1 & - v_{ij} \\ -1 & ( \hat{v}_i ' \mathbf{1}_n - v_{ij} ) & v_{ij} \\ - v_{ij} & v_{ij} & ( \hat{v}_i ' \mathbf{1}_n - 1 ) \end{array} \right) \end{array} \end{equation} where the determinant $|V'_{(OH)} V_{(OH)}|$ can be computed as, $$ |V'_{(OH)} V_{(OH)}| = n^{-6} ( \hat{v}_i ' \mathbf{1}_n - 1 - v_{ij} ) \ . $$ Next observe that,\begin{equation}\label{aaaoptis4c} V'_{(OH)} \frac{1}{n} \mathbf{1}_n = \frac{1}{n^2} \left( \begin{array}{ccc} \hat{v}_i ' \mathbf{1}_n \\ 1 \\ 1 \end{array} \right) \ . \end{equation} From $(\ref{aaaoptis4c})$ we observe that right-multiplying $(\ref{aaamatinv1})$ by $V'_{(OH)} \frac{1}{n} \mathbf{1}_n$ and left-multiplying by $V_{(OH)}$ then yields $(\ref{aaaisvhat})$, $$ \begin{array}{llll} & V_{(OH)} ( V'_{(OH)} V_{(OH)})^{-1} V'_{(OH)} \frac{1}{n} \mathbf{1}_n \\ &  \ \ \ \ \ \ \ \  = \frac{1}{\hat{v}_i ' \mathbf{1}_n - 1 - \hat{v}_{ij} } V_{(OH)} \left( \begin{array}{ccc} \hat{v}_i ' \mathbf{1}_n - 1 - \hat{v}_{ij} \\ 0 \\ \hat{v}_i ' \mathbf{1}_n ( 1 - \hat{v}_{ij} ) + \hat{v}_{ij} - 1 \end{array} \right) \\ & \ \ \ \ \ \ \ \ = \frac{1}{n} v^{ij} ( t^{ij}_1 , t^{ij}_0 ) \ . \end{array} $$ \end{proof}

\begin{lem}\label{aaalemis1}(OH Element-Wise Non-Increasing Error of the Normal Consensus Estimate) The error $\big( \hat{v}_{i \ell} (t) - 1/n \big)^2$ of each element $\hat{v}_{i \ell} (t)$ is a non-increasing function of $t \geq 0$ for all $(i, \ell) \in \mathcal{V}^2$. \end{lem} \begin{proof} The result follows immediately from the normal consensus update $(\ref{aaaisvhat})$. \end{proof}

\begin{lem}\label{aaalemis2}(OH Normal Consensus Estimate Convergence to Average-Consensus) Under the OH algorithm, any node $i \in \mathcal{V}$ obtains average-consensus by time $t_1$ for some communication sequence $C_{[0,t_1)}$ iff at least one of the following two conditions holds, \begin{itemize} \item (C1i): there is a signal $S^{ij}( t^{ij}_0, t^{ij}_1) \in C_{[0,t_1)}$ for each $j \in \mathcal{V}_{-i}$. \item (C2i): there is a communication path $C^{i \ell}_{[t_0(i \ell),t_1(i \ell)]} \in C_{[0,t_1)}$ from at least one node $\ell$ such that $S^{\ell j}( t^{\ell j}_0, t^{\ell j}_1) \in C_{[0,t_0(i \ell))}$ for all $j \in \mathcal{V}_{-\ell}$.\end{itemize}\end{lem} \begin{proof} (Sufficiency.) If a communication sequence $C_{[0,t_1)}$ implies (C1i) for some node $i$ then under the update $(\ref{aaaisvhat})$ there will exist a time $t^{ij}_1 \in [0,t_1)$ such that $\hat{v}_{ij}(t^{i j}_1 (+)) = \frac{1}{n}$ for each $j \in \mathcal{V}_{-i}$, and thus Lemma $\ref{aaalemis1}$ together with $(\ref{initsol})$ imply $\hat{v}_{i}(t_1)= \frac{1}{n} \mathbf{1}_n$ and hence by $(\ref{a10})$ the node $i$ reaches average-consensus by time $t_1$. If a communication sequence $C_{[0,t_1)}$ implies (C2i) then by the previous reasoning the node $\ell$ will obtain average-consensus by time $t_0(i \ell)$, and thus by the update $(\ref{aaaisvhat})$ any node $j \in \mathcal{V}$ that $\ell$ sends a signal $S^{j \ell}( t^{j \ell}_0, t^{j \ell}_1) \in C_{[t_0(i \ell),t_1)}$ will obtain average-consensus at $t^{j \ell}_1(+)$. If the node $\ell$ has a communication path $C^{i \ell}_{[t_0(i \ell),t_1(i \ell)]} \in C_{[0,t_1)}$ to node $i$ it then follows that node $i$ will have obtained average-consensus by time $t_1(i \ell)(+)$, and thus by Lemma $\ref{aaalemis1}$ node $i$ obtains average-consensus at time $t_1$.

(Necessity.) Under the OH update $(\ref{aaaisvhat})$, if there is not a signal $S^{ij}( t^{ij}_0, t^{ij}_1) \in C_{[0,t_1)}$ then there will not exist a time $t \in [0, t_1)$ such that $\hat{v}_{ij}(t) = \frac{1}{n}$ unless there exists a communication path $C^{i \ell}_{[t_0(i \ell),t_1(i \ell)]} \in C_{[0,t_1)}$ from some node $\ell$ such that $S^{\ell j}( t^{\ell j}_0, t^{\ell j}_1) \in C_{[0,t_0(i \ell))}$ for all $j \in \mathcal{V}_{-\ell}$. It thus follows that node $i$ cannot obtain average-consensus by time $t_1$ for any communication sequence $C_{[0,t_1)}$ that does not imply either (C1i) or (C2i). \end{proof}

\emph{Theorem $\ref{aaacor5}$}(OH Network Convergence to Average-Consensus) \begin{proof} The condition (C) stated in Theorem $\ref{aaacor5}$ is equivalent to when either (C1i) or (C2i) hold for each node $i \in \mathcal{V}$, the result thus follows immediately from Lemma $\ref{aaalemis2}$. \end{proof}

\emph{Proof.}(Theorem $\ref{aaathm5}$) Lemma $\ref{aaalemdda0}$-$\ref{aaalemdda4d}$. \emph{Overview of Proof.} Similar to Theorem $\ref{thm6}$, we show that every normal consensus estimate $\hat{v}_i(t)$ satisfies the normalization property $(\ref{norm1})$, the ``zero local error'' property $(\ref{pp1})$, and furthermore the discretization $\hat{v}_i(t) \in \mathbb{R}^n_{0,\frac{1}{n}}$. The Lemma $\ref{aaalemdda0w1}$ and Lemma $\ref{aaalemdda3}$ shows that the update $(\ref{aaaddaveq})$ respectively implies the error of each normal consensus estimate is non-decreasing with time, and that the normal consensus estimate will not change unless there is a reduction in error. We then show, analogous to Lemma $\ref{lemkda4b}$, that the reduction in error that results from any signal will eventually vanish if $C_{[0,t_1]}$ satisfies $(\ref{star1g})$. Due to the discretization $\hat{v}_i(t) \in \mathbb{R}^n_{0,\frac{1}{n}}$ this implies the reduction in error that results from any signal will eventually strictly equal zero if $C_{[0,t_1]}$ satisfies $(\ref{star1g})$, see Lemma $\ref{aaalemdda4d}$. When this occurs we can show that $\hat{v}_i(t) = \frac{1}{n} \mathbf{1}_n$ by utilizing Lemma $\ref{aaalemdda4}$, Lemma $\ref{aaalemdda5}$, together with the ``zero local error'' property $(\ref{pp1})$, hence $(\ref{acdef1})$ holds and so by Definition $\ref{def11}$ a network average-consensus is obtained.

\begin{lem}\label{aaalemdda0}(DDA Normal Consensus Estimate Discretization) Every normal consensus estimate $\hat{v}_i(t)$ satisfies $\hat{v}_i(t) \in \mathbb{R}^n_{0,\frac{1}{n}}$ for all $i \in \mathcal{V}$ and $t \geq 0$. \end{lem} \begin{proof} Note that the initialization $(\ref{initsol})$ implies $\hat{v}_i(0) \in \mathbb{R}^n_{0,\frac{1}{n}}$. The optimization problem $(\ref{aaaoptdisc})$ requires that any solution $\hat{v}_i(t(+))$ satisfies $\hat{v}_i(t(+)) \in \mathbb{R}^n_{0,\frac{1}{n}}$. Under the DDA algorithm, the assumption (A7) implies every normal consensus estimate remains fixed unless updated via $(\ref{aaaoptdisc})$, it thus follows that $\hat{v}_i(t) \in \mathbb{R}^n_{0,\frac{1}{n}}$ for all $i \in \mathcal{V}$ and $t \geq 0$. \end{proof}

\begin{lem}\label{aaalemdda1}(DDA Consensus Estimate Normalization) Every normal consensus estimate $\hat{v}_i(t)$ satisfies $(\ref{norm1})$, and furthermore, \begin{equation}\label{aaanorm1a} \big( \hat{v}^{-i}_i(t) \big)^2 = \frac{1}{n} \hat{v}^{-i}_i (t) ' \mathbf{1}_{n-1} \ , \ \forall \ i \in \mathcal{V} \ , \ \forall \ t \geq 0 \ . \end{equation} \end{lem} \begin{proof} The Lemma $\ref{aaalemdda0}$ implies, \begin{equation}\label{aaanormda}\begin{array}{llll} \frac{1}{n} \hat{v}_i (t) ' \mathbf{1}_n & = \frac{1}{n} \sum_{\ell = 1}^n \hat{v}_{i \ell} (t) \ , \\ & = \frac{1}{n} \sum_{\ell \in \mathcal{V} \ : \ \hat{v}_{i \ell} (t) = \frac{1}{n} } \ \big( \frac{1}{n} \big) \ , \\ & = \sum_{\ell \in \mathcal{V} \ : \ \hat{v}_{i \ell} (t) = \frac{1}{n} } \ \big( \frac{1}{n^2} \big) \ , \\ & = \sum_{\ell = 1}^n \hat{v}_{i \ell} (t) ^2 \ , \\ & = \hat{v}_i (t) ^2 \ . \end{array}\end{equation} If $\hat{v}_i(t) \in \mathbb{R}^n_{0,\frac{1}{n}}$ then $\hat{v}^{-i}_i(t) \in \mathbb{R}^{n-1}_{0,\frac{1}{n}}$, thus a similar argument to $(\ref{aaanormda})$ implies $(\ref{aaanorm1a})$. \end{proof}

\begin{lem}\label{aaalemdda2}(DDA Local Zero Error Property) Every normal consensus estimate $\hat{v}_i(t)$ satisfies $(\ref{pp1})$. \end{lem} \begin{proof} The Lemma $\ref{lembm0}$ implies $\hat{v}_{ii} (0) = \frac{1}{n}$ for each $i \in \mathcal{V}$. Next observe that under $(A7)$ the estimate $\hat{v}_i(t)$ will not change unless a signal $S^{i j} ( t^{i j}_0, t^{i j}_1 )$ is received at node $i$. If a signal is received then $\hat{v}_i(t)$ is updated by $(\ref{aaaddaveq})$ given below. We now show that under the assumption $(\ref{pp1})$, the solution $\hat{v}_i ( t^{ij}_1(+))$ specified by $(\ref{aaaddaveq})$ will imply $(\ref{pp1})$ for every set vectors $\{e_i, \hat{v}_i ( t^{ij}_1), \hat{v}_j ( t^{ij}_0) \}$. If $\hat{v}^{-i} _i ( t^{ij}_1 ) ' \hat{v}^{-i} _j ( t^{ij}_0 ) = 0$ then $\hat{v}_i ( t^{ij}_1 (+) ) = \hat{v} _i ( t^{ij}_1 ) + \hat{v}_j( t^{ij}_0 ) - \hat{v}_{ji} (t^{ij}_0 ) e_i$ and thus, $$ \begin{array}{llll} \hat{v}_{ii} ( t^{ij}_1 (+) ) & = \hat{v} _{ii} ( t^{ij}_1 ) + \hat{v}_{ji}( t^{ij}_0 ) - \hat{v}_{ji} (t^{ij}_0 ) \ , \\ & = \frac{1}{n} \ . \end{array} $$ If $\hat{v}^{-i} _i ( t^{ij}_1 ) ' \hat{v}^{-i} _j ( t^{ij}_0 ) > 0$ and $\hat{v}^{-i} _i ( t^{ij}_1 ) ^2 < \hat{v}^{-i} _j ( t^{ij}_0 )^2$ then $\hat{v}_i ( t^{ij}_1 (+) ) = \hat{v} _j ( t^{ij}_0 ) + e_i \big( \frac{1}{n} - \hat{v}_{ji} (t^{ij}_0 ) \big)$ and thus, $$ \begin{array}{llll} \hat{v}_i ( t^{ij}_1 (+) ) & = \hat{v} _{ji} ( t^{ij}_0 ) + \frac{1}{n} - \hat{v}_{ji} (t^{ij}_0 ) \ , \\ & = \frac{1}{n} \ . \end{array} $$ Finally, if $\hat{v}^{-i} _i ( t^{ij}_1 ) ' \hat{v}^{-i} _j ( t^{ij}_0 ) > 0$ and $\hat{v}^{-i} _i ( t^{ij}_1 ) ^2 \geq \hat{v}^{-i} _j ( t^{ij}_0 )^2$ then $\hat{v}_i ( t^{ij}_1 (+) ) = \hat{v} _i ( t^{ij}_1 )$ and thus $(\ref{pp1})$ follows by assumption. \end{proof}

\begin{lem}\label{aaalemdda01}(DDA Normal Consensus Estimate Magnitude Equivalence to Error) For any two normal consensus estimates $\hat{v}_i(t)$ and $\hat{v}_j(t)$, \begin{equation}\label{aaaerr1} \hat{v}_i(t) ^2 \geq \hat{v}_j(t) ^2 \ \ \Leftrightarrow \ \ E_{\frac{1}{n} \mathbf{1}_n} ^2 \big( \hat{v}_{i} (t) \big) \leq E_{\frac{1}{n} \mathbf{1}_n} ^2 \big( \hat{v}_{j} (t) \big) \ . \end{equation} \end{lem} \begin{proof} For any normal consensus estimate $\hat{v}_i(t)$, \begin{equation}\label{aaaerr1p}\begin{array}{llll} E_{\frac{1}{n} \mathbf{1}_n} ^2 \big( \hat{v}_{i} (t) \big) & = \big( \hat{v}_{i} (t) - \frac{1}{n} \mathbf{1}_n \big) ^2 \ , \\ & = \hat{v}_{i} (t) ^2 - 2 \frac{1}{n} \hat{v}_{i} (t) ' \mathbf{1}_n + \frac{1}{n} \ , \\ & = \frac{1}{n} - \hat{v}_{i} (t) ^2 \end{array}\end{equation} where the third line holds due to Lemma $\ref{aaalemdda1}$. The equivalence $(\ref{aaaerr1})$ then follows directly from $(\ref{aaaerr1p})$. \end{proof}

\begin{lem}\label{aaalemkda1}(DDA Normal Consensus Estimate Update) Applying $(\ref{dda1})$ and $(\ref{dda2})$ to the optimization problem $(\ref{aaaoptdisc})$ implies the DDA normal consensus estimate update $\hat{v}_i(t^{ij}_1(+))$ can be defined as in $(\ref{aaaddaveqa})-(\ref{aaaddaveqb})$. \end{lem} \begin{proof} Note that $(\ref{aaaddaveqa})-(\ref{aaaddaveqb})$ implies, \begin{equation}\label{aaaddaveq} \hat{v}_i ( t^{ij}_1 (+) ) = \left\{ \begin{array}{l l} \hat{v} _i ( t^{ij}_1 ) + \hat{v}_j( t^{ij}_0 ) - \hat{v}_{ji} (t^{ij}_0 ) e_i \ & , \ \mbox{if $\hat{v}^{-i} _i ( t^{ij}_1 ) ' \hat{v}^{-i} _j ( t^{ij}_0 ) = 0$,} \\ \hat{v} _j ( t^{ij}_0 ) + e_i \big( \frac{1}{n} - \hat{v}_{ji} (t^{ij}_0 ) \big) \ & , \ \mbox{if $\hat{v}^{-i} _i ( t^{ij}_1 ) ' \hat{v}^{-i} _j ( t^{ij}_0 ) > 0 \ , \ \hat{v}^{-i} _i ( t^{ij}_1 ) ^2 < \hat{v}^{-i} _j ( t^{ij}_0 )^2 \ , $} \\ \hat{v} _i ( t^{ij}_1 ) \ & , \ \mbox{if $\hat{v}^{-i} _i ( t^{ij}_1 ) ' \hat{v}^{-i} _j ( t^{ij}_0 ) > 0 \ , \ \hat{v}^{-i} _i ( t^{ij}_1 ) ^2 \geq \hat{v}^{-i} _j ( t^{ij}_0 )^2 \ . $} \end{array} \right. \end{equation} We thus will prove that $(\ref{dda1})$ and $(\ref{dda2})$ imply optimization problem $(\ref{aaaoptdisc})$ yields the DDA normal consensus estimate update $\hat{v}_i(t^{ij}_1(+))$ defined by $(\ref{aaaddaveq})$. Under $(\ref{dda1})$ and $(\ref{dda2})$ we can re-write $(\ref{aaaoptdisc})$ as, \begin{equation}\label{aaaoptkda1}\begin{array}{llll} \hat{v}_i ( t^{ij}_1(+)) = & \arg_{\mbox{$\tilde{v}$}} \ \min \big( \tilde{v} - \frac{1}{n} \mathbf{1}_n \big)^2 , \\ & \mbox{s.t. $(\ref{a10})$ holds, $\tilde{v} \in \mathbb{R}^n_{0,\frac{1}{n}}$, given $\hat{s}_i(t^{ij}_1) = \mathbf{S} \hat{v}_i(t^{ij}_1)$,} \\ & \mbox{   $\hat{s}_j(t^{ij}_0) = \mathbf{S} \hat{v}_j(t^{ij}_0), s_i(0) = \mathbf{S} e_i.$} \end{array}\end{equation} Given that $\hat{s}_i(t^{ij}_1) = \mathbf{S} \hat{v}_i(t^{ij}_1) , \hat{s}_j(t^{ij}_0) = \mathbf{S} \hat{v}_j(t^{ij}_0)$, and $s_i(0) = \mathbf{S} e_i$ are known, the set of vectors $\hat{v}_i(t^{ij}_1(+))$ for which $(\ref{a10})$ holds is $\mbox{span} \{ \hat{v}_i(t^{ij}_1) , \hat{v}_j(t^{ij}_0) , e_i \}$, thus $(\ref{aaaoptkda1})$ can be re-written as \begin{equation}\label{aaaoptkda3}\begin{array}{llll} \hat{v}_i ( t^{ij}_1(+)) & = \arg_{\mbox{$\tilde{v}$}} \ \min_{\mbox{$\tilde{v} \in \mbox{span} \{ \hat{v}_i(t^{ij}_1) , \hat{v}_j(t^{ij}_0) , e_i \} \cap \mathbb{R}^n_{0,\frac{1}{n}}$}} \ \big( \tilde{v} - \frac{1}{n} \mathbf{1}_n \big)^2 , \\ & = \arg_{\mbox{$\tilde{v} = \big( a \hat{v}_i(t^{ij}_1) + b \hat{v}_j(t^{ij}_0) + c \frac{1}{n} e_i \big) \cap \mathbb{R}^n_{0,\frac{1}{n}}$}} \ \min_{(a,b,c)} \ \big( \tilde{v} - \frac{1}{n} \mathbf{1}_n \big)^2 \ . \end{array}\end{equation} For notational convenience denote $\hat{v}_i = \hat{v}_i ( t^{ij}_1)$, $\hat{v}_j = \hat{v}_j ( t^{ij}_0)$, and $\hat{v}_i(t^{ij}_1(+)) = \hat{v}_i(+)$. Note that the constraint in $(\ref{aaaoptkda3})$ can be expressed as follows, \begin{equation}\label{aaacons1} a \hat{v}_{i \ell} + b \hat{v}_{j \ell} + c \frac{1}{n} e_{i \ell} \in ( 0, \frac{1}{n} ) \ , \ \forall \ \ell \in \mathcal{V} . \end{equation} Due to Lemma $\ref{aaalemdda2}$ the $i^{th}$ constraint in $(\ref{aaacons1})$ is, $$ \begin{array}{llll} & a \frac{1}{n} + b \hat{v}_{j i} + c \frac{1}{n} \in ( 0, \frac{1}{n} ) \ , \ \Rightarrow \\ & c = \left\{ \begin{array}{l l} & \ \ - a + b \hat{v}_{j i} n \ , \\ & 1 - a - b \hat{v}_{j i} n \ . \end{array} \right. \end{array} $$ If $c = - a + b \hat{v}_{j i} n$ we then have the candidate update $\hat{v}^{(1)}_i(+)$, $$ \begin{array}{llll} \hat{v}^{(1)}_i (+) & = a \hat{v}_i + b \hat{v}_j + \frac{1}{n} e_i ( - a + b \hat{v}_{j i} n ) \ , \ \Rightarrow \\ \hat{v}^{(1)}_i(+) ^2 & = \frac{1}{n} \hat{v}^{(1)}_i(+) ' \mathbf{1}_n \\ & = a \hat{v}_i^2 + b \hat{v}_j ^2 - a \frac{1}{n^2} - b \hat{v}_{ji} / n \ , \end{array} $$ where the magnitude $\hat{v}^{(1)}_i(+) ^2$ is computed using Lemma $\ref{aaalemdda1}$. If $c = 1 - a + b \hat{v}_{j i} n$ we then have the candidate update $\hat{v}^{(2)}_i(+)$, \begin{equation}\label{aaacan2}\begin{array}{llll} \hat{v}^{(2)}_i(+) & = a \hat{v}_i + b \hat{v}_j + \frac{1}{n} e_i ( 1 - a + b \hat{v}_{j i} n ) \ , \ \Rightarrow \\ \hat{v}^{(2)}_i(+) ^2 & = \frac{1}{n} \hat{v}^{(1)}_i(+) ' \mathbf{1}_n \\ \quad \quad \ \ & = a \hat{v}_i^2 + b \hat{v}_j ^2 + \frac{1}{n^2} - a \frac{1}{n^2} - b \hat{v}_{ji}/n \ , \\ & = \hat{v}^{(1)}_i(+) ^2 + \frac{1}{n^2} \ . \end{array} \end{equation} Since $e_{i \ell} = 0$ for all $\ell \neq i$, the $i^{th}$ constraint of $(\ref{aaacons1}$ is the only constraint involving the optimization variable $c$, thus by applying Lemma $\ref{aaalemdda01}$ it follows from $(\ref{aaacan2})$ that if $c \neq 1 - a + b \hat{v}_{j i} n$ then the resulting solution $\hat{v}_i(+)$ cannot be a solution to $(\ref{aaaoptkda3})$, hence $\hat{c} = 1 - a + b \hat{v}_{j i} n$. Next observe that if $\hat{v}_{i \ell} = \hat{v}_{j \ell} = 0$ then the $\ell ^{th}$ constraint in $(\ref{aaacons1})$ places no restriction on $a$ or $b$, thus due to Lemma $\ref{aaalemdda0}$ and Lemma $\ref{aaalemdda2}$ we can consider the three possible scenarios posed by $(\ref{aaacons1})$ given the $i^{th}$ constraint is satisfied by $(\ref{aaacan2})$: one constraint, $$ \begin{array}{llll} & (M1A) \ \ a (0) + b( \frac{1}{n} ) \in (0, \frac{1}{n}) \ \Rightarrow \ b \in (0, 1) \ , \\ & (M1B) \ \ a (\frac{1}{n}) + b( \frac{1}{n} ) \in (0, \frac{1}{n}) \ \Rightarrow \ b \in (1 -a, -a) \ , \end{array} $$ two constraints, $$ \begin{array}{llll} & (M2A) \ \ a (0) + b( \frac{1}{n} ) \in (0, \frac{1}{n}) \ \Rightarrow \ b \in (0, 1) \\ & \ \ \ \ \ \ \ \ \ \ \ a (\frac{1}{n}) + b( 0 ) \in (0, \frac{1}{n}) \ \Rightarrow \ a \in (0, 1) \ , \\ & (M2B) \ \ a (0) + b( \frac{1}{n} ) \in (0, \frac{1}{n}) \ \Rightarrow \ b \in (0, 1) \\ & \ \ \ \ \ \ \ \ \ \ \ a (\frac{1}{n}) + b( \frac{1}{n} ) \in (0, \frac{1}{n}) \ \Rightarrow \ a \in (1 -b, -b) \ , \\ & (M2C) \ \ a ( \frac{1}{n}) + b( 0 ) \in (0, \frac{1}{n}) \ \Rightarrow \ a \in (0, 1) \\ & \ \ \ \ \ \ \ \ \ \ \ a (\frac{1}{n}) + b( \frac{1}{n} ) \in (0, \frac{1}{n}) \ \Rightarrow \ b \in (1 -a, -a) \ , \end{array} $$ or three constraints $$ \begin{array}{llll} & (M3) \ \ a (0) + b( \frac{1}{n} ) \in (0, \frac{1}{n}) \ \Rightarrow \ b \in (0, 1) \\ & \ \ \ \ \ \ \ \ \ a (\frac{1}{n} ) + b( 0 ) \in (0, \frac{1}{n}) \ \Rightarrow \ a \in (0, 1) \\ & \ \ \ \ \ \ \ \ \ a (\frac{1}{n}) + b( \frac{1}{n} ) \in (0, \frac{1}{n}) \ \Rightarrow \ b \in (1 -a, -a) \ . \end{array} $$ If $\hat{v}^{-i}_i \ ' \hat{v}^{-i}_j = 0$ and there is one constraint, then Lemma $\ref{aaalemdda2}$ implies $\hat{v}^{-i}_i = 0$ and hence $(M1A)$. In this case we have $\hat{v}_i = \frac{1}{n} e_i$ and the following candidate solutions, \begin{equation}\label{aaadd1} \hat{v}_i(+) = \left\{ \begin{array}{l l} a \frac{1}{n} e_i + \frac{1}{n} e_i (1 - a) = \frac{1}{n} e_i , \ \ \ \ \ \ \ (b = 0) \\ a \frac{1}{n} e_i + \hat{v}_j + \frac{1}{n} e_i (1 - a - \hat{v}_{ji} n) , \ (b = 1) \end{array} \right. \end{equation} where the second line in $(\ref{aaadd1})$ simplifies to $\hat{v}_j + \frac{1}{n} e_i (1 - \hat{v}_{ji} n )$. The solutions in $(\ref{aaadd1})$ possess the following magnitudes, $$ \hat{v}_i(+) ^2 = \left\{ \begin{array}{l l} & \frac{1}{n^2} \ , \ \ \ \ \ \ \ \ \ \ \ \ (b = 0) \\ & \hat{v}^{-i}_j \ ^2 + \frac{1}{n^2} \ , \ \ (b = 1) \end{array} \right. $$ and thus since Lemma $\ref{aaalemdda2}$ implies $\hat{v}^{-i}_j \ ^2 >0$, the optimal solution is, \begin{equation}\label{aaadd1b}\begin{array}{llll} \hat{v}_i(+) & = \hat{v}_j + \frac{1}{n} e_i ( 1 - \hat{v}_{ji} n ) \ , \\ & = \hat{v}_i + \hat{v}_j - e_i \hat{v}_{ji} \ , \end{array} \end{equation} where the last line follows under the given assumption $\hat{v}_i = \frac{1}{n} e_i$. If $\hat{v}^{-i}_i \ ' \hat{v}^{-i}_j = 0$ and there are two constraints then $(M2A)$ necessarily follows. In this case $\hat{v}^{-i}_i \ ^2 >0$ and we have the candidate solutions, \begin{equation}\label{aaadd2} \hat{v}_i(+) = \left\{ \begin{array}{l l} \frac{1}{n} e_i \ , \ \ \ \ \ \ \ \ \ \ \ \ \ \ \ \ \ \ \ \ ( a = 0, \ b = 0) \\ \hat{v}_i \ , \ \ \ \ \ \ \ \ \ \ \ \ \ \ \ \ \ \ \ \ \ \ ( a= 1, \ b = 0) \\ \hat{v}_j + \frac{1}{n} e_i (1 - \hat{v}_{ji} n ) \ , \ \ (a = 0 , \ b = 1) \\ \hat{v}_i + \hat{v}_j - e_i \hat{v}_{ji} \ , \ \ \ \ \ \ \ (a = 1 , \ b = 1) \ . \end{array} \right. \end{equation} The solutions in $(\ref{aaadd2})$ possess the following magnitudes, $$ \hat{v}_i(+) ^2 = \left\{ \begin{array}{l l} & \frac{1}{n^2} \ , \ \ \ \ \ \ \ \ \ \ \ \ \ \ \ \ \ \ \ \ \ \ \ (a = 0 , \ b = 0) \\ & \hat{v}^{-i}_i \ ^2 + \frac{1}{n^2} \ , \ \ \ \ \ \ \ \ \ \ \ \ \ (a = 1, \ b = 0) \\ & \hat{v}^{-i}_j \ ^2 + \frac{1}{n^2} \ , \ \ \ \ \ \ \ \ \ \ \ \ \ (a = 0, \ b = 1) \\ & \hat{v}^{-i}_i \ ^2 + \hat{v}^{-i}_j \ ^2 + \frac{1}{n^2} \ , \ \ \ (a = 1, \ b = 1) \end{array} \right. $$ and thus since Lemma $\ref{aaalemdda2}$ implies $\hat{v}^{-i}_j \ ^2 >0$, the optimal solution is, \begin{equation}\label{aaadd2b} \hat{v}_i(+) = \hat{v}_i + \hat{v}_j - e_i \hat{v}_{ji} \ . \end{equation} If $\hat{v}^{-i}_i \ ' \hat{v}^{-i}_j > 0$ and there is one constraint then we necessarily have $\hat{v}^{-i}_i = \hat{v}^{-i}_j$ and thus $(M1B)$. In this case the candidate solutions are, \begin{equation}\label{aaadd3} \hat{v}_i(+) = \left\{ \begin{array}{l l} a \hat{v}_i + ( 1 - a ) \hat{v}_j + \frac{1}{n} e_i ( 1- a - (1 - a) \hat{v}_{ji} n ) \\ a \hat{v}_i - a \hat{v}_j + \frac{1}{n} e_i (1 - a + a \hat{v}_{ji} n ) \ . \end{array} \right. \end{equation} Note that the first and second line in $(\ref{aaadd3})$ correspond respectively to when $b = 1 - a$ and $b = -a$, also note that each line simplifies respectively to $\hat{v}_j + \frac{1}{n} e_i (1 - \hat{v}_{ji} n )$ and $\frac{1}{n} e_i$. The solutions in $(\ref{aaadd3})$ thus possess the following magnitudes, $$ \hat{v}_i(+) ^2 = \left\{ \begin{array}{l l} \hat{v}^{-i}_j \ ^2 + \frac{1}{n^2} \ , \ \ (b = 1-a) \\ \frac{1}{n^2} \ , \ \ \ \ \ \ \ \ \ \ \ \ ( b = -a) \ , \end{array} \right. $$ and since Lemma $\ref{aaalemdda2}$ implies $\hat{v}^{-i}_j \ ^2 >0$, the optimal solution is, \begin{equation}\label{aaadd3b}\begin{array}{llll} \hat{v}_i(+) & = \hat{v}_j + \frac{1}{n} e_i ( 1 - \hat{v}_{ji} n ) \ , \\ & = \hat{v}_i \ , \end{array} \end{equation} where the last line follows under the assumption $\hat{v}^{-i}_i = \hat{v}^{-i}_j$. If $\hat{v}^{-i}_i \ ' \hat{v}^{-i}_j > 0$ and there are two constraints, then if $\hat{v}^{-i}_i \ ^2 < \hat{v}^{-i}_j \ ^2$ it follows that $(M2B)$ holds. The candidate solutions in this case are, \begin{equation}\label{aaadd4} \hat{v}_i(+) = \left\{ \begin{array}{l l} \frac{1}{n} e_i \ , \ \ \ \ \ \ \ \ \ \ \ \ \ \ \ \ \ \ \ \ \ \ \ \ \ \ \ ( a = 0, \ b = 0) \\ \hat{v}_i \ , \ \ \ \ \ \ \ \ \ \ \ \ \ \ \ \ \ \ \ \ \ \ \ \ \ \ \ \ \ ( a= 1, \ b = 0) \\ - \hat{v}_i + \hat{v}_j - \frac{1}{n} e_i (2 - \hat{v}_{ji} n ) \ , \ (a = -1 , \ b = 1) \\ \hat{v}_j + \frac{1}{n} e_i (1 - \hat{v}_{ji} n ) \ , \ \ \ \ \ \ \ \ \ (a = 0 , \ b = 1) . \end{array} \right. \end{equation} The solutions in $(\ref{aaadd4})$ possess the following magnitudes, $$ \hat{v}_i(+) ^2 = \left\{ \begin{array}{l l} & \frac{1}{n^2} \ , \ \ \ \ \ \ \ \ \ \ \ \ \ \ \ \ \ \ \ \ \ \ \ (a = 0 , \ b = 0) \\ & \hat{v}^{-i}_i \ ^2 + \frac{1}{n^2} \ , \ \ \ \ \ \ \ \ \ \ \ \ \ (a = 1, \ b = 0) \\ & - \hat{v}^{-i}_i \ ^2 + \hat{v}^{-i}_j \ ^2 + \frac{1}{n^2} \ , \ (a = -1, \ b = 1) \\ & \hat{v}^{-i}_j \ ^2 + \frac{1}{n^2} \ , \ \ \ \ \ \ \ \ \ \ \ \ \ (a = 0, \ b = 1) . \end{array} \right. $$ The assumption $\hat{v}^{-i}_j \ ^2 > \hat{v}^{-i}_i \ ^2$ then implies the optimal solution, \begin{equation}\label{aaadd4b} \hat{v}_i(+) = \hat{v}_j - \frac{1}{n} e_i ( 1 - \hat{v}_{ji} n) \ . \end{equation} If $\hat{v}^{-i}_i \ ' \hat{v}^{-i}_j > 0$ and there are two constraints, then if $\hat{v}^{-i}_i \geq \hat{v}^{-i}_j$ it follows that $(M2C)$ holds. The candidate solutions in this case are, \begin{equation}\label{aaadd5} \hat{v}_i(+) = \left\{ \begin{array}{l l} \frac{1}{n} e_i \ , \ \ \ \ \ \ \ \ \ \ \ \ \ \ \ \ \ \ \ ( a = 0, \ b = 0) \\ \hat{v}_j + \frac{1}{n} e_i (1 - \hat{v}_{ji} n ) \ , \ (a = 0 , \ b = 1) \\ \hat{v}_i - \hat{v}_j + \frac{1}{n} e_i \ , \ \ \ \ \ \ \ (a = 1 , \ b = -1) \\ \hat{v}_i \ , \ \ \ \ \ \ \ \ \ \ \ \ \ \ \ \ \ \ \ \ \ ( a= 1, \ b = 0) \ . \end{array} \right. \end{equation} The solutions in $(\ref{aaadd5})$ possess the following magnitudes, $$ \hat{v}_i(+) ^2 = \left\{ \begin{array}{l l} \frac{1}{n^2} \ , \ \ \ \ \ \ \ \ \ \ \ \ \ \ \ \ \ \ \ \ \ (a = 0 , \ b = 0) \\ \hat{v}^{-i}_j \ ^2 + \frac{1}{n^2} \ , \ \ \ \ \ \ \ \ \ \ \ (a = 0, \ b = 1) \\ \hat{v}^{-i}_i \ ^2 - \hat{v}^{-i}_j \ ^2 + \frac{1}{n^2} \ , \ (a = 1, \ b = -1) \\ \hat{v}^{-i}_i \ ^2 + \frac{1}{n^2} \ , \ \ \ \ \ \ \ \ \ \ \ (a = 1, \ b = 0) . \end{array} \right. $$ The assumption $\hat{v}^{-i}_j \ ^2 \leq \hat{v}^{-i}_i \ ^2$ then implies the global solution, \begin{equation}\label{aaadd5b} \hat{v}_i(+) = \hat{v}_i \ . \end{equation} If $\hat{v}^{-i}_i \ ' \hat{v}^{-i}_j > 0$ and there are three constraints then $(M3)$ necessarily follows. The candidate solutions in this case are, \begin{equation}\label{aaadd6} \hat{v}_i(+) = \left\{ \begin{array}{l l} \frac{1}{n} e_i \ , \ \ \ \ \ \ \ \ \ \ \ \ \ \ \ \ \ \ \ ( a = 0, \ b = 0) \\ \hat{v}_j + \frac{1}{n} e_i (1 - \hat{v}_{ji} n ) \ , \ (a = 0 , \ b = 1) \\ \hat{v}_i \ , \ \ \ \ \ \ \ \ \ \ \ \ \ \ \ \ \ \ \ \ \ ( a= 1, \ b = 0) . \end{array} \right. \end{equation} The solutions in $(\ref{aaadd6})$ possess the following magnitudes, $$ \hat{v}_i(+) ^2 = \left\{ \begin{array}{l l} \frac{1}{n^2} \ , \ \ \ \ \ \ \ \ \ \ \ (a = 0 , \ b = 0) \\ \hat{v}^{-i}_j \ ^2 + \frac{1}{n^2} \ , \ (a = 0, \ b = 1) \\ \hat{v}^{-i}_i \ ^2 + \frac{1}{n^2} \ , \ (a = 1, \ b = 0) . \end{array} \right. $$ The assumption $\hat{v}^{-i}_j \ ^2 > \hat{v}^{-i}_i \ ^2$ then implies the optimal solution, \begin{equation}\label{aaadd6ba} \hat{v}_i(+) = \hat{v}_j + \frac{1}{n} e_i (1 - \hat{v}_{ji} n ) \ . \end{equation} In contrast, the assumption $\hat{v}^{-i}_j \ ^2 \leq \hat{v}^{-i}_i \ ^2$ implies the global solution, \begin{equation}\label{aaadd6bb} \hat{v}_i(+) = \hat{v}_i \ . \end{equation} Combining $(\ref{aaadd1b})$, $(\ref{aaadd2b})$, $(\ref{aaadd3b})$, $(\ref{aaadd4b})$, $(\ref{aaadd5b})$, $(\ref{aaadd6ba})$, and $(\ref{aaadd6bb})$, together imply $(\ref{aaaddaveq})$. Note that if $\hat{v}^{-i}_j \ ^2 = \hat{v}^{-i}_i \ ^2$ and $\hat{v}^{-i}_j \neq \hat{v}^{-i}_i$ then there are two global solutions to $(\ref{aaaoptkda3})$. By specifying $(\ref{aaadd5b})$ and $(\ref{aaadd6bb})$ we have chosen the global solutions that are necessary for Lemma $\ref{aaalemdda3}$, which is in turn necessary for the DDA algorithm to obtain average-consensus under the sufficient communication condition stated in Theorem $\ref{aaathm5}$. \end{proof}

\begin{lem}\label{aaalemdda0w1}(DDA Normal Consensus Estimate Non-Decreasing Error) The error of every normal consensus estimate $\hat{v}_i(t)$ is a non-decreasing function of $t \geq 0$ for all $i \in \mathcal{V}$. \end{lem} \begin{proof} The Lemma $\ref{aaalemkda1}$ implies that upon reception of any signal $S^{ij}( t^{ij}_0, t^{ij}_1)$, the normal consensus estimate $\hat{v}_i(t^{ij}_1)$ is updated to a solution of $(\ref{aaaoptkda3})$, thus \begin{equation}\label{neweq1} \hat{v}_i(t^{ij}_1(+)) \in \mbox{span} \{ \hat{v}_i(t^{ij}_1), \hat{v}_j(t^{ij}_0), e_i \} \cap \mathbb{R}^n_{0,\frac{1}{n}} \ . \end{equation} The Lemma $\ref{aaalemdda0}$ implies $\hat{v}_i(t^{ij}_1) \in \mathbb{R}^n_{0,\frac{1}{n}}$, thus $(\ref{neweq1})$ implies that any candidate solution $\hat{v}^{(1)}(t^{ij}_1(+))$ that does not satisfy $E_{\frac{1}{n} \mathbf{1}_n} \big(\hat{v}^{(1)}(t^{ij}_1(+)) \big) \leq E_{\frac{1}{n} \mathbf{1}_n} \big(\hat{v}(t^{ij}_1) \big)$ cannot be a solution $(\ref{aaaoptkda3})$. \end{proof}

\begin{lem}\label{aaalemdda3}(DDA Fixed Normal Consensus Estimate Given No Reduction in Error) Under the DDA algorithm, for any signal $S^{i j} ( t^{i j}_0, t^{i j}_1 )$ we have, \begin{equation}\label{aaaimpl1} E_{\frac{1}{n}}^2 \big( \hat{v}_i ( t^{ij}_1 (+) ) \big) = E_{\frac{1}{n}}^2 \big( \hat{v}_i ( t^{ij}_1 ) \big) \ \Rightarrow \ \hat{v}_i ( t^{ij}_1 (+) ) = \hat{v}_i ( t^{ij}_1 ) . \end{equation} \end{lem} \begin{proof} Applying Lemma $\ref{aaalemdda01}$ to the LHS of $(\ref{aaaimpl1})$ implies, $$ \begin{array}{llll} & E_{\frac{1}{n}}^2 \big( \hat{v}_i ( t^{ij}_1 (+) ) \big) = E_{\frac{1}{n}}^2 \big( \hat{v}_i ( t^{ij}_1 ) \big) \\ & \ \ \ \ \ \ \ \Leftrightarrow \ \ \big( {v}_i ( t^{ij}_1 (+) ) \big)^2 = \big( \hat{v}_i ( t^{ij}_1 ) \big)^2 \ . \end{array} $$ To prove $(\ref{aaaimpl1})$ we thus have only to show, \begin{equation}\label{aaacimpl1} \hat{v}_i ( t^{ij}_1 (+) ) \neq \hat{v}_i ( t^{ij}_1 ) \Rightarrow \big( \hat{v}_i ( t^{ij}_1 (+) ) \big)^2 \neq \big( \hat{v}_i ( t^{ij}_1 ) \big)^2 \ . \end{equation} Under the update $(\ref{aaaddaveq})$, to prove $(\ref{aaacimpl1})$ we need only to show that either of the two cases, \begin{equation}\label{aaacimpl2}\begin{array}{llll} & \hat{v}_i ( t^{ij}_1 (+) ) = \hat{v} _i ( t^{ij}_1 ) + \hat{v}_j( t^{ij}_0 ) - \hat{v}_{ji} (t^{ij}_0 ) e_i \ , \\ & \ \ \ \ \ \ \ \ \ \ \ \mbox{if $\hat{v}^{-i} _i ( t^{ij}_1 ) ' \hat{v}^{-i} _j ( t^{ij}_0 ) = 0$,} \\ & \hat{v}_i ( t^{ij}_1 (+) ) = \hat{v} _j ( t^{ij}_0 ) + e_i \big( \frac{1}{n} - \hat{v}_{ji} (t^{ij}_0 ) \big) \ , \\ & \ \ \ \ \ \ \ \ \ \ \ \mbox{if $\hat{v}^{-i} _i ( t^{ij}_1 ) ' \hat{v}^{-i} _j ( t^{ij}_0 ) > 0 , \hat{v}^{-i} _i ( t^{ij}_1 ) ^2 < \hat{v}^{-i} _j ( t^{ij}_0 )^2 $} \end{array} \end{equation} will imply $\hat{v}^2_{i} (t^{ij}_1(+) ) > \hat{v}^2_{i} (t^{ij}_1 )$. If $\hat{v}^{-i} _i ( t^{ij}_1 ) ' \hat{v}^{-i} _j ( t^{ij}_0 ) = 0$ then $(\ref{aaacimpl2})$ together with $(\ref{norm1})$ imply, $$ \begin{array}{llll} \hat{v}^2_{i} (t^{ij}_1(+) ) & = \frac{1}{n} \hat{v}_{i} (t^{ij}_1(+) ) ' \mathbf{1}_n \ , \\ & = \frac{1}{n} \big( \hat{v} _i ( t^{ij}_1 ) + \hat{v}_j( t^{ij}_0 ) - \hat{v}_{ji} (t^{ij}_0 ) e_i \big) ' \mathbf{1}_n \ , \\ & = \frac{1}{n} \big( \hat{v} _i ( t^{ij}_1 )' \mathbf{1}_n + \hat{v} _j ( t^{ij}_0 ) ' \mathbf{1}_n - \hat{v} _{ji} ( t^{ij}_0 ) \big) \ , \\ & = \hat{v}^2 _i ( t^{ij}_1 ) + \frac{1}{n} \hat{v}^{-i} _j ( t^{ij}_0 ) ' \mathbf{1}_{n-1} \ , \\ & > \hat{v}^2 _i ( t^{ij}_1 ) \ , \end{array} $$ where the last inequality holds because $(\ref{pp1})$ implies $\hat{v}^{-i}_j(t) ' \mathbf{1}_n \geq \frac{1}{n}$ for all $j \in \mathcal{V}_{-i}$ and $t \geq 0$. If $\hat{v}^{-i} _i ( t^{ij}_1 ) ' \hat{v}^{-i} _j ( t^{ij}_0 ) > 0$ and $\hat{v}^{-i} _i ( t^{ij}_1 ) ^2 < \hat{v}^{-i} _j ( t^{ij}_0 )^2$, then due to $(\ref{aaacimpl2})$ it follows that, $$ \begin{array}{llll} \hat{v}^2_{i} (t^{ij}_1(+) ) & = \frac{1}{n} \hat{v}_{i} (t^{ij}_1(+) ) ' \mathbf{1}_n \ , \\ & = \frac{1}{n} \Big( \hat{v} _j ( t^{ij}_0 ) + e_i \big( \frac{1}{n} - \hat{v}_{ji} (t^{ij}_0 ) \big) \Big) ' \mathbf{1}_n \ , \\ & = \frac{1}{n} \big( \hat{v} _j ( t^{ij}_0 ) ' \mathbf{1}_n + \frac{1}{n} - \hat{v}_{ji} (t^{ij}_0) \big) \ , \\ & = \frac{1}{n} \big( \hat{v}^{-i} _j ( t^{ij}_0 ) ' \mathbf{1}_{n-1} + \frac{1}{n} \big) \ , \\ & > \frac{1}{n} \big( \hat{v}^{-i} _i ( t^{ij}_1 ) ' \mathbf{1}_{n-1} + \frac{1}{n} \big) \ , \\ & = \frac{1}{n} \hat{v} _i ( t^{ij}_1 ) ' \mathbf{1}_{n} = \hat{v}^2_i ( t^{ij}_1 ) \ , \end{array} $$ where the final inequality holds under the given assumption that $\hat{v}^{-i} _i ( t^{ij}_1 ) ^2 < \hat{v}^{-i} _j ( t^{ij}_0 )^2$, and the second last equality holds due to Lemma $\ref{aaalemdda2}$. \end{proof}

\begin{lem}\label{aaalemdda4}(DDA Lower Bound on Increase in Normal Consensus Magnitude) Under the DDA algorithm, for any signal $S^{ij} ( t^{i j}_0, t^{i j}_1 )$ we have, $$ \big( \hat{v}^{-i}_i ( t^{ij}_1 (+) ) \big) ^2 \geq \max \{ \big( \hat{v}^{-i}_i ( t^{ij}_1 ) \big) ^2 \ , \ \big( \hat{v}^{-i}_j ( t^{ij}_0 ) \big) ^2 \} \ . $$ \end{lem} \begin{proof} If $\hat{v}^{-i} _i ( t^{ij}_1 ) ' \hat{v}^{-i} _j ( t^{ij}_0 ) = 0$ then $\hat{v}_i ( t^{ij}_1 (+) ) = \hat{v} _i ( t^{ij}_1 ) + \hat{v}_j( t^{ij}_0 ) - \hat{v}_{ji} (t^{ij}_0 ) e_i$ and thus applying $(\ref{aaanorm1a})$ implies, $$ \begin{array}{llll} \big( \hat{v}^{-i}_{i} ( t^{ij}_1 (+) ) \big)^2 & = \frac{1}{n} \hat{v}^{-i}_{i} ( t^{ij}_1 (+) ) ' \mathbf{1}_{n-1} \ , \\ & = \frac{1}{n} \big( \hat{v}^{-i} _i ( t^{ij}_1 ) ' \mathbf{1}_{n-1} + \hat{v}^{-i}_j( t^{ij}_0 ) ' \mathbf{1}_{n-1} \big) \ , \\ & \geq \max \{ \big( \hat{v}^{-i}_i ( t^{ij}_1 ) \big) ^2 \ , \ \big( \hat{v}^{-i}_j ( t^{ij}_0 ) \big) ^2 \} \ . \end{array} $$ If $\hat{v}^{-i} _i ( t^{ij}_1 ) ' \hat{v}^{-i} _j ( t^{ij}_0 ) > 0$ and $\hat{v}^{-i} _i ( t^{ij}_1 ) ^2 < \hat{v}^{-i} _j ( t^{ij}_0 )^2$ then $\hat{v}_i ( t^{ij}_1 (+) ) = \hat{v} _j ( t^{ij}_0 ) + e_i \big( \frac{1}{n} - \hat{v}_{ji} (t^{ij}_0 ) \big)$ and thus, $$ \begin{array}{llll} \big( \hat{v}^{-i}_i ( t^{ij}_1 (+) ) \big) ^2 & = \frac{1}{n} \Big( \hat{v}^{-i} _j ( t^{ij}_0 ) + e^{-i}_i \big( \frac{1}{n} - \hat{v}_{ji} (t^{ij}_0 ) \big) \Big) \mathbf{1}_{n-1} \\ & = \frac{1}{n} \hat{v}^{-i} _j ( t^{ij}_0 ) ' \mathbf{1}_{n-1} \\ & \geq \max \{ \big( \hat{v}^{-i}_i ( t^{ij}_1 ) \big) ^2 \ , \ \big( \hat{v}^{-i}_j ( t^{ij}_0 ) \big) ^2 \} \end{array} $$ where the last inequality holds by the assumption $\hat{v}^{-i} _i ( t^{ij}_1 ) ^2 < \hat{v}^{-i} _j ( t^{ij}_0 )^2$. Finally, if $\hat{v}^{-i} _i ( t^{ij}_1 ) ' \hat{v}^{-i} _j ( t^{ij}_0 ) > 0$ and $\hat{v}^{-i} _i ( t^{ij}_1 ) ^2 \geq \hat{v}^{-i} _j ( t^{ij}_0 )^2$ then $\hat{v}_i ( t^{ij}_1 (+) ) = \hat{v} _i ( t^{ij}_1 )$ and thus $$ \begin{array}{llll} \big( \hat{v}^{-i}_{i} ( t^{ij}_1 (+) ) \big)^2 & = \frac{1}{n} \hat{v}^{-i}_{i} ( t^{ij}_1 (+) ) ' \mathbf{1}_{n-1} \ , \\ & = \frac{1}{n} \hat{v}^{-i}_{i} ( t^{ij}_1 ) ' \mathbf{1}_{n-1} \ , \\ & \geq \max \{ \big( \hat{v}^{-i}_i ( t^{ij}_1 ) \big) ^2 \ , \ \big( \hat{v}^{-i}_j ( t^{ij}_0 ) \big) ^2 \} \ , \end{array} $$ where the last inequality holds by the assumption $\hat{v}^{-i} _i ( t^{ij}_1 ) ^2 \geq \hat{v}^{-i} _j ( t^{ij}_0 )^2$. \end{proof}

\begin{lem}\label{aaalemdda5}(DDA Upper Bound on the Magnitude $\big( \hat{v}^{-i}_i ( t) \big) ^2$) If $ \hat{v}^2_i ( t) = \hat{v}^2_i ( t') $ for all $(t, t') \in [t_0, t_1]$ then $$ \big( \hat{v}^{-i}_i ( t ) \big) ^2 \leq \big( \hat{v}^{-j}_i ( t' ) \big) ^2 \ , \ \forall \ j \in \mathcal{V}_{-i} \ , \ \forall \ (t, t') \in [t_0, t_1] . $$ \end{lem} \begin{proof} The result follows since $\hat{v}_i ( t ) \in \mathbb{R}^n_{0,\frac{1}{n}}$ implies $\hat{v}_{ij}(t) \in (0, \frac{1}{n})$ for all $j \in \mathcal{V}$ and $t \geq 0$, thus, $$ \begin{array}{llll} \big( \hat{v}^{-i}_i ( t ) \big) ^2 & = \hat{v}^2_i ( t ) - \frac{1}{n^2} \ , \\ & = \hat{v}^2_i ( t' ) - \frac{1}{n^2} \ , \\ & \leq \hat{v}^2_i ( t' ) - \hat{v}_{ij} ^2 (t') \ , \\ & = \big( \hat{v}^{-j}_i ( t' ) \big) ^2 \ . \end{array} $$ \end{proof}

\begin{lem}\label{aaalemdda4a}(Error Expression for $C_{[0,t_1]}$ satisfying $(\ref{star1g})$) For any communication sequence $C_{[0,t_1]}$ satisfying $(\ref{star1g})$ the total reduction in normal consensus error from $t=0$ to $t=t_1(+)$ is, \begin{equation}\label{aaaerrred}\begin{array}{llll} \mathbf{E}^2 ( C_{[0,t_1]} ) & = \sum_{i=1}^n \Big( E^2_{\frac{1}{n} \mathbf{1}_n} \big( \hat{v}_i(0) \big) - E^2_{\frac{1}{n} \mathbf{1}_n} \big( \hat{v}_i(t_1(+)) \big) \Big) \\ & = \frac{n-1}{n} - \sum_{i=1}^n E^2_{\frac{1}{n} \mathbf{1}_n} \big( \hat{v}_i(t_1(+)) \big) \\ & = \sum_{\ell \in \mathbb{N}} \mathbf{E}^2 ( C_{[t^{\ell}_0,t^{\ell}_1 ]} ) \ , \ C_{[t^{\ell}_0,t^{\ell}_1 ]} \in \mbox{S$\mathcal{V}$CC} \\ & \leq \frac{n-1}{n} \ . \end{array}\end{equation} \end{lem} \begin{proof} The proof is identical to Lemma $\ref{lemkda4a}$ with $(\ref{errred})$, $(\ref{star1ga})$, and $\mbox{S$\mathcal{V}$SC}$ replaced by $(\ref{aaaerrred})$, $(\ref{star1g})$, and $\mbox{S$\mathcal{V}$CC}$ respectively. \end{proof}

\begin{lem}\label{aaalemdda4b}(DDA Vanishing Reduction in Error for $C_{[0,t_1]}$ satisfying $(\ref{star1g})$) For any communication sequence $C_{[0,t_1]}$ satisfying $(\ref{star1g})$ there exists an integer $\ell_\varepsilon$ such that, $$ \mathbf{E}^2 ( C_{[t^{\ell}_0,t^{\ell}_1 ]} ) \leq \varepsilon \ , \ \forall \ \ell \geq \ell_\varepsilon \ , $$ for any $\varepsilon > 0$. \end{lem} \begin{proof} The proof is identical to Lemma $\ref{lemkda4b}$ with $(\ref{errred})$ and $(\ref{star1ga})$ replaced by $(\ref{aaaerrred})$ and $(\ref{star1g})$ respectively. \end{proof}

\begin{lem}\label{aaalemdda4c}(DDA Lower Bound on Non-Zero Reduction in Error) For any communication sequence $C_{[t_0,t_1]}$ we have, \begin{equation}\label{aaanonre1} \mathbf{E}^2 ( C_{[t_0,t_1 ]} ) > 0 \ \Rightarrow \ \mathbf{E}^2 ( C_{[t_0,t_1 ]} ) \geq \frac{1}{n^2} \ . \end{equation} \end{lem} \begin{proof} Applying Lemma $\ref{aaalemdda01}$ to the LHS of $(\ref{aaanonre1})$ implies there exists some signal $S^{ij} ( t^{i j}_0, t^{i j}_1 ) \in C_{[t_0,t_1 ]}$ such that, \begin{equation}\label{aaanonre1a} \hat{v}^2_i(t^{ij}_1(+)) > \hat{v}^2_i(t^{ij}_1) \ . \end{equation} We now show that $(\ref{aaanonre1a})$ implies $\hat{v}^2_i(t^{ij}_1(+)) \geq \hat{v}^2_i(t^{ij}_1) + \frac{1}{n^2}$, the result $(\ref{aaanonre1})$
then follows directly by Lemma $\ref{aaalemdda0w1}$ together with Lemma $\ref{aaalemdda01}$. The Lemma $\ref{aaalemdda0}$ implies $\hat{v}_i(t) \in \mathbb{R}^n_{0,\frac{1}{n}}$ for all $i \in \mathcal{V}$ and $t \geq 0$, thus under $(\ref{aaanonre1a})$ it follows that, \begin{equation}\label{aaanoreb}\begin{array}{llll} \hat{v}^2_i(t^{ij}_1(+)) & = \frac{1}{n} \hat{v}_i(t^{ij}_1(+)) ' \mathbf{1}_n \ , \\ & = \frac{1}{n^2} | \{ \ell \ : \ \hat{v}_{i \ell } (t^{ij}_1(+)) \} | \ , \\ & > \hat{v}^2_i(t^{ij}_1) \ , \\ & = \frac{1}{n} \hat{v}_i(t^{ij}_1) ' \mathbf{1}_n \ , \\ & = \frac{1}{n^2} | \{ \ell \ : \ \hat{v}_{i \ell } (t^{ij}_1) \} | \ . \end{array}\end{equation} From $(\ref{aaanoreb})$ it then follows, \begin{equation}\label{aaanoreba}\begin{array}{llll} | \{ \ell \ : \ \hat{v}_{i \ell } (t^{ij}_1(+)) \} | & > | \{ \ell \ : \ \hat{v}_{i \ell } (t^{ij}_1) \} | \ \Rightarrow \\ | \{ \ell \ : \ \hat{v}_{i \ell } (t^{ij}_1(+)) \} | - 1 & \geq | \{ \ell \ : \ \hat{v}_{i \ell } (t^{ij}_1) \} | \ . \end{array}\end{equation} Under the constraint $\hat{v}_i(t) \in \mathbb{R}^n_{0,\frac{1}{n}}$ the last inequality in $(\ref{aaanoreba})$ implies $\hat{v}^2_i(t^{ij}_1(+)) \geq \hat{v}^2_i(t^{ij}_1) + \frac{1}{n^2}$. \end{proof}

\begin{lem}\label{aaalemdda4d}(DDA Existence of a Time for Zero Reduction in Error) For any communication sequence $C_{[0,t_1]}$ satisfying $(\ref{star1g})$ there exists an integer $\ell_{\frac{1}{n^2}}$ such that, \begin{equation}\label{aaanew1} \mathbf{E}^2 ( C_{[t^{\ell}_0,t^{\ell}_1 ]} ) = 0 \ , \ \forall \ \ell \geq \ell_{\frac{1}{n^2}} \ . \end{equation} \end{lem} \begin{proof} The result follows immediately by applying Lemma $\ref{aaalemdda0w1}$ and Lemma $\ref{aaalemdda4c}$ to Lemma $\ref{aaalemdda4b}$. \end{proof}

\emph{Theorem $\ref{aaathm5}$}(DDA Network Convergence to Average-Consensus) \begin{proof} For any communication sequence $C_{[0,t_1]}$ satisfying $(\ref{star1g})$ the Lemma $\ref{aaalemdda4d}$ implies there exists an integer $\ell_{\frac{1}{n^2}}$ such that $(\ref{aaanew1})$ holds. We now show that the condition, \begin{equation}\label{aaacon1} \mathbf{E}^2 ( C_{[t^{\ell}_0,t^{\ell}_1 ]} ) = 0 \ , \ C_{[t^{\ell}_0,t^{\ell}_1 ]} \in \mbox{S$\mathcal{V}$CC} \end{equation} implies $(\ref{acdef1})$ 
at $t = t^\ell_0$, and hence by Definition $\ref{def11}$ a network average-consensus is obtained at $t^\ell_0$. Applying Lemma $\ref{aaalemdda4}$ to $(\ref{aaacon1})$ implies, \begin{equation}\label{aaacon2}\begin{array}{llll} \big( \hat{v}^{-\hat{i}}_{\hat{i}} (t^{\hat{i} j}_1(+)) \big) ^2 & \geq \big( \hat{v}^{-\hat{i}}_j (t^{\hat{i} j}_0) \big) ^2 \ , \ \forall \ j \in \mathcal{V}_{- \hat{i}} \ , \\ & \geq \big( \hat{v}^{-j}_j (t^{\hat{i} j}_0) \big) ^2 \ , \ \forall \ j \in \mathcal{V}_{- \hat{i}} \end{array}\end{equation} where the last inequality holds due to Lemma $\ref{aaalemdda5}$. Likewise, applying Lemma $\ref{aaalemdda4}$ and Lemma $\ref{aaalemdda5}$ to $(\ref{aaacon1})$ implies, \begin{equation}\label{aaacon3}\begin{array}{llll} \big( \hat{v}^{-j}_j (t^{j \hat{i}}_1(+)) \big) ^2 & \geq \big( \hat{v}^{-j}_{\hat{i}} (t^{j \hat{i}}_0) \big) ^2 \ , \ \forall \ j \in \mathcal{V}_{- \hat{i}} \ , \\ & \geq \big( \hat{v}^{-\hat{i}}_{\hat{i}} (t^{j \hat{i}}_0) \big) ^2 \ , \ \forall \ j \in \mathcal{V}_{- \hat{i}} \ . \end{array}\end{equation} Note that applying Lemma $\ref{aaalemdda3}$ to $(\ref{aaacon1})$ implies that $(\ref{aaacon2})$ and $(\ref{aaacon3})$ can be combined to obtain, \begin{equation}\label{aaacon4}\begin{array}{llll} & \big( \hat{v}^{-\hat{i}}_{\hat{i}} (t^\ell_0) \big) ^2 = \big( \hat{v}^{-j}_{\hat{i}} (t^\ell_0) \big) ^2 \ , \ \forall \ j \in \mathcal{V}_{- \hat{i}} \ \ \Rightarrow \\ & \hat{v}_{\hat{i}} ^2 (t^\ell_0) - \frac{1}{n^2} = \hat{v}_{\hat{i}} ^2 (t^\ell_0) - \hat{v}_{\hat{i} j } ^2 (t^\ell_0) \ , \ \forall \ j \in \mathcal{V}_{- \hat{i}} \ . \end{array}\end{equation} The second line in $(\ref{aaacon4})$ together with Lemma $\ref{aaalemdda2}$ implies, $$ \hat{v}_{\hat{i} j } (t^\ell_0) = \frac{1}{n} \ , \ \forall \ j \in \mathcal{V} \ , $$ and thus $\hat{v}_{\hat{i}} (t^\ell_0) = \frac{1}{n} \mathbf{1}_n$. Since $C_{[t^{\ell}_0,t^{\ell}_1 ]} \in \mbox{S$\mathcal{V}$CC}$, the Lemma $\ref{aaalemdda3}$ together with the update $(\ref{aaaddaveq})$ implies $\hat{v}_j (t^\ell_0) = \frac{1}{n} \mathbf{1}_n$ for each $j \in \mathcal{V}_{-\hat{i}}$, thus $(\ref{acdef1})$ holds and so by Definition $\ref{def11}$ average-consensus is obtained at time $t^\ell_0 < t_1$. \end{proof}

\begin{rem}\label{rem1} We observe that if the $\mbox{S$\mathcal{V}$CC}$ condition is defined by the condition (C) stated in Theorem $\ref{aaacor5}$, then using the definition $(\ref{star1g})$ for the $\mbox{I$\mathcal{V}$CC}$ condition will not imply Theorem $\ref{aaathm5}$. In other words, using condition (C) to define an $\mbox{S$\mathcal{V}$CC}$ sequence will imply there exist examples of $\mbox{I$\mathcal{V}$CC}$ sequences for which the DDA algorithm will not obtain average-consensus. This is why we have defined $\mbox{S$\mathcal{V}$CC}$ only by $(\ref{symcondef2})$, which is actually a special case of the condition (C) stated in Theorem $\ref{aaacor5}$. Furthermore, the DDA algorithm normal consensus update $(\ref{aaaddaveq})$ is only a global solution to $(\ref{aaaoptdisc})$, it is not a unique solution. Under $(\ref{aaasdasig5})$ and $(\ref{aaais2})$, the alternative global solution to $(\ref{aaaoptdisc})$ is, \begin{equation}\label{aaaddaveq2} \hat{v}_i ( t^{ij}_1 (+) ) = \left\{ \begin{array}{l l} \hat{v} _i ( t^{ij}_1 ) + \hat{v}_j( t^{ij}_0 ) - \hat{v}_{ji} (t^{ij}_0 ) e_i \ & , \ \mbox{if $\hat{v}^{-i} _i ( t^{ij}_1 ) ' \hat{v}^{-i} _j ( t^{ij}_0 ) = 0 \ ,$} \\ \hat{v} _j ( t^{ij}_0 ) + e_i \big( \frac{1}{n} - \hat{v}_{ji} (t^{ij}_0 ) \big) \ & , \ \mbox{if $\hat{v}^{-i} _i ( t^{ij}_1 ) ' \hat{v}^{-i} _j ( t^{ij}_0 ) > 0 \ , \ \hat{v}^{-i} _i ( t^{ij}_1 ) ^2 < \hat{v}^{-i} _j ( t^{ij}_0 )^2 \ , $} \\ \hat{v} _i ( t^{ij}_1 ) \ & , \ \mbox{if $\hat{v}^{-i} _i ( t^{ij}_1 ) ' \hat{v}^{-i} _j ( t^{ij}_0 ) > 0 \ , \ \hat{v}^{-i} _i ( t^{ij}_1 ) ^2 > \hat{v}^{-i} _j ( t^{ij}_0 )^2 \ , $}  \\ \hat{v} _j ( t^{ij}_0 ) \ & , \ \mbox{if $\hat{v}^{-i} _i ( t^{ij}_1 ) ' \hat{v}^{-i} _j ( t^{ij}_0 ) > 0 \ , \ \hat{v}^{-i} _i ( t^{ij}_1 ) ^2 = \hat{v}^{-i} _j ( t^{ij}_0 )^2 \ . $} \end{array} \right. \end{equation} The above remark still holds even when the alternative global solution $(\ref{aaaddaveq2})$ is used to update the normal consensus estimate; however, by randomly switching between the updates $(\ref{aaaddaveq})$ and $(\ref{aaaddaveq2})$ leads to the following conjecture.

\begin{conj}\label{conject4} Let the $\mbox{S$\mathcal{V}$CC}$ condition be defined by the condition (C) stated in Theorem $\ref{aaacor5}$. Suppose upon reception of each signal the normal consensus estimate update is defined by $(\ref{aaaddaveq})$ with probability $p \in (0,1)$ and defined by $(\ref{aaaddaveq2})$ with probability $1-p$. Then $(\ref{acdef})$ holds at time $t = t_1(+)$ almost surely for any communication sequence $C_{[0,t_1]}$ satisfying $(\ref{star1g})$. \end{conj}

The significance of the above result is due to the fact that condition (C) in Theorem $\ref{aaacor5}$ is considerably weaker than $(\ref{symcondef2})$. If Conjecture $\ref{conject4}$ holds, then by defining the DDA algorithm using the above randomized protocol, and defining the $\mbox{S$\mathcal{V}$CC}$ condition by the condition (C) stated in Theorem $\ref{aaacor5}$, the Venn diagram in Fig.$\ref{simpfig2}$ will then be completely accurate. We note there are alternative definitions of the $\mbox{S$\mathcal{V}$CC}$ condition and OH algorithm that will also lead to the same Venn diagram as Fig.$\ref{simpfig2}$. However, we know of no \emph{weaker} sufficient condition than that stated in Theorem $\ref{aaathm5}$ for convergence of the DDA algorithm $(\ref{dda1})-(\ref{ddainit1})$, and this sufficient condition is based entirely on the $\mbox{S$\mathcal{V}$CC}$ condition $(\ref{symcondef2})$. \end{rem}

\subsection{Comparison Algorithms: Gossip and ARIS}\label{sec:compalg} In this section we define the two comparison algorithms, Gossip and ARIS, in terms of the class of distributed algorithms $(\ref{update0}),(\ref{update1})$.
\subsubsection{Comparison Algorithm 1 (Gossip)}\label{sec:gossip} The Gossip algorithm proposed in \cite{SH06} implies a signal specification and knowledge set update defined respectively as $(\ref{gossipsig})$ and $(\ref{gossipknow})$ below. \begin{framed} $$ \mbox{Gossip Algorithm:}  $$ \begin{equation}\label{gossipsig} \mbox{Signal Specification:} \ \ \ \ \ \ \ \ \ \ \ \ \ \ \ \ \ \ \ \ \ \ \ S^{ij}(t^{ij}_0,t^{ij}_1 ) = \mathcal{K}_j(t^{ij}_0) \ \ \ \ \ \ \ \ \ \ \ \ \ \ \ \ \ \ \ \ \ \ \ \ \ \ \ \ \ \ \ \ \ \ \ \ \ \ \ \ \end{equation} \begin{equation}\label{gossipknow} \mbox{Knowledge Set Update:} \ \ \ \ \ \ \ \ \ \ \begin{array}{llll} & \ \ \ \mathcal{K}_i(t^{ij}_1(+)) = \{ \hat{s}_i(t^{ij}_1(+)) \} \\ &  \hat{s}_i(t^{ij}_1(+)) = \frac{1}{2} \big( \hat{s}_i(t^{ij}_1) + \hat{s}_j(t^{ij}_0) \big) 
\end{array} \ \ \ \ \ \ \ \ \ \ \ \ \ \ \ \ \ \ \ \ \ \ \ \ \ \ \ \ \ \ \end{equation} \begin{equation}\label{gossipinit1} \mbox{Initialization:} \ \ \ \ \ \ \ \ \ \ \ \ \ \ \ \ \ \ \ \ \ \mathcal{K}_i(0) = \{ \hat{s}_i(0) \} \ , \ \hat{s}_i(0) \ = s_i(0) \ , \ \forall \ i \in \mathcal{V} \ . \ \ \ \ \ \ \ \ \ \ \ \ \ \ \ \ \ \ \ \ \ \ \ \end{equation} \end{framed} We note that under the Gossip algorithm the only communication conditions proven to ensure average-consensus require instantaneous and bi-directional communication, thus implying \begin{equation}\label{gossipassume}\begin{array}{llll} & S^{ij}(t^{ij}_0, t^{ij}_1 ) \in C_{[0,t_1]} \Leftrightarrow S^{ji}(t^{ji}_0,t^{ji}_1 ) \in C_{[0,t_1]} \ , \\ & \ \ \ \ \ \ \ \ \ \ \ \ \ \ t^{ij}_0 = t^{ij}_1 = t^{ji}_0 = t^{ji}_1 \ , \end{array}\end{equation} for any signal $S^{ij}(t^{ij}_0,t^{ij}_1 ) \in C_{[0,t_1]}$. Under the assumption $(\ref{gossipassume})$, any $C_{[0 , t_1]}$ satisfying the I$\mathcal{V}$SC condition $(\ref{star1ga})$ will imply the Gossip algorithm obtains average-consensus at time $t= t_1(+)$, that is $(\ref{acdef})$ is satisfied at $t= t_1(+)$. We note that in some works (e.g. \cite{mo06}, \cite{bl05}) the Gossip algorithm is referred to as ``pairwise averaging''.

\subsubsection{Comparison Algorithm 2 (ARIS)}\label{sec:ris} The adaptation of the randomized information spreading algorithm proposed in \cite{SH08} that we will consider can be defined by the signal specification $(\ref{arissig})$ and knowledge set update $(\ref{arisknow})$ below. \begin{framed} $$ \mbox{ARIS Algorithm:}  $$ \begin{equation}\label{arissig}\mbox{Signal Specification:} \ \ \ \ \ \ \ \ \ \ \ \ \ \ \ \ \  S^{ij}(t^{ij}_0,t^{ij}_1 ) = \mathcal{K}_j(t^{ij}_0) \setminus \{s_i(0), i, n \} \ \ \ \ \ \ \ \ \ \ \ \ \ \ \ \ \ \ \ \ \ \ \ \ \ \ \ \ \ \ \ \ \ \ \ \ \end{equation}\begin{equation}\label{arisknow}\mbox{Knowledge Set Update:} \ \ \ \ \ \ \begin{array}{llll} & \mathcal{K}_i(t^{ij}_1(+)) = \{ k_i( t^{ij}_1(+) ) , \hat{s}_i(t^{ij}_1(+)) , \\ & \ \ \ \ \ \ \ \ \ \ \ \ \ W^i (t^{ij}_1(+)) , \hat{w}_i( t^{ij}_1(+) ), s_i(0), i, n \}  \end{array}  \ \ \ \ \ \ \ \ \ \ \ \ \ \ \ \ \ \ \ \ \ \ \ \ \ \ \end{equation} \begin{equation}\label{gossipinit}  \mbox{Initialization:} \ \ \ \ \ \ \ \ \ \ \ \ \ \ \ \ \ \ \ \begin{array}{llll} & \mathcal{K}_i(0) = \{ k_i( 0 ) , \hat{s}_i(0) , W^i (0) , \hat{w}_i( 0), s_i(0), i, n \} \\ & k_i( 0 ) = 0 \ , \ \hat{s}_i(0) = s_i(0)/n \ , \ \hat{w}_i( 0 ) = e_i \ , \\ & W^i (0) \in \mathbb{R}^{d \times r} \ , \ W^i_{\ell q} (0) \ \sim \ \mbox{exp $\{s_{i \ell}(0)\}$} \ , \\ & \ \ \ \ \ \ \forall \ \ell = 1,2, \ldots, d \ , \ q = 1,2,\ldots,r \ . \end{array} \ \ \ \ \ \ \ \ \ \ \ \ \ \ \ \ \ \ \ \ \ \ \ \end{equation} \end{framed}

We clarify that the $(\ell q)^{th}$ element in the matrix $W^i(0)$ is an independent realization of a random variable from an exponential distribution with rate parameter $s_{i \ell}(0)$, this is why the elements of each initial consensus vector $s_i(0)$ are required to be positive valued under any version of the RIS algorithm. We next define the ARIS update for each term in the knowledge set $\mathcal{K}_i(t^{ij}_1(+))$, we will omit the time indices for convenience. \begin{framed} $$ \mbox{ARIS Knowledge Set Update Procedure:} $$ $$ \begin{array}{llll} & \hat{w}^1 = \left\{ \begin{array}{l l} \mathbf{1}_n - \delta [ \hat{w}_i + \hat{w}_j ] , & \ \mbox{if $k_j = k_i,$} \\ \mathbf{1}_n - \delta [ e_i + \hat{w}_j ] , & \ \mbox{if $k_j > k_i$,} \\ 0 , & \ \mbox{if $k_j < k_i$,} \end{array} \right. \end{array} $$ \begin{equation}\label{gossipupdate1} k_i(+) = \left\{ \begin{array}{l l} k_j, \ \ \ \ \ \ \mbox{if $k_j > k_i , \hat{w}^1 \neq \mathbf{1}_n$} \\ k_j + 1 , \ \mbox{if $k_j > k_i , \hat{w}^1 = \mathbf{1}_n$} \\ k_i , \ \ \ \ \ \ \mbox{if $k_j = k_i , \hat{w}^1 \neq \mathbf{1}_n$} \\ k_i + 1 , \ \mbox{if $k_j = k_i , \hat{w}^1 = \mathbf{1}_n$} \\ k_i , \ \ \ \ \ \ \mbox{if $k_j < k_i,$} \end{array} \right. \end{equation} \begin{equation}\label{gossipupdate1new} \hat{w}_i(+) = \left\{ \begin{array}{l l} \hat{w}^1 , & \ \mbox{if $\hat{w}^1 \notin \{\mathbf{1}_n , 0 \}$} \\ e_i , & \ \mbox{if $\hat{w}^1 = \mathbf{1}_n$} \\ \hat{w}_i , & \ \mbox{if $\hat{w}^1 = 0$,} \end{array} \right. \end{equation} $$ \begin{array}{llll} & \hat{W}^1_{\ell q} = \left\{ \begin{array}{l l} \ \sim \ \mbox{exp $\{s_{i \ell}(0)\}$} , \ \mbox{if $k_i(+) > k_i$} \\ \min \{ W^j_{\ell q} , W^i_{\ell q} \} , \ \ \mbox{if $k_i(+) = k_i = k_j$} \\ W^i_{\ell q} , \ \ \ \ \ \ \ \ \ \ \ \ \ \ \ \ \mbox{if $k_j < k_i$} \end{array} \right. \\ & \ \ \ \ \ \ \ell = 1, 2, \ldots,d, \ \ q = 1,2,\ldots,r, \end{array} $$ \begin{equation}\label{risupdate1a}\begin{array}{llll} & W^i_{\ell q}(+) = \left\{ \begin{array}{l l} \ \sim \ \mbox{exp $\{s_{i \ell}(0)\}$} , \ \mbox{if $k_j > k_i , \hat{w}^1 = \mathbf{1}_n$} \\ \min \{ W^j_{\ell q} , \hat{W}^1_{\ell q} \} , \ \ \mbox{if $k_j = k_i , \hat{w}^1 \neq \mathbf{1}_n$} \\ \hat{W}^1_{\ell q} , \ \ \ \ \ \ \ \ \ \ \ \ \ \ \ \ \mbox{if $k_j \leq k_i$} \end{array} \right. \\ & \ \ \ \ \ \ \ell = 1, 2, \ldots,d, \ \ q = 1,2, \ldots ,r, \end{array}\end{equation} $$ \begin{array}{llll} & \bar{w}_{i \ell} = \frac{1}{r} \sum_{q = 1}^r \min \{ W^j_{\ell q} (t^{ij}_0) , W^i_{\ell q} (t^{ij}_1) \} \ , \ \ell = 1, 2, \ldots,d, \\ & \underbar{w}_{i \ell} = \frac{1}{r} \sum_{q = 1}^r \min \{ W^j_{\ell q} (t^{ij}_0) , \hat{W}^1_{\ell q} \} \ \ \ , \ \ell = 1, 2, \ldots,d. \end{array} $$ \begin{equation}\label{rislast} \begin{array}{llll} & \hat{s}_i(+) = \left\{ \begin{array}{l l} \hat{s}_i , \ \ \ \ \ \ \ \ \ \ \ \ \mbox{if $k_j < k_i ,$} \\ \hat{s}_i , \ \ \ \ \ \ \ \ \ \ \ \ \mbox{if $k_j = k_i , \hat{w}^1 \neq \mathbf{1}_n$} \\ \frac{k_i \hat{s}_i + \bar{w}_i ^{-1} /n }{k_i+1} , \ \mbox{if $k_j = k_i , \hat{w}^1 = \mathbf{1}_n$} \\ \hat{s}_j , \ \ \ \ \ \ \ \ \ \ \ \ \mbox{if $k_j > k_i , \hat{w}^1 \neq \mathbf{1}_n$} \\ \frac{k_j \hat{s}_j + \underbar{w}_i ^{-1} /n }{k_j+1} , \ \mbox{if $k_j > k_i , \hat{w}^1 = \mathbf{1}_n$} \end{array} \right. \end{array} \end{equation}  \end{framed}

The integer $r$ is an RIS algorithm parameter that affects the convergence error of the algorithm; as $r$ increases the algorithm is expected to converge closer to the true average-consensus vector $\bar{s}(0)$. Based on the strong law of large numbers we make the following conjecture.
\begin{conj}\label{conject1} For any $C_{[0 , t_1]}$ satisfying the S$\mathcal{V}$SC condition $(\ref{scdef})$, both the RIS algorithm proposed in \cite{SH08} and the ARIS algorithm $(\ref{arissig})-(\ref{rislast})$ imply $(\ref{acdef})$ at time $t = t_1(+)$ almost surely in the limit as $r$ approaches infinity. \end{conj} 

As explained in Sec.$\ref{sec:rescost}$, the total resource cost of the RIS algorithm increases on the order $O(rd)$, and likewise, the total resource cost of the ARIS algorithm increases on the order $O(rd + n)$, thus it is not practical to assume $r$ can be made arbitrarily large as the Conjecture $\ref{conject1}$ requires. Note that due to the resource costs of the RIS and ARIS algorithms, the Conjecture $\ref{conject1}$ does not contradict Conjecture $\ref{conject0}$, even though both assume the same communication conditions.

An informal description of the ARIS algorithm is as follows: at $t =0$ each node generates $r$ random variables independently from an exponential distribution with rate parameter $s_{i\ell}(0)$ for each $\ell = 1, 2,\ldots,d.$ A local counter $k_i$ is set to zero and a local vector $\hat{w}_i$ is set to $e_i$. Upon reception of any signal, if the transmitted counter $k_j$ is less than the local counter $k_i$, then the signal is ignored, if $k_j = k_i$ then the receiving node records the minimum between the received $r$ random variables and local $r$ random variables for each $\ell = 1,2,\ldots,d$. The vector $\hat{w}_i$ maintains a record of the set of nodes that have a communication path with the node $i$ for any given counter value. Whenever $\hat{w}_i(t)$ is updated to $\mathbf{1}_n$, then the counter $k_i$ is increased by one, the consensus estimate $\hat{s}_i(t)$ is updated as a running average of the inverse of the currently recorded mean of the minimum random variables for each $\ell = 1,2, \ldots, d$, $\hat{w}_i(t)$ is reset to $e_i$, and a new set of $rd$ random variables are locally generated. If $k_j > k_i$ then the local counter $k_i$ is set to $k_j$, the consensus estimate $\hat{s}_i(t)$ is set to the received estimate $\hat{s}_j$, $\hat{w}_i(t)$ is reset to $e_i$ and updated by $\hat{w}_j$, a new set of $rd$ random variables are locally generated, and the node $i$ then records the minimum between the received $r$ random variables and the newly generated local $r$ random variables for each $\ell = 1,2,\ldots,d$. It is not difficult to see that the local counter of each node will approach infinity iff $C_{[0 , t_1]}$ satisfies the I$\mathcal{V}$SC condition $(\ref{star1ga})$, thus each element in the consensus estimate $\hat{s}_i(t)$ of each node $i \in \mathcal{V}$ will approach the inverse of the mean of the minimum of infinitely many random variables. By a well-known property of the minimum of a set of independently generated exponential random variables, together with the strong law of large numbers, we thus make the following conjecture. \begin{conj}\label{conject2} For any positive integer $r$, the ARIS algorithm $(\ref{arissig})-(\ref{rislast})$ implies $(\ref{acdef})$ at time $t = t_1(+)$ iff the communication sequence $C_{[0 , t_1]}$ satisfies the I$\mathcal{V}$SC condition $(\ref{star1ga})$. \end{conj}

The above conjecture is significant because, besides the flooding technique and the DA algorithm $(\ref{kda1})-(\ref{dainit1})$, there is no other consensus protocol in the literature that guarantees average-consensus for every communication sequence $C_{[0 , t_1]}$ that satisfies the I$\mathcal{V}$SC condition $(\ref{star1ga})$.

\subsection{Resource Cost Derivations}\label{sec:rescostapp} In this section we explain how the entries of Table I in Sec.$\ref{sec:rescost}$ are obtained. We will assume an arbitrary vector $v \in \mathbb{R}^m$ requires $2m$ scalars to be defined, and similarly, an unordered set of scalars $S$ with cardinality $|S|$ requires $|S|$ scalars to be defined, that is \begin{equation}\label{rescosteq}\begin{array}{llll} & \mbox{Resource cost of $v \in \mathbb{R}^m = 2m$,} \\ & \ \ \ \mbox{Resource cost of $S = |S|$.} \end{array}\end{equation} The rationale for $(\ref{rescosteq})$ is that each element in $v$ requires one scalar to define the location of the element, and one scalar to define the value of the element itself. The location of each scalar in $S$ is irrelevant because $S$ is an unordered set, thus $S$ can be defined using only $|S|$ scalars. An alternative to $(\ref{rescosteq})$ is to assume that an arbitrary vector $v \in \mathbb{R}^m$ requires $m$ scalars to be defined. Although this alternative will imply different entries for the \emph{exact} values in Table I, the \emph{order} of the storage and communication costs under each algorithm would remain the same. We adhere to $(\ref{rescosteq})$ for our resource cost computations due to its relative precision.

Next observe that under the BM, OH, and DDA algorithms each normal consensus estimate $\hat{v}_i(t)$ will contain only elements belonging to the set $\{0, \frac{1}{n}\}$, similarly each vector $\hat{w}_i(t)$ under the ARIS algorithm contains only elements belonging to the set $\{0, 1\}$. Under the BM, OH, and DDA algorithms we can thus define each $\hat{v}_i(t)$ based only on the set $\hat{V}_i(t)$, \begin{equation}\label{deflocset} \hat{V}_i(t) = \left \{ \begin{array}{ll} \{ - \ell \ : \ \hat{v}_{i \ell}(t) = 0 \} \ \ \mbox{if $n \hat{v}_i(t) ' \mathbf{1}_n \geq 1/2 $,} \\ \{ \ \ \ell \ : \ \hat{v}_{i \ell}(t) = \frac{1}{n} \} \ \ \mbox{if $n \hat{v}_i(t) ' \mathbf{1}_n < 1/2$.} \end{array} \right. \end{equation} A set $\hat{W}_i(t)$ can be defined analogous to $(\ref{deflocset})$ to specify $\hat{w}_i(t)$. Note that both $\hat{V}_i(t)$ and $\hat{W}_i(t)$ are unordered sets of scalars, and thus, assuming average-consensus has not been obtained, we have from $(\ref{deflocset})$, $(\ref{initsol})$, and $(\ref{gossipinit})$, \begin{equation}\label{cost1} 1 \leq \{ |\hat{V}_i(t)| , |\hat{W}_i(t)| \} \leq \lfloor n/2 \rfloor \ . \end{equation} Under the condition $(\ref{initsol})$, applying $(\ref{rescosteq})$ and $(\ref{cost1})$ to the algorithm update equations $(\ref{bm1})$, $(\ref{bm2})$, $(\ref{aaasdasig5})$, $(\ref{aaais2})$, $(\ref{dda1})$, $(\ref{dda2})$, then yield the respective upper and lower resource costs for the BM, OH, and DDA algorithms stated in Table I in Sec.$\ref{sec:rescost}$. Likewise, applying $(\ref{rescosteq})$ and $(\ref{cost1})$ to $(\ref{arissig})$-$(\ref{risupdate1a})$ yield the respective upper and lower resource costs for the ARIS algorithm. As detailed in $(\ref{gossipsig})-(\ref{gossipinit1})$, the Gossip algorithm proposed in \cite{SH06} requires only the local consensus estimate $\hat{s}_i(t)$ to be communicated and stored at each node, thus the communication and storage cost are both fixed at $2d$ under this algorithm.

We next observe that under the DA algorithm $(\ref{kda1})-(\ref{dainit1})$, if a normal consensus estimate $\hat{v}_i(t)$ contains any elements equal to zero then these elements may be omitted from the signal and knowledge set. Given the condition $(\ref{initsol})$ it follows from $(\ref{rescosteq})$ that the minimum number of scalars required to define $\hat{v}_i(t)$ under the DA algorithm is $2$, while the maximum number of scalars required to define $\hat{v}_i(t)$ is $2n$. Together with these upper and lower limits on the resource cost of $\hat{v}_i(t)$, applying $(\ref{rescosteq})$ to $(\ref{kda1})$ and $(\ref{kda2})$ then yields the upper and lower resource costs for the DA algorithm stated in Table I in Sec.$\ref{sec:rescost}$.

\subsection{Example $\mbox{S$\mathcal{V}$SC}$ Sequence}\label{sec:examp}

The sequence defined in $(\ref{double})$ below is an $\mbox{S$\mathcal{V}$SC}$ sequence that implies the DA, DDA, and BM algorithm all obtain average-consensus at the same instant. \begin{equation}\label{double} \begin{array}{llll} & C_{[0,4n-5]} = \{ S_1, S_2, \ldots, S_{2(n-1)} \} \ , \\ & S_i = S^{i+1,i}( 2(i-1),2i -1) \ , \ i = 1,2, \ldots, n-1, \\ & S_n = S^{1,n}(2(n-1),2n-1) \ , \\ & S_i = S^{i-n+1,i-n}(2(i-1),2i-1) \ , \ i = n+1 ,n+2, \ldots, 2n-2 \ . \end{array} \end{equation} It is a ``unit-delay'' sequence because each signal sent at time $t$ is received at time $t+1$, and if a signal is received at time $t+1$ then the next signal is sent at time $t+2$. Together with Theorem $\ref{bmcor}$, the example $(\ref{double})$ implies that the DA and DDA algorithms possess the weakest possible necessary conditions for average-consensus that any algorithm can have. In contrast, the OH algorithm does not achieve average-consensus under $(\ref{double})$.

\end{document}